\documentclass[10pt,twoside]{article}
\pdfoutput=1

\usepackage[square,sort&compress,numbers]{natbib}

\usepackage[english]{babel}
\usepackage{amssymb,amsmath,amsfonts}
\numberwithin{equation}{section}
\usepackage{hyperref}
\usepackage{graphicx}
\usepackage{slashed}
\usepackage{amsthm}
\usepackage[capitalise]{cleveref}

\usepackage{xcolor}
\definecolor{MyBlue}{cmyk}{1,0.13,0,0.63}
\definecolor{MyGreen}{cmyk}{0.91,0,0.88,0.52}
\newcommand{\mylinkcolor}{MyBlue}
\newcommand{\mycitecolor}{MyGreen}
\newcommand{\myurlcolor}{webbrown}

\theoremstyle{plain}
\newtheorem{thm}{Theorem}[section]
\newtheorem{lem}[thm]{Lemma}
\newtheorem{prop}[thm]{Proposition}

\theoremstyle{definition}
\newtheorem{defn}[thm]{Definition}
\newtheorem{remark}[thm]{Remark}
\newtheorem{example}[thm]{Example}

\renewcommand{\eqref}[1]{\labelcref{#1}}
\crefname{thm}{Theorem}{Theorems}
\crefname{lem}{Lemma}{Lemmas}
\crefname{prop}{Proposition}{Propositions}
\crefname{coro}{Corollary}{Corollaries}
\crefname{defn}{Definition}{Definitions}
\crefname{example}{Example}{Examples}
\crefname{remark}{Remark}{Remarks}
\hypersetup{%
  bookmarksnumbered=true,bookmarksopen=false,%
  plainpages=false,
  linktocpage=true,%
  colorlinks=true,breaklinks=true,%
  linkcolor=\mylinkcolor,citecolor=\mycitecolor,urlcolor=\myurlcolor,%
  pdfpagelayout=OneColumn,%
  pageanchor=true,%
}

\makeatletter
\renewcommand\section{\@startsection {section}{1}{\z@}%
                                   {-3.5ex \@plus -1ex \@minus -.2ex}%
                                   {2.3ex \@plus.2ex}%
                                   {\normalfont
                                   \large\bfseries}}
\renewcommand\subsection{\@startsection{subsection}{2}{\z@}%
                                     {-3.25ex\@plus -1ex \@minus -.2ex}%
                                     {1.5ex \@plus .2ex}%
                                     {\normalfont
                                     \normalsize\bfseries\itshape}}
\renewcommand\subsubsection{\@startsection{subsubsection}{3}{\z@}%
                                     {-3.25ex\@plus -1ex \@minus -.2ex}%
                                     {1.5ex \@plus .2ex}%
                                     {\normalfont\normalsize\itshape}}
\makeatother

\allowdisplaybreaks

\newcommand{\sD}{\slashed{D}}
\newcommand{\1}{\mathbb{I}}
\newcommand{\I}{\mathbb{I}}
\newcommand{\R}{\mathbb{R}}
\newcommand{\C}{\mathbb{C}}
\newcommand{\Z}{\mathbb{Z}}
\newcommand{\qH}{\mathbb{H}}
\newcommand{\A}{\mathcal{A}}
\newcommand{\mA}{\mathfrak{A}}
\newcommand{\mB}{\mathfrak{B}}
\renewcommand{\L}{\mathcal{L}}
\newcommand{\mH}{\mathcal{H}}
\newcommand{\F}{\mathcal{F}}
\newcommand{\lu}{\mathfrak{u}}
\newcommand{\su}{\mathfrak{su}}
\newcommand{\g}{\mathfrak{g}}
\newcommand{\h}{\mathfrak{h}}
\newcommand{\G}{\mathcal{G}}
\newcommand{\Pf}{\textnormal{Pf}}
\newcommand{\Id}{\textnormal{Id}}
\newcommand{\Tr}{\textnormal{Tr}}
\newcommand{\End}{\textnormal{End}}
\newcommand{\Inn}{\textnormal{Inn}}
\newcommand{\Aut}{\textnormal{Aut}}
\newcommand{\Out}{\textnormal{Out}}
\newcommand{\Diff}{\textnormal{Diff}}
\newcommand{\Ker}{\textnormal{Ker}}
\renewcommand{\Im}{\textnormal{Im}}
\renewcommand{\bar}{\overline}
\newcommand{\Sub}[1]{_{\scriptscriptstyle#1}}
\newcommand{\til}[1]{\widetilde{#1}}
\newcommand{\GeV}{\ensuremath\text{GeV}}
\newcommand{\eV}{\ensuremath\text{eV}}
\DeclareMathOperator{\Ad}{\textnormal{Ad}}
\DeclareMathOperator{\ad}{\textnormal{ad}}

\newcommand{\mattwo}[4]{
  \left(\!\begin{array}{c@{~}c}#1&#2\\#3&#4\\\end{array}\!\right)
}
\newcommand{\vectwo}[2]{
  \left(\!\begin{array}{c}#1\\#2\\\end{array}\!\right)
}
\newcommand{\matthree}[3]{
  \left(\!\begin{array}{c@{~}c@{~}c}#1\\#2\\#3\\\end{array}\!\right)
}
\newcommand{\matfour}[4]{
  \left(\!\begin{array}{c@{~}c@{~}c@{~}c}#1\\#2\\#3\\#4\\\end{array}\!\right)
}

\newcommand{\theTitle}{Particle Physics from Almost-Commutative Spacetimes}
\newcommand{\theAuthors}{Koen van den Dungen and Walter D. van Suijlekom}

\hypersetup{pdftitle={\theTitle},pdfauthor={\theAuthors}}

\date{\today}

\begin{document}

\markboth{\theAuthors}{\theTitle}

\title{\theTitle}

\author{Koen van den Dungen\thanks{{\it Current address}: Mathematical Sciences Institute, Australian National University, Canberra, ACT 0200, Australia}  
~and Walter D. van Suijlekom\\[4mm]
{\normalsize Institute for Mathematics, Astrophysics and Particle Physics}\\ 
{\normalsize Radboud University Nijmegen}\\
{\normalsize Heyendaalseweg 135, 6525 AJ Nijmegen, The Netherlands}\\[4mm]
{\normalsize \texttt{koen.vandendungen@anu.edu.au}}\\
{\normalsize \texttt{waltervs@math.ru.nl}}
}

\maketitle

\begin{abstract}
Our aim in this review article is to present the applications of Connes' noncommutative geometry to elementary particle physics. Whereas the existing literature is mostly focused on a mathematical audience, in this article we introduce the ideas and concepts from noncommutative geometry using physicists' terminology, gearing towards the predictions that can be derived from the noncommutative description. Focusing on a light package of noncommutative geometry (so-called `almost-commutative manifolds'), we shall introduce in steps: electrodynamics, the electroweak model, culminating in the full Standard Model. We hope that our approach helps in understanding the role noncommutative geometry could play in describing particle physics models, eventually unifying them with Einstein's (geometrical) theory of gravity.
\end{abstract}

\newpage
\section{Introduction}
\label{chap:introduction}

One of the outstanding quests in modern theoretical physics is the unification of the four fundamental forces. There have been several theories around that partly fulfill this goal, all succeeding in some (different) aspects of such a theory. We will give an introduction to one of them, namely noncommutative geometry. It is a bottom-up approach in that it unifies the well-established Standard Model of high-energy physics with Einstein's general theory of relativity, thus not starting with extra dimensions, loops or strings. All this fits nicely in a mathematical framework, which was established by Connes in the 1980s \cite{Connes94}. 

Of course, there is a price that one has to pay for having such a rigid mathematical basis: at present the unification has been obtained only at the classical level. The main reason for this can actually be found in any gauge theory (such as Yang-Mills theory) as well: its quantization is still waiting for a sound mathematical description. The noncommutative geometrical description of (classical) Yang-Mills theories -- or the Standard Model, for that matter -- minimally coupled to gravity encounters the same trouble in formulating the corresponding quantum theory, in addition troubled by the quantization of gravity. It needs no stressing that this situation needs to be improved (though some progress has been made recently, see the Outlook), and it is our hope that this review article strengthens the dialogue with for instance string theory or quantum gravity. Intriguingly, noncommutative spacetimes naturally appear in the context of both of these theories. In string theory, this started with the work of Seiberg and Witten \cite{SW99}, in loop quantum gravity the quantized area operator is a manifestation of an underlying quantum geometry (cf. \cite{Rov04} and references therein). This has lead to a fruitful acre where ideas from the fields involved are combined: noncommutative geometry and string theory already in \cite{CDS98}, see the recent account in \cite{Kar10} and references therein; noncommutative geometry and loop quantum gravity in \cite{AG06,AG07,AGN08} and more recently in \cite{MZ08}.

Even though the noncommutative description of the Standard Model \cite{CCM07} does not require the introduction of extra spacetime dimensions, its construction is very much like the original Kaluza-Klein theories \cite{Kaluza21,Klein26}. In fact, one starts with a product 
$$
M \times F
$$
of ordinary four-dimensional spacetime $M$ with an {\it internal space} $F$ which is to describe the gauge content of the theory. Of course, spacetime itself still describes the gravitational part. The main difference with Kaluza-Klein theories is that the additional space is a discrete (zero dimensional) space whose structure is described by a (potentially) {\it noncommutative algebra}, the idea essentially dating back to \cite{CL91}. This is very much like the description of spacetime $M$ by its coordinate functions as usual in General Relativity, which form an algebra under pointwise multiplication:
$$
(x^\mu x^\nu) (p) = x^\mu(p) x^\nu(p) .
$$
Such commutative relations are secretly used in any physics textbook. However, for a discrete space, there are not many coordinates so we propose to describe $F$ by matrices instead, yielding a much richer internal (algebraic) structure. Multiplication is given by ordinary matrix multiplication:
$$
(AB)_{il} =  \sum_j A_{ij}B_{jl} .
$$ 
The corresponding {\it matrix algebra} of coordinates on $F$ is typically $M_N(\C)$ or direct sums thereof. It turns out (and we will explain that below) that one can also describe a metric on $F$ in terms of algebraic data, so that we can fully describe the geometrical structure of $M \times F$. This type of noncommutative manifolds are called {\it almost-commutative (AC) manifolds}, in contrast with the more serious noncommutative spaces such as the Moyal plane appearing in, say, Seiberg-Witten theory and for which (in contrast with the above)
$$
[x^\mu,x^\nu]=i \hbar.
$$ 
We stress that this is not the type of noncommutativity that we are dealing with here. Nevertheless, also such spaces fit in the framework of noncommutative geometry, see for instance \cite{CL01,CD02} for the compact and \cite{GGISV03,GI05} for the noncompact case. 

In this review article, we will give several examples of almost-commutative manifolds of interest in physics. In fact, adopting the chronological order of ordinary high-energy physics textbooks, we will {\it derive} electrodynamics, the Glashow-Weinberg-Salam electroweak theory, and the full Standard Model (including Higgs mechanism) from noncommutative spaces. 

Let us briefly sketch how this goes, at the same time giving an overview of the present article. 
Given an AC manifold $M \times F$, one first studies its symmetries: it turns out that the group of diffeomorphisms ({\it i.e.} general coordinate transformations) generalized to such noncommutative spaces combines ordinary diffeomorphisms of $M$ with gauge symmetries \cite{Connes96}. In other words, for the above three examples we obtain a combination of general coordinate transformations on $M$ with the respective groups $U(1)$, $U(1) \times SU(2)$ and finally $U(1) \times SU(2) \times SU(3)$, appearing as unitaries in the corresponding matrix algebras. This is the first hint at the aforementioned unification of gauge theories with gravity. All this is described at length in \cref{chap:ACG} below. 

The next step is to construct a Lagrangian from the geometry of $M \times F$. This is accomplished by the {\it spectral action principle} \cite{CC96,CC97}: it is a simple counting of the eigenvalues of a Dirac operator on $M \times F$ which are lower than a cutoff $\Lambda$. This is discussed in \cref{chap:spectral}, where we derive local formulas (integrals of Lagrangians) for the spectral action using heat kernel methods (cf. \cite{Vas03}). The fermionic action is given as usual by an inner product. The Lagrangians that one obtains in this way for the above examples are the right ones, and in addition minimally coupled to gravity. This is unification with gravity of electrodynamics, the Glashow-Weinberg-Salam electroweak theory, and finally the full Standard Model. This is presented in the respective \cref{chap:ex_ED,chap:ex_GWS,chap:ex_SM}. 

We study conformal invariance of the spectral action in \cref{chap:conf}, with particular emphasis on the Higgs mechanism coupled to the gravitational background. This has already found fruitful applications in cosmology \cite{MP09,MPT10,BFS10,NOS10,Sak10}.

Finally, in \cref{chap:RGE} we present some predictions that can be derived from the noncommutative description of the Standard Model (based on \cite{CCM07}). In fact, the Lagrangian derived through the spectral action principle from the relevant noncommutative space is not just the Standard Model Lagrangian, but it implies that there are {\it relations between some of the Standard Model couplings and masses}. This allows for a postdiction of (an upper bound on) the top quark mass, and a prediction of the Higgs mass.

We end by presenting an Outlook on some future developments. 

\bigskip

The main source to the noncommutative description of the Standard Model is \cite{CCM07} or, for the more mathematically inclined reader, the book \cite{CM07}, on which our discussion is heavily based. However, we attempt to make this material more accessible to theoretical physicists by adopting a light version of noncommutative geometry and describing separately the electromagnetic and -weak theory. The case of electrodynamics was first formulated in noncommutative terms only quite recently in \cite{DS11}. There are also some shorter companions on the applications of noncommutative geometry to high-energy physics such as \cite{JKSS07-ncsm} and \cite{CC10,CC11}. 

For more general treatments on noncommutative geometry, presented in its full mathematical glory, we refer to \cite{Landi97,GVF01,Var06,Kha09} or the original \cite{Connes94}.

\newpage
\section{Almost-Commutative Manifolds and Gauge Theories}
\label{chap:ACG}

\subsection{Spin manifolds in noncommutative geometry}
\label{sec:NCG_motiv}

One of the central ideas in Connes' noncommutative geometry \cite{Connes94} is to characterize ordinary Riemannian manifolds by a spectral data. By generalizing this spectral data in a suitable way, one arrives at the notion of a {\it noncommutative manifold}. Within this generalization, we will focus on the special case of what is called an {\it almost-commutative manifold}. This special case is given by the `noncommutative-geometric product' $M\times F$ of a Riemannian manifold $M$ and a \emph{finite space} $F$ as described in the next section. 

Before we are ready to provide a description of an almost-commutative manifold, let us first focus on a reformulation of the description of an ordinary Riemannian manifold. We restrict ourselves to the case of a compact $4$-dimensional Riemannian spin manifold $M$. The restriction on the dimension is just for notational simplicity and by no means essential; the $d$-dimensional case can readily be obtained. 
Our aim in this section is to illustrate how one replaces the usual topological and geometric description of $M$ by a spectral data in terms of operators on a Hilbert space. In the next section, we will see that this spectral description of manifolds also perfectly lends itself for the description of so-called finite spaces, and subsequently of almost-commutative manifolds. 

The first step in noncommutative geometry is to shift our focus from the manifold $M$ towards the coordinate functions on $M$. We will thus consider the set of smooth (infinitely differentiable) functions $C^\infty(M)$. These functions form an algebra under pointwise multiplication. As said, we restrict our attention to a spin manifold $M$, so that we may also consider the spinor bundle $S\rightarrow M$, and spinor fields on $M$ are then given by (smooth) sections $\psi\in\Gamma(M,S)$. We will consider the Hilbert space $\mH = L^2(M,S)$ of square-integrable spinors on $M$. The algebra $\A=C^\infty(M)$ acts on $\mH$ as multiplication operators by pointwise scalar multiplication $(f\psi)(x) := f(x)\psi(x)$ for a function $f$ and a spinor field $\psi$. 

Using the spin (Levi--Civita) connection $\nabla^S$ on the spinor bundle $S$, we can describe the \emph{Dirac operator} $\sD$. In local coordinates, $\sD$ is given by $-i\gamma^\mu\nabla^S_\mu$, in terms of the Dirac gamma matrices $\gamma^\mu$. The Dirac operator acts as a first-order differential operator on the spinors $\psi\in\Gamma(M,S)$. The spin connection satisfies the Leibniz rule 
\begin{align*}
\nabla^S_\mu(f \psi) = f \nabla^S_\mu(\psi) + (\partial_\mu f) \psi ,
\end{align*}
for all functions $f\in C^\infty(M)$ and all spinors $\psi\in\Gamma(S)$. For the commutator $[\sD,f]$ acting as an operator on $\psi$, we then calculate
\begin{align}
\label{eq:canon_Dirac_comm}
[\sD,f] \psi &= -i\gamma^\mu\nabla^S_\mu(f\psi) + i f \gamma^\mu\nabla^S_\mu \psi = -i \gamma^\mu (\partial_\mu f) \psi .
\end{align}
Therefore we see that the action of $[\sD,f]$ is simply given by multiplication with $-i \gamma^\mu (\partial_\mu f)$. In particular, this means that although $\sD$ is an unbounded operator, the commutator $[\sD,f]$ is always a bounded operator for any smooth function $f\in C^\infty(M)$. In fact, this bounded operator is given by the Clifford representation of the $1$-form $df = dx^\mu \partial_\mu f$. Let us gather all data described here so far into the following definition.

\begin{defn}
\label{defn:canon_triple}
Let $M$ be a compact $4$-dimensional Riemannian spin manifold. The \emph{canonical triple} is given by the set of data
\begin{align*}
\left(C^\infty(M),L^2(M,S),\sD\right) .
\end{align*}
Furthermore, we also have a $\Z_2$-\emph{grading} $\gamma_5 := \gamma^1\gamma^2\gamma^3\gamma^4$, and an antilinear isomorphism $J_M$, which is the \emph{charge conjugation} operator on the spinors. 
\end{defn}

The operator $\gamma_5$ is a $\Z_2$-grading, which means that ${\gamma_5}^2 = 1$, $\gamma_5^*=\gamma_5$, and $\gamma_5$ decomposes the Hilbert space $L^2(M,S) = L^2(M,S)^+\oplus L^2(M,S)^-$ into a positive and a negative eigenspace: chirality. With respect to this grading, the Dirac operator is an \emph{odd} operator, which means that $\gamma_5\sD = -\sD\gamma_5$. In other words, we have $\sD\colon L^2(M,S)^\pm\rightarrow L^2(M,S)^\mp$. For the charge conjugation operator $J_M$, one obtains the relations ${J_M}^2 = -1$, $J_M\sD = \sD J_M$ and $J_M\gamma_5 = \gamma_5J_M$. 

It turns out that a compact Riemannian spin manifold $M$ can be fully described by this canonical triple \cite{Connes96,Connes08}. The canonical triple is the motivating example for the definition of so-called spectral triples (see \cref{remark:spectral_triple}), and as such it lies at the foundation of noncommutative geometry. 

\subsubsection{Geodesic distance}

The claim that a Riemannian spin manifold $M$ is fully described by the canonical triple suggests that we are able to recover the \emph{Riemannian geometry} of $M$ from the data given with the canonical triple. We shall illustrate how this works by considering the notion of distance. The usual geodesic distance between two points $x$ and $y$ in $M$ is given by
\begin{align*}
d_g(x,y) = \inf_\gamma \int_\gamma ds = \inf_\gamma \int_0^1 \sqrt{g_{\mu\nu}\dot\gamma^\mu(t)\dot\gamma^\nu(t)} dt ,
\end{align*}
where $ds^2 = g_{\mu\nu}dx^\mu dx^\nu$ and the infimum is taken over all paths $\gamma$ from $x$ to $y$, with parametrizations such that $\gamma(0)=x$ and $\gamma(1)=y$. Let us confront this formula with the following distance formula in terms of the data in the canonical triple:
\begin{align}
\label{eq:distance_D}
d_\sD(x,y) = \sup \left\{ |f(x) - f(y)| \colon f\in C^\infty(M), \|[\sD,f]\|\leq1 \right\} .
\end{align}
Because we take a supremum over functions that locally have `slope' less than or equal to 1, we can indeed measure the Riemannian distance between two points in this dual manner. Full details can be found in e.g.\ \cite[Prop.~1.119]{CM07}. Thus, the geometric structure of $M$ is captured into the Dirac operator $\sD$. Later on, we will use \eqref{eq:distance_D} to generalize the geodesic distance to a notion of distance on almost-commutative manifolds. 
It has been noted in \cite{Rieffel99,DM10} that this distance formula, in the case of locally compact complete manifolds, is in fact a reformulation of the Wasserstein distance in the theory of optimal transport.

\begin{remark}
\label{remark:Lorentz}
We stress that the spin manifolds that can be captured in this way are Riemannian and not pseudo-Riemannian. In particular, this leaves out Lorentzian manifolds, which are of particular interest in physics. The reason for this is that dealing with Dirac operators on pseudo-Riemannian manifolds is technically difficult, since one loses the property of these being elliptic operators. Some progress in this direction has been obtained in \cite{Haw96,Str01,Mor02b,Sui04,PV04,PS06,Franco10,vdDPR12} though this program is far from being completed. 
\end{remark}

\subsection{Almost-commutative manifolds}
\label{sec:acm}

In the previous section we have provided an alternative algebraic description of a spin manifold $M$ in terms of the canonical triple $\left(C^\infty(M),L^2(M,S),\sD\right)$, following Connes \cite{Connes94}. In this section we will define the notion of an almost-commutative manifold as a generalization of spin manifolds. 

Such manifolds already appeared in the work of Connes and Lott \cite{CL91} and around the same time in a series of papers by Dubois-Violette, Kerner and Madore \cite{D-VKM1,D-VKM2,D-VKM3,D-VKM4}, who studied the noncommutative differential geometry for the algebra of functions tensored with a matrix algebra, and its relevance to the description of gauge and Higgs fields. 
Almost-commutative manifolds were later used by Chamseddine, Connes and Marcolli \cite{CCM07} to geometrically describe Yang-Mills theories and the Standard Model of elementary particles. The name almost-commutative manifolds was coined in \cite{Class_IrrACG_I}, their classification started in \cite{Kra97} (see also the more recent \cite{JS08,JKSS07-ncsm}). 

The general idea is to take the `product' $M\times F$ of the spin manifold $M$ with some finite space $F$, which in general may be noncommutative. Whereas the canonical triple encodes the structure of the spacetime $M$, the finite space $F$ will be used to encode the `internal degrees of freedom' at each point in spacetime. The main goal of this section is to show how these internal degrees of freedom will lead to the description of a gauge theory on $M$. 

Analogously to our description of a spin manifold $M$, we will describe a (generally noncommutative) \emph{finite space} $F$ by a triple
\begin{align*}
F := \left( \A_F, \mH_F, D_F \right) .
\end{align*}
Here we now have a finite-dimensional complex Hilbert space $\mH_F$, say of dimension $N$. The algebra $\A_F$ is a (real or complex) matrix algebra, which acts on the Hilbert space via matrix multiplication. The operator $D_F$ is given by a complex $N\times N$-matrix acting on $\mH_F$, and $D_F$ is required to be hermitian.

\begin{example}
Consider the matrix algebra $\A_F=M_N(\C)$ of $N \times N$ matrices acting on itself by left matrix multiplication, {\it i.e.} $\mH_F = M_N(\C)$. The operator $D_F$ is simply a hermitian $N^2 \times N^2$ matrix, $N^2$ being the dimension of $\mH_F$.
\end{example}

Suppose that the Hilbert space $\mH_F$ is $\Z_2$-graded, i.e.\ there exists a grading operator $\gamma_F$ for which $\gamma_F^* = \gamma_F$ and ${\gamma_F}^2=1$. This grading operator decomposes the Hilbert space $\mH_F = \mH_F^+ \oplus \mH_F^-$ into its two eigenspaces, where $\mH_F^\pm := \{\psi\in\mH_F \mid \gamma_F\psi = \pm\psi \}$. If such a grading operator exists, we say that the finite space $F$ is \emph{even} if furthermore we have $[\gamma_F,a]=0$ for all $a\in\A_F$ and $\gamma_FD_F=-D_F\gamma_F$. In other words, for an even finite space we require that the elements of the algebra are even operators and that $D_F$ is an odd operator. Thus, in the above example, a grading $\gamma_F = 1$ would force $D_F$ to vanish.

Let $M$ be a compact $4$-dimensional spin manifold. We will now take the `noncom\-mutative-geometric product' of $M$ with some even finite space $F$ as described above. The resulting product space $M\times F$ is called an almost-commutative (spin) manifold, or AC-manifold, and provides a generalization of compact Riemannian spin manifolds to the mildly noncommutative world. 

\begin{defn} 
An \emph{even almost-commutative (spin) manifold}, or AC-manifold, is described by
\begin{align*}
M\times F := \left( C^\infty(M,\A_F), L^2(M,S)\otimes\mH_F, \sD\otimes\1 + \gamma_5\otimes D_F \right) ,
\end{align*}
together with a grading $\gamma = \gamma_5\otimes\gamma_F$. The operator $D := \sD\otimes\1 + \gamma_5\otimes D_F$ is called the \emph{Dirac operator} of the almost-commutative manifold. 
\end{defn}

The canonical triple describing a spin manifold $M$, is a special case of an almost-commutative manifold $M\times F$. Indeed, if we take for the finite space $F$ the simple choice $(\A_F, \mH_F, D_F) = (\C,\C,0)$, then the almost-commutative manifold $M\times F$ is identical to the canonical triple for $M$. This can be interpreted as taking the product of $M$ with a single point: the functions on the point constitute the algebra $\C$, the spinor fields on the points form $\C$ and the Dirac operator cannot be anything else but 0.

\begin{table}
\centering
\begin{tabular}{l|l|l}
Data & Spin manifold $M$ & Finite space $F$ \\
\hline
Algebra $\A$ & Coordinate functions & Internal structure \\
Hilbert space $\mH$ & Spinor fields & Particle content \\
Operator $D$ & Dirac operator $\sD$ & Yukawa mass matrix $D_F$ \\
\end{tabular}
\caption{Spin manifolds vs.\ finite spaces\label{table:data}}
\end{table}

Let us take a closer look at the data in the triple describing an almost-commu\-ta\-tive manifold $M\times F$. In order to understand how to interpret these data, consider \cref{table:data} for a comparison between a spin manifold $M$ and a finite space $F$. Both $M$ and $F$ are described by a triple, consisting of an algebra, a Hilbert space, and an operator on this Hilbert space. For the spin manifold $M$, the algebra is given by the coordinate functions on $M$. The algebra for the finite space $F$ is typically a matrix algebra, which might be noncommutative. This matrix algebra may be interpreted as describing the `internal structure' of each point in spacetime. For the almost-commutative manifold this combines to give matrix-valued coordinate functions, which describe not only spacetime, but also the internal structure of spacetime. 

The Hilbert spaces are used to describe fermionic particles. The Hilbert space $L^2(M,S)$ for the spin manifold $M$ makes sure that each fermion is described by a spinor field. In the finite space $F$, we shall interpret each basis element of the Hilbert space $\mH_F$ as describing a different fermionic particle. The Hilbert space $\mH_F$ thus encodes the fermionic particle content of the model. 

The properties of each fermion, and the interactions among the fermions, will be determined by the way in which the algebra $\A$ and the operator $D$ act on the fermions. The Dirac operator $\sD$ acts as a first-order differential operator on the spinors, and this will provide us with the kinetics of the fermions through the Dirac equation $D \psi = 0$. The finite Dirac operator $D_F$ contains the Yukawa masses of the fermions, and therefore $D_F$ will sometimes be referred to as the Yukawa operator. Lastly, the action of the finite algebra $\A_F$ on the fermions will determine their gauge interactions. We will consider this at length in what follows, but observe already that the gauge particle content is described by the internal degrees of freedom.

\subsubsection{Generalized distance on AC-manifolds}

In \eqref{eq:distance_D} we have found an alternative formula for the geodesic distance on a spin manifold $M$. We can straightforwardly generalize this formula to provide a notion of distance on almost-commutative manifolds, and define the generalized distance function by 
\begin{align}
\label{eq:distance_ACM}
d_{D}(x,y) = \sup \left\{ \|a(x) - a(y)\| \colon a\in\A, \|[D,a]\|\leq1 \right\} .
\end{align}
where $\| \cdot \|$ is the (matrix) norm in $\A_F$. 
This distance formula has been studied in more detail in for instance \cite{MW02,Martinetti06}.

\subsubsection{Charge conjugation}

For a spin manifold $M$, we have a charge conjugation operator $J_M$. This operator is an anti-unitary operator acting on the Hilbert space $\mH = L^2(M,S)$ of square-integrable spinors. We would like to have a similar notion of a conjugation operator $J_F$ for a finite space $F$. This operator $J_F$ will be called a \emph{real structure}, and is defined as follows.

\begin{defn}
\label{defn:real_structure}
An (even) finite space $F$ is called \emph{real} if there exists a \emph{real structure} $J_F$, i.e.\ an anti-unitary operator $J_F$ on $\mH_F$ such that $J_F^2 = \varepsilon$, $J_FD_F = \varepsilon'D_FJ_F$ and (if $F$ is even) $J_F\gamma_F = \varepsilon''\gamma_FJ_F$. The signs $\varepsilon$, $\varepsilon'$ and $\varepsilon''$ 
depend on 
the \emph{KO-dimension} $n$ modulo $8$ of the finite space, according to the following table.%
\footnote{Note that in particular these signs can not be independently chosen: in the odd case there is no grading and hence no sign $\varepsilon''$, and in the even case the sign $\varepsilon'$ is always equal to $1$.}
\begin{align*}
\begin{array}{c|cccccccc}
n & 0 & 1 & 2 & 3 & 4 & 5 & 6 & 7 \\
\hline
\varepsilon & 1 & 1 & -1 & -1 & -1 & -1 & 1 & 1 \\
\varepsilon' & 1 & -1 & 1 & 1 & 1 & -1 & 1 & 1 \\
\varepsilon'' & 1 & & -1 &  & 1 & & -1 & \\
\end{array}
\end{align*}
Moreover, the action of $\A_F$ satisfies the commutation rule 
\begin{align}
\label{eq:order0}
[a,b^0] = 0 \qquad\forall a,b\in\A_F ,
\end{align}
where we have defined the right action $b^0$ of $b$ by
\begin{align}
\label{eq:right_action}
b^0 := J_Fb^*J_F^* .
\end{align}
The operator $D_F$ satisfies the so-called \emph{order one condition}
\begin{align}
\label{eq:order1}
\big[[D_F,a],b^0\big] = 0 \qquad\forall a,b\in\A_F .
\end{align}
If a finite space $F$ is real, then the corresponding almost-commutative manifold $M\times F$ is automatically also real, because the operator $J := J_M \otimes J_F$ then determines a real structure of KO-dimension $n+4\mod 8$ for the almost-commutative manifold $M\times F$.
\end{defn}

\begin{example}
\label{ex:YM}
Let us provide a general example of an almost-commutative manifold, which later on will be shown to describe a Yang-Mills gauge theory. Consider the finite space $F\Sub{YM}$ given by the triple
\begin{align*}
F\Sub{YM} := \left(M_N(\C), M_N(\C), 0\right) .
\end{align*}
Both the algebra $\A_F$ and the Hilbert space $\mH_F$ are given by the complex $N\times N$-matrices. The action of $\A_F$ on $\mH_F$ is simply given by left matrix multiplication. The finite Dirac operator is taken to be zero. We endow this finite space with the trivial grading given by the identity matrix $\gamma_F := \1_N$ (note that the relation $\gamma_FD_F=-D_F\gamma_F$ is now only satisfied by virtue of $D_F$ being zero). The real structure $J_F$ is defined as taking the hermitian conjugate of a matrix $m\in \mH_F$, i.e.\ $J_Fm := m^*$. We then have that $J_F^* = J_F$. In \eqref{eq:right_action}, we have defined the \emph{right action} of an element $b\in\A_F$ by $b^0 := J_Fb^*J_F^*$. In this case, we see that 
\begin{align*}
b^0 m = J_Fb^*J_F m = J_Fb^* m^* = mb ,
\end{align*}
so the action of $b^0$ is indeed given by right matrix multiplication with $b$. One then readily checks that $J_F$ defines a real structure of KO-dimension $0$. We now define the \emph{Yang-Mills manifold} as the almost-commutative manifold
\begin{align*}
M\times F\Sub{YM} := \left( C^\infty(M,M_N(\C)), L^2(M,S)\otimes M_N(\C), \sD\otimes\1\right) ,
\end{align*}
with a grading $\gamma := \gamma_5 \otimes \1_N$ and a real structure $J := J_M \otimes J_F$ of KO-dimension $4$. Throughout this section, we will frequently return to this illustrative example. 
\end{example}

\begin{lem}
\label{lem:class_real}
For any real even finite space $F$, we can write with respect to the decomposition $\mH = \mH^+ \oplus \mH^-$:
\begin{align*}
\textnormal{KO-dimension }0\colon\quad J_F &= \mattwo{j_+}{0}{0}{j_-} C \quad \textnormal{for symmetric } j_\pm\in U(\mH^\pm) ;\\
\textnormal{KO-dimension }2\colon\quad J_F &= \mattwo{0}{j}{-j^T}{0} C \quad \textnormal{for } jj^* = j^*j = \I ;\\
\textnormal{KO-dimension }4\colon\quad J_F &= \mattwo{j_+}{0}{0}{j_-} C \quad \textnormal{for anti-symmetric } j_\pm\in U(\mH^\pm) ;\\
\textnormal{KO-dimension }6\colon\quad J_F &= \mattwo{0}{j}{j^T}{0} C \quad \textnormal{for } jj^* = j^*j = \I .
\end{align*}
\end{lem}
\begin{proof}
Let the operator $C$ denote complex conjugation. Then any anti-unitary operator $J_F$ can be written as $UC$, where $U$ is some unitary operator on $\mH_F$. We then have $J_F^* = CU^* = U^TC$, and $J_FJ_F^* = UU^* = \I$. 
The different possibilities for the choice of $J_F$ are characterized by the relations $J_F^2 = UCUC = U\bar U = \varepsilon$ and $J_F\gamma_F = \varepsilon''\gamma_FJ_F$. By inserting $\varepsilon,\varepsilon''=\pm1$ according to the KO-dimension, the exact form of $J_F$ can be directly computed by imposing these relations. 
\end{proof}

Note that in the above Lemma we have not yet used all aspects of the definition of the real structure $J_F$. There are still two commutation rules that are required to be satisfied, namely
\begin{align*}
[a,b^0] = 0 \qquad\forall a,b\in\A_F ,\\
\big[[D_F,a],b^0\big] = 0 \qquad\forall a,b\in\A_F ,
\end{align*}
where $b^0 := J_Fb^*J_F^*$ (see \cref{defn:real_structure}). Furthermore, we must have $J_FD_F = D_FJ_F$ for even KO-dimensions. We will not examine the precise implications of these commutation rules here, but one should be aware that these rules impose further restrictions on the operators $D_F$ and $J_F$. Later on in \cref{prop:zero_D_F}, we will use these restrictions to show that the two-point space does not simultaneously allow a real structure $J_F$ and a non-zero Dirac operator $D_F$. 

\begin{remark}
\label{remark:spectral_triple}
An almost-commutative manifold is a special case of a \emph{spectral triple}  (see e.g.\ \cite[Ch.~1, \S10]{CM07}). In general, a spectral triple $(\A, \mH, D)$ is given by an involutive unital algebra $\A$ represented faithfully as bounded operators on a Hilbert space $\mH$ and a selfadjoint (in general unbounded) operator $D$ with compact resolvent (i.e.\ $(1+D^2)^{-1/2}$ is a compact operator) such that all commutators $[D,a]$ are bounded for $a\in\A$. A spectral triple is \emph{even} if the Hilbert space $\mH$ is endowed with a $\Z_2$-grading $\gamma$ which commutes with any $a\in\A$ and anticommutes with $D$. A spectral triple has a \emph{real structure} of KO-dimension $n$ if there is an antilinear isomorphism $J\colon \mH\rightarrow\mH$ with $J^2 = \varepsilon$, $JD = \varepsilon' DJ$ and, if the spectral triple is even, $J\gamma = \varepsilon'' \gamma J$, where the signs are again determined by the table in \cref{defn:real_structure}.
\end{remark}

Even though many definitions and results that follow in this section equally well apply to spectral triples in general, we will only present them in terms of almost-commutative manifolds. The reason is that we want to put main emphasis on explicit results for almost-commutative manifolds, avoiding the more technical notion of spectral triple. For our purposes we shall therefore have no need for a description of spectral triples in general. 
So, from here on we shall consider a real even almost-commutative manifold $M\times F$, as described by the following data, using the notation as above:
\begin{itemize}
\item An algebra $\A = C^\infty(M,\A_F)$;
\item A Hilbert space $\mH = L^2(M,S) \otimes \mH_F$;
\item An operator $D = \sD\otimes\1 + \gamma_5\otimes D_F$ on $\mH$;
\item A $\Z_2$-grading $\gamma = \gamma_5\otimes\gamma_F$ on $\mH$;
\item A real structure $J = J_M\otimes J_F$ on $\mH$.
\end{itemize}

\subsection{Subgroups and subalgebras}

In this section, we shall have a closer look at the algebra $\A = C^\infty(M,\A_F)$, and especially at some of its subalgebras and subgroups. These subsets are presented here in preparation for the next section, in which we shall discuss the gauge group. 

\subsubsection{Commutative subalgebras}
\label{sec:subalgs}

We define a subalgebra of $\A$ by
\begin{align*}
\til\A_J &:= \big\{ a\in\A \mid aJ = Ja^* \big\} = \big\{ a\in\A \mid a^0 = a \big\} .
\end{align*} 
This definition is very similar to the definition of $\A_J$ in \cite[Prop.~3.3]{CCM07} (cf.\ \cite[Prop.~1.125]{CM07}), which is a {\it real} commutative subalgebra in the center of $\A$. We have provided a similar but different definition for $\til\A_J$, since this subalgebra will turn out to be very useful for the description of the gauge group in \cref{sec:gauge_group}. 

Let us assume that the algebra $\A$ is complex. We easily see that $aJ = Ja^*$ implies $a = Ja^*J^* = a^0$ and vice versa. Since we must have $[a,b^0] = 0$ for any $a,b \in\A$ (cf. \eqref{eq:order0}), we have $[a,b] = 0$ for any $a\in\A$ and $b\in\til\A_J$, so $\til\A_J$ is contained in the center of $\A$. The requirement $a = a^0$ is complex linear, and also implies that $a^* = (a^0)^* = (a^*)^0$, so we have $a^*\in\til\A_J$ for $a\in\til\A_J$. Finally, we check that for $a,b\in\til\A_J$, we find $(ab)^0 = b^0a^0 = ba = ab$, so $ab\in\til\A_J$. Therefore, $\til\A_J$ is an involutive commutative \emph{complex} subalgebra of the center of $\A$. 

If we consider a finite space $F$, the definition of $(\til\A_F)_{J_F}$ is exactly as given above, so it is the subalgebra of $\A_F$ determined by the relation $aJ_F=J_Fa^*$. For an almost-commutative manifold $M\times F$, we have the real structure $J=J_M\otimes J_F$. Since the effect of $J_M$ on a function on $M$ is simply complex conjugation, we obtain that the requirement $aJ=Ja^*$ must be satisfied pointwise, i.e.\ $a(x)J_F = J_Fa(x)^*$, for $a(x)\in\A_F$. This implies that for $a\in\til\A_J$ we obtain $a(x)\in(\til\A_F)_{J_F}$. Thus, for an almost-commutative manifold, the subalgebra $\til\A_J$ is given by 
\begin{align*}
\til\A_J &= C^\infty\big(M,(\til\A_F)_{J_F}\big) .
\end{align*}

\begin{example}
\label{ex:YM_subalg}
Let us return to the Yang-Mills manifold $M\times F\Sub{YM}$ of \cref{ex:YM}. We have already seen that the right action was given by $a^0m = ma$. If we consider the requirement $a^0 = a$, we see that this implies that $a$ must commute with all $N\times N$-matrices $m\in\mH_F$, so $a$ is contained in the center $\C\1_N$ of $M_N(\C)$. For the Yang-Mills manifold, we thus obtain that $\til\A_J \simeq C^\infty(M)\otimes\1_N$. 
\end{example}

\subsubsection{Unitary subgroups}

The \emph{unitary group} $U(\A)$ of a unital, involutive algebra $\A$ is defined by
\begin{align*}
U(\A) = \big\{ u\in\A \mid uu^* = u^*u = \1 \big\} .
\end{align*}
The conjugation on $\A = C^\infty(M,\A_F)$ is given by pointwise conjugation on $\A_F$. So, the requirement $uu^*=u^*u=\1$ must hold for each $x\in M$, which gives $u(x)u(x)^* = u(x)^*u(x) = \1$. Hence, $u\in U(\A) \Leftrightarrow u(x)\in U(\A_F)$, and the unitary group is given by $U(\A) = C^\infty(M,U(\A_F))$.

The Lie algebra of this unitary group is given by all anti-hermitian elements of the algebra:
\begin{align}
\label{eq:lu}
\lu(\A) := \{X\in \A \mid X^* = -X \} .
\end{align}
As for the unitary group, we now obtain $\lu(\A) = C^\infty(M,\lu(\A_F))$. 

For the finite-dimensional algebra $\A_F$, an element $a\in\A_F$ acts on the finite Hilbert space $\mH_F$ via matrix multiplication. Therefore, we can define the determinant $\det(a)$ of an element $a\in\A_F$ simply as the determinant of this matrix. We can then define the \emph{special unitary group} $SU(\A_F)$ by
\begin{align*}
SU(\A_F) = \big\{ u\in U(\A_F) \mid \det(u) = 1 \big\} .
\end{align*}
The Lie algebra of $SU(\A_F)$ consists of the \emph{traceless} anti-hermitian elements 
\begin{align*}
\su(A_F) = \{X\in \A_F \mid X^* = -X, \Tr(X)=0 \} .
\end{align*}
The condition $\det(u) = 1$ or $\Tr(X)=0$ is often referred to as the \emph{unimodularity condition}.

\subsubsection{The adjoint action} 

For a finite space $F := (\A_F, \mH_F, D_F, \gamma_F, J_F)$, the operator $J_F$ provides a right action of $a\in\A_F$ on $\mH_F$ by $a^0 = J_Fa^*J_F^*$, as in \eqref{eq:right_action}. Using this right action, we can define maps $\Ad\colon U(\A_F)\rightarrow\End(\mH_F)$ and $\ad\colon\lu(\A_F)\rightarrow\End(\mH_F)$ by
\begin{align*}
(\Ad u) \xi &:= u\xi u^* = u(u^*)^0 \xi ,\notag\\
(\ad A)\xi &:= A\xi - \xi A = (A-A^0)\xi ,
\end{align*}
for $\xi\in\mH_F$. By inserting $a^0 = J_Fa^*J_F^*$, we obtain
\begin{align*}
\Ad u &= uJuJ^* ,\\
\ad A &= A - JA^*J^* = A + JAJ^* ,\qquad\text{for } A^*=-A . \notag
\end{align*}
If we would replace $A = iB$, we could then define
\begin{align}
\label{eq:ad_J}
\ad B := -i\ad(iB)= B - JBJ^* ,\qquad\text{for } B^*=B .
\end{align}
The maps $\Ad u$ and $\ad A$ are often called the \emph{adjoint representations} of $u$ and $A$, respectively, on the Hilbert space $\mH_F$. 

Let us consider the adjoint action $\Ad u$ for an element $u\in U(\til\A_J)$ in the unitary group of the subalgebra $\til\A_J$. Since in this subalgebra we have $uJ=Ju^*$, we see that $\Ad u = uJuJ^* = Juu^*J^* = 1$. In other words, the group $U(\til\A_J)$ acts trivially via the adjoint representation. We obtain a similar result for the Lie algebra $\lu(\til\A_J)$. If we take a hermitian element $X=X^*\in i\,\lu(\til\A_J)$, we see that $\ad X = X - JXJ^* = X - X^* = 0$.

\subsection{Gauge symmetry}

\subsubsection{Diffeomorphisms and automorphisms}

For a spin manifold $M$, we have the group of diffeomorphisms $\Diff(M)$, which are smooth invertible maps from $M$ to $M$. A diffeomorphism is given by a coordinate transformation, and the algebraic object corresponding to this coordinate transformation is an automorphism (i.e.\ an invertible algebra homomorphism) from the algebra $C^\infty(M)$ to itself. Namely, for a diffeomorphism $\phi\colon M\rightarrow M$, we can define the automorphism $\alpha_\phi\colon f \mapsto f\circ\phi^{-1}$ for a function $f\in C^\infty(M)$. For an algebra $\A$, we denote by $\Aut(\A)$ the group of algebra automorphisms of $\A$. For the algebra of coordinate functions we then have $\Aut(C^\infty(M)) \simeq \Diff(M)$. Based on this isomorphism, we will define the group of diffeomorphisms of an almost-commutative manifold as
\begin{align*}
\Diff(M\times F) := \Aut\big(C^\infty(M,\A_F)\big) .
\end{align*}
A diffeomorphism $\phi\in\Diff(M)$ will also yield a diffeomorphism of $M\times F$, namely for $a\in C^\infty(M,\A_F)$ we again define the automorphism $\alpha_\phi\colon a \mapsto a\circ\phi^{-1}$. More explicitly, we thus have $\big(\alpha_\phi(a)\big)(x) := a\big(\phi^{-1}(x)\big)$. However, because of the `internal structure' of the almost-commutative manifold, given by the finite algebra $\A_F$, there are now more automorphisms than just the diffeomorphisms of $M$. For instance, for a function $u\in C^\infty(M,U(\A_F))$ which takes values in the unitary elements of $\A_F$, we can define the automorphism $\alpha_u\colon a\mapsto uau^*$, so $\big(\alpha_u(a)\big)(x) = u(x)a(x)u^*(x)$. In the mathematics literature, such automorphisms are called \emph{inner automorphisms}, and the group of inner automorphisms $\alpha_u\colon a\mapsto uau^*$ is denoted by $\Inn(\A)$. 

The group $\Inn(\A)$ is always a normal subgroup of $\Aut(\A)$, which can be shown as follows. For $\beta\in\Aut(A)$ and $\alpha_u\in\Inn(\A)$, we find that 
\begin{align*}
\beta\circ\alpha_u\circ\beta^{-1}(a) = \beta\big(u\beta^{-1}(a)u^*\big) = \beta(u)a\beta(u)^* = \alpha_{\beta(u)} .
\end{align*}
This means that we can define the \emph{outer automorphisms} by the quotient 
\begin{align*}
\Out(\A) := \Aut(\A) / \Inn(\A) . 
\end{align*}

An inner automorphism $\alpha_u$ is completely determined by the unitary element $u\in U(\A)$, but it is not uniquely determined by $u$. In other words, the map $\phi\colon U(\A) \rightarrow \Inn(\A)\colon u \mapsto \alpha_u$ is surjective, but it is not injective. The kernel is given by $\Ker(\phi) = \{ u\in U(\A) \mid uau^* = a,\; \forall a\in\A \}$. The relation $uau^* = a$ implies $ua=au$ for all $a\in\A$. Let $Z$ be the subgroup of $U(\A)$ that commutes with $\A$. We thus see that $\Ker(\phi) = Z$. In other words, the group of inner automorphisms is given by the quotient 
\begin{align}
\label{eq:Inn_A}
\Inn(\A) \simeq U(\A) / Z .
\end{align}

\subsubsection{Unitary transformations}

We would like to study the notion of `symmetry' for almost-commutative manifolds. Since the symmetry of an ordinary manifold $M$ is determined by its group of diffeomorphisms $\Diff(M)$, we might be inclined to define the symmetry group of an almost-commutative manifold as $\Diff(M\times F) := \Aut\big(C^\infty(M,\A_F)\big)$. However, it turns out that an almost-commutative manifold has an even richer symmetry, which we will now attempt to derive. 

Our starting point will be the notion of a \emph{unitary transformation} as defined below. The symmetry will then be revealed when it turns out that the bosonic and fermionic action functionals, as defined in \cref{sec:act_invar}, are invariant under these unitary transformations. We take our definition of unitary transformations from \cite[\S6.9]{Landi97}, but make a slight modification by incorporating the algebra isomorphism $\alpha$.

Let $M\times F$ be an almost-commutative manifold given by the triple $(\A, \mH, D)$. Let us now explicitly write the representation $\pi$ of the algebra $\A$ on the Hilbert space $\mH$, so the action of $a$ on $\mH$ is given by $\pi(a)$. 
\begin{defn}
A \emph{unitary transformation} of the almost-commutative manifold is given by a unitary operator $U\colon\mH\rightarrow\mH$. This unitary transformation yields the triple $(\A, \mH, UDU^*)$, where the action of the algebra $\A$ on the Hilbert space is now given by $U\pi(a)U^*$. If the almost-commutative manifold is even, the grading $\gamma$ transforms into $U\gamma U^*$, and if the almost-commutative manifold is real, the real structure $J$ transforms into $UJU^*$. 
\end{defn}

Let us consider two basic examples of such unitary transformations on a real even almost-commutative manifold $M\times F$. First, we consider the unitary operator $U=\pi(u)$ given by a unitary element of the algebra, so $u\in U(\A)$. Since the grading commutes with the algebra, we see that $\gamma$ is unaffected by this transformation. For the action of the algebra, we obtain that $U\pi(a)U^* = \pi(u)\pi(a)\pi(u^*) = \pi(uau^*) = \pi\circ\alpha_u(a)$ for the inner automorphism $\alpha_u$. 

Second, let us consider the \emph{adjoint action} of the unitary group $U(\A)$, so we take the unitary transformation $U = \Ad u = uJuJ^*$. The grading is again unaffected by the transformation, since $\pi(u)J\pi(u)J^*\gamma = (\epsilon'')^2\gamma \pi(u)J\pi(u)J^*$. Because $J\pi(u)J^*$ commutes with $\pi(a)$ (cf.\ \eqref{eq:order0}), we find that 
\begin{align*}
U\pi(a)U^* &= \pi(u)J\pi(u)J^*\pi(a)J\pi(u)^*J^*\pi(u)^* = \pi(u)\pi(a)J\pi(u)J^*J\pi(u)^*J^*\pi(u)^* \\
&= \pi(u)\pi(a)\pi(u)^* = \pi(uau^*) = \pi\circ\alpha_u(a) .
\end{align*}
Using $J^*=\epsilon J$, we see that 
\begin{align*}
J' &= UJU^* = \pi(u)J\pi(u)J^*JJ\pi(u)^*J^*\pi(u)^* = \pi(u)J\pi(u)J\pi(u)^*J^*\pi(u)^* \\
&= \pi(u)J\pi(u)\pi(u)^*J\pi(u)^*J^* = \pi(u)JJ\pi(u)^*J^* = \epsilon J^* = J . \end{align*}
Hence we find that the unitary transformation of the AC manifold yields the data $(\A, \mH, UDU^*, \gamma, J)$, where the action of the algebra is again given by $\pi\circ\alpha_u(a)$. This second case is especially interesting because we see that the unitary transformation has no effect on $J$. The group generated by all operators of the form $U=uJuJ^*$ characterizes equivalent AC-manifolds $(\A, \mH, UDU^*, \gamma, J)$, in which only the Dirac operator is affected by the unitary transformation. This group shall be interpreted as the gauge group, and this interpretation will later be justified by \cref{thm:gauge_acm}.

\subsubsection{The gauge group}
\label{sec:gauge_group}

\begin{defn}
\label{defn:gauge_group_NCG}
For a real almost-commutative manifold $M\times F$ given by the data $(\A, \mH, D, J)$, we define the \emph{gauge group} $\G(M\times F)$ as
\begin{align*}
\G(M\times F) := \left\{ U=uJuJ^* \mid u\in U(\A) \right\} .
\end{align*}
\end{defn}

In order to evaluate this gauge group in more detail, let us consider the map $\Ad\colon U(\A)\rightarrow \G(M\times F)$ given by $u\mapsto u(u^*)^0$. This map $\Ad$ is by definition surjective. And, indeed $\Ad$ is a group homomorphism, since the commutation relation $[a,JbJ^*]=0$ of \eqref{eq:order0} implies that $\Ad(b)\Ad(a) = bJbJ^*aJaJ^* = baJbaJ^* = \Ad(ba)$. This map has kernel $\Ker(\Ad) = \{u\in U(\A) \mid uJuJ^*=1\}$. The relation $uJuJ^*=1$ is equivalent to $uJ=Ju^*$, and we note that this is the defining relation of the commutative subalgebra $\til\A_J$ (see \cref{sec:subalgs}). Hence we have $\Ker(\Ad) = U(\til\A_J)$. We thus obtain the isomorphism
\begin{align}
\label{eq:gauge_group}
\G(M\times F) \simeq U(\A) / U(\til\A_J) .
\end{align}

From \cref{sec:subalgs} we know that $\til\A_J$ is a subalgebra of the center of $\A$. Hence the group $U(\til\A_J)$ of the previous proposition is contained in the subgroup $Z$ of $U(\A)$. From \eqref{eq:Inn_A,eq:gauge_group} we then see that in general, the gauge group $\G(M\times F)$ is larger than the group of inner automorphisms $\Inn(\A)$. Only if $U(\til\A_J)$ is \emph{equal} to $Z$, we have in fact $\Inn(\A) \simeq \G(M\times F)$. 

In the same way as we have defined the gauge group $\G(M\times F)$, we also obtain the groups $\G(M)$ and $\G(F)$. For the canonical triple describing the spin manifold $M$, we have seen that $\til\A_J=\A$, which means that the group $\G(M)$ is just the trivial group. For the finite space $F$, we obtain the \emph{finite} or \emph{local} gauge group $\G(F)$. Let us have a closer look at the structure of this local gauge group. We define two subsets of $\A_F$ by
\begin{subequations}
\begin{align}
H_F &:= U\big((\til\A_F)_{J_F}\big) ,\\
\label{eq:unitary_subLiealg_F}
\h_F &:= \lu\big((\til\A_F)_{J_F}\big) .
\end{align}
\end{subequations}
Note that the group $H_F$ is the finite counterpart of the group $U(\til\A_J)$ in \eqref{eq:gauge_group}. Let us evaluate the structure of this group in more detail. Since $(\til\A_F)_{J_F}$ is a subalgebra of $\A_F$, we know that $H_F$ is a subgroup of $U(\A_F)$, and in fact it is a Lie subgroup. Because $H_F$ is contained in the center of $\A_F$ (see \cref{sec:subalgs}), the condition $uvu^* = v$ for $v\in H_F$ and $u\in U(\A_F)$ is evidently satisfied, and hence $H_F$ is a normal subgroup.

The set $\h_F$ forms a real subspace of the real Lie algebra $\lu(\A_F)$. The elements of $\h_F$ are contained in the center of $\A_F$, so all commutators vanish: $[\h_F,\lu(\A_F)] = \{0\}$. In particular, this implies that $\h_F$ is a Lie algebra ideal of $\lu(\A_F)$. In fact, $\h_F$ is the Lie algebra of the normal subgroup $H_F$ of $U(\A_F)$. 

The local gauge group $\G(F)$ is given by the quotient $U(\A_F) / H_F$, which consists of the equivalence classes $[u]$ for $u\in U(\A_F)$, where $[uh]=[u]$ for all $h\in H_F$ determines the equivalence relation. We can write $u=e^X$ for $X\in\lu(\A_F)$ and $h=e^Y$ for $Y\in\h_F$. Thus we obtain an equivalence class of $X\in\lu(\A_F)$ by $e^{[X]} = [u] = [uh] = e^{[X+Y]}$ for all $Y\in\h_F$. We then recognize that the equivalence relation $[X+Y]=[X]$ defines the quotient $\g(F) := \lu(\A_F) / \h_F$, so the Lie algebra of $\G(F)$ is given by $\g(F)$. 

\begin{prop}
\label{prop:acm_gauge}
The gauge group $\G(M\times F)$ of an almost-commutative manifold is given by $C^\infty(M,\G(F))$, where $\G(F)=U(\A_F) / H_F$ is the local gauge group of the finite space.  
\end{prop}
\begin{proof}
We know that $H_F$ is a normal subgroup of $U(\A_F)$, so their quotient $\G(F)$ is well-defined. The subgroup $U(\til\A_J)$ equals $C^\infty(M,H_F)$. We thus obtain that $\G(M\times F) = U(\A) / U(\til\A_J) = C^\infty(M,U(\A_F)) / C^\infty(M,H_F) = C^\infty(M,\G(F))$. 
\end{proof}

\paragraph{Unimodularity}

Suppose that $\A_F$ is a complex algebra. The algebra has identity $\1$, and by complex linearity we see that $\C\1\subset(\til\A_F)_{J_F}$. By restricting to unitary elements, we then find that $U(1)$ is a subgroup of $H_F$. Because $H_F$ is commutative, $U(1)$ is then automatically a normal subgroup of $H_F$.

If, on the other hand, $\A_F$ is a real algebra, we can only say that $\R\1\subset(\til\A_F)_{J_F}$. By restricting to unitary (i.e.\ in this case orthogonal) elements, we then only obtain that $\{1,-1\}$ is a normal subgroup of $H_F$. 

\begin{prop}
\label{prop:unimod}
If $\A_F$ is a complex algebra, the gauge group is isomorphic to
\begin{align*}
\G(F) &\simeq SU(\A_F) / SH_F ,
\end{align*}
where we have defined $SH_F = \{ g\in H_F \mid \det g = 1 \}$. 
\end{prop}
\begin{proof}
An element of the quotient $\G(F) = U(\A_F)/H_F$ is given by the equivalence class $[u]$ of some element $u\in U(\A_F)$, subject to the equivalence relation $[u] = [uh]$ for all $h\in H_F$. Similarly, the quotient $SU(\A_F) / SH_F$ consists of the classes $[v]$ for $v\in SU(\A_F)$ with the equivalence relation $[v] = [vg]$ for all $g\in SH_F$. We first show that this quotient is well-defined, i.e.\ that $SH_F$ is a normal subgroup of $SU(\A_F)$. We thus need to check that $vgv^{-1}\in SH_F$ for all $v\in SU(\A_F)$ and $g\in SH_F$. We already know that $vgv^{-1}\in H_F$, because $H_F$ is a normal subgroup of $U(\A_F)$. We then also see that $\det(vgv^{-1}) = \det g = 1$, so $vgv^{-1}\in SH_F$, and the quotient $SU(\A_F) / SH_F$ is indeed well-defined. 

There exists a $\lambda_u\in U(1)$ such that ${\lambda_u}^{N} = \det u$, where $N$ is the dimension of the finite Hilbert space $\mH_F$. Since $U(1)$ is a subgroup of $U(\A_F)$ (because we assume $\A_F$ to be a complex algebra), we then see that ${\lambda_u}^{-1} u\in SU(\A_F)$. We can then define the group homomorphism $\varphi\colon \G(F) \rightarrow SU(\A_F) / SH_F$ by $\varphi([u]) = [{\lambda_u}^{-1} u]$. We need to check that $\varphi$ is well-defined, i.e.\ that $\varphi([u])$ is independent of the choice of the representative $u\in U(\A_F)$, as well as independent of the choice of $\lambda_u$. Suppose  we also have $\lambda_u'$ such that ${\lambda_u'}^N = \det u$. We then must have ${\lambda_u}^{-1}\lambda_u'\in \mu_N$, where $\mu_N$ is the multiplicative group of the $N$-th roots of unity. Since $U(1)$ is a subgroup of $H_F$, we see that $\mu_N$ is a subgroup of $SH_F$, so $[{\lambda_u}^{-1} u] = [{\lambda_u'}^{-1} u]$ and the image of $\varphi$ is indeed independent of the choice of $\lambda_u$. Next, for any $h\in H_F$, we also check that
\begin{align*}
\varphi([u]) = [{\lambda_u}^{-1} u] = [{\lambda_u}^{-1} u {\lambda_h}^{-1} h] = [({\lambda_u}{\lambda_h})^{-1} uh] = \varphi([uh]) ,
\end{align*}
where we have used that $g = {\lambda_h}^{-1} h \in SH_F$ (because $U(1)$ is a subgroup of $H_F$) and that $({\lambda_u}{\lambda_h})^N = \det uh$. 

Since $SU(\A_F) \subset U(\A_F)$, the homomorphism $\varphi$ is clearly surjective. Now suppose $\varphi([u_1])=\varphi([u_2])$ for some $u_1,u_2\in U(\A_F)$. This means that ${\lambda_{u_1}}^{-1} u_1 = {\lambda_{u_2}}^{-1} u_2 g$ for some $g\in SH_F$. We then obtain that $u_1 = u_2 h$ for an element $h = \lambda_{u_1} {\lambda_{u_2}}^{-1} g \in H_F$, so $[u_1]=[u_2]$ and $\varphi$ is also injective. Hence $\varphi$ is a group isomorphism. 
\end{proof}

The significance of \cref{prop:unimod} is that, in the case of a complex algebra with a complex representation, an equivalence class of the quotient $\G(F) = U(\A_F) / H_F$ can always be represented (though not uniquely) by an element of $SU(\A_F)$. In that sense, all elements of $\G(F)$ naturally satisfy the unimodularity condition. In the case of an algebra with a real representation, this is not true. For this reason one needs to impose the unimodularity condition on the inner fluctuations in the derivation of the standard model from noncommutative geometry (see \cref{sec:gauge_SM}). 

\begin{example}
\label{ex:YM_gauge_group}
Let us again consider the Yang-Mills manifold $M\times F\Sub{YM}$ of \cref{ex:YM}. In \cref{ex:YM_subalg} we have seen that the commutative subalgebra $\til\A_J$ is given by $C^\infty(M)\otimes\1_N$. The unitary elements of this subalgebra are then given by $U(\til\A_J) \simeq C^\infty(M,U(1)) \otimes \1_N$. Note that in this case $U(\til\A_J)$ is equal to the subgroup $Z$ of $U(N)$ that commutes with the algebra $M_N(\C)$. We thus obtain that the finite gauge group is given by the quotient $\G(F\Sub{YM}) = U(N) / U(1) = PSU(N)$, which is equal to the group of inner automorphisms of $M_N(\C)$. As in \cref{prop:unimod}, this group can also be written as $SU(N) / \mu_N$, where the multiplicative group $\mu_N$ of $N$-th roots of unity is the center of $SU(N)$. The Lie algebra $\g(F\Sub{YM})$ consists of the traceless anti-hermitian matrices $\su(N)$. 
\end{example}

\subsubsection{Full symmetry group}

Suppose we have two groups $N$ and $H$, and an action of $H$ on $N$ given by a group homomorphism $\theta\colon H\rightarrow \Aut(N)$. The \emph{semi-direct product} $N\rtimes H$ is defined to be the group $\{(n,h) \mid n\in N, h\in H\}$ with the product given by $(n,h)(n',h') := (n\theta(h)n',hh')$. One may verify that this product is associative, that the unit is given by $(1,1)\in N\rtimes H$ and that each element $(n,h)\in N\rtimes H$ has inverse $(\theta(h^{-1})(n^{-1}),h^{-1})$. Furthermore, $H$ is a subgroup and $N$ is a normal subgroup of $N\rtimes H$. Note that this automatically means that $H$ is given by the quotient $(N\rtimes H) / N$. 

We use this semi-direct product  for the description of the full symmetry group of an almost-commutative manifold $M\times F$. The `internal symmetries' of an almost-commutative manifold are given by the gauge group $\G(M\times F)$. Furthermore, we also still have invariance under the group of diffeomorphisms $\Diff(M)$. There exists a group homomorphism $\theta\colon \Diff(M)\rightarrow \Aut\big(\G(M\times F)\big)$ given by
\begin{align*}
\theta(\phi)U := U\circ\phi^{-1} ,
\end{align*}
for $\phi\in\Diff(M)$ and $U\in\G(M\times F)$. 
Hence, we can describe the \emph{full symmetry group} by the semi-direct product $\G(M\times F) \rtimes \Diff(M)$.

\paragraph{Principal bundles}

As an aside, let us now put the gauge group in the context of principal fibre bundles (see e.g.\ \cite[Def.3.2.1.]{Bleecker81}). Let $G$ be a Lie group, and suppose $P$ is a principal $G$-bundle $\pi\colon P\rightarrow M$. Let the group of automorphisms $\Aut(P)$ be given by the diffeomorphisms $f\colon P\rightarrow P$ which satisfy $f(pg) = f(p)g$ for all $p\in P$ and $g\in G$. Note that $f$ induces a well-defined diffeomorphism $\bar f\colon M\rightarrow M$ given by $\bar f\big(\pi(p)\big) := \pi\big(f(p)\big)$. Let us consider the subgroup $\G(P)$ of $\Aut(P)$ defined by
\begin{align*}
\G(P) := \left\{ g\in\Aut(P) \mid \bar g = \Id_M \right\} .
\end{align*}
Note that the condition $\bar g = \Id_M$ is equivalent to $\pi\big(g(p)\big) = \pi(p)$ for $g\in\G(P)$ and $p\in P$. This subgroup $\G(P)$ is called the group of \emph{gauge transformations} of $P$. We show that it is in fact a \emph{normal} subgroup. From the definition $\bar f\big(\pi(p)\big) := \pi\big(f(p)\big)$ we readily see that $\bar{f\circ g\circ f^{-1}} = \bar f\circ \bar g\circ \bar f^{-1}$ for $f\in\Aut(P)$ and $g\in\G(P)$. Since $\bar g = \Id_M$ for $g\in\G(P)$, we see that also $\bar f\circ \bar g\circ \bar f^{-1} = \Id_M$. Hence $f\circ g\circ f^{-1}$ is also an element of $\G(P)$, so $\G(P)$ is indeed a normal subgroup of $\Aut(P)$. For their quotient, we find 
\begin{align*}
\Aut(P) / \G(P) \simeq \Diff(M) .
\end{align*}

Consider now the (globally trivial) principal $\G(F)$-bundle $P = M\times\G(F)$. The group of gauge transformations is then given by $\G(P) = C^\infty(M,\G(F))$, which is precisely the gauge group $\G(M\times F)$ of the AC-manifold. The full symmetry group of an AC-manifold is thus given by 
$$
\G(M\times F) \rtimes \Diff(M) \simeq \Aut(P) .
$$
This can be extended to topologically non-trivial principal bundles as was done in \cite{BoeS10,Cac11}.

\subsection{Inner fluctuations and gauge transformations}

\subsubsection{Inner fluctuations}

In the previous section we have described the gauge group for an almost-commu\-ta\-tive manifold. The next step towards the description of a gauge theory is to determine the \emph{gauge fields}. These gauge fields will be seen to be given by so-called {\it inner fluctuations}. These inner fluctuations arise from considering Morita equivalences between algebras. We will not discuss such Morita equivalences here, but refer to \cite{Connes96} or \cite[Ch.~1, \S10.8]{CM07} for more details. Instead, we will simply give the resulting definition, which is of similar nature as the usual minimal coupling in the physics literature.

\begin{defn}
For a real AC-manifold $M\times F$ given by the data $(\A, \mH, D, J)$, we define the set
\begin{align*}
\Omega^1_D := \big\{ \sum_j a_j[D,b_j] \mid a_j, b_j\in\A \big\} .
\end{align*}
The hermitian elements $A=A^*\in\Omega^1_D$ are called the \emph{inner fluctuations} of the AC-manifold. We define the \emph{fluctuated Dirac operator} by 
\begin{align*}
D_A := D + A + \varepsilon'JAJ^* ,
\end{align*}
for an inner fluctuation $A=A^*\in\Omega^1_D$. 
\end{defn}

Note that for the canonical triple of a spin manifold $M$, for which $\A = C^\infty(M)$ and $D=\sD$, we have by \eqref{eq:canon_Dirac_comm} the commutation relation
\begin{equation}
[\sD,f] = -i \gamma^\mu (\partial_\mu f)
\end{equation}
for all $f\in\A$. In other words, $\Omega^1_\sD$ is given by the Clifford representation of the $1$-forms $\A^1(M)$. The elements of $\Omega^1_D$ for a general Dirac operator $D$ are therefore regarded as a generalization of $1$-forms. They will be interpreted as gauge potentials or gauge fields. 

We take $a,b\in C^\infty(M)$ and calculate elements of the form $A = a[\sD,b]$. By using the local formula $\sD = -i\gamma^\mu\nabla^S_\mu$ we find the inner fluctuation
\begin{align*}
A = -i\gamma^\mu a\partial_\mu b =: \gamma^\mu A_\mu.
\end{align*}
Since $A$ must be hermitian, $A_\mu = -ia\partial_\mu b$ must be a real function in $C^\infty(M)$. Since $J_M$ commutes with $\sD = -i\gamma^\mu\nabla^S_\mu$ and anticommutes with $i$, we know that $J_M$ must anticommute with $\gamma^\mu$. Furthermore, $J_M$ commutes with $A_\mu$ since $A_\mu$ is real. Hence we conclude
\begin{align*}
\sD_A &= \sD + A + J_MAJ_M^* = \sD + A - AJ_MJ_M^* = \sD . 
\end{align*}
So, for the canonical triple of a spin manifold $M$ there are no fluctuations of the Dirac operator $\sD$, and hence there is no gauge field (see also \cite{Landi97}).

Let us now calculate the inner fluctuations for a general AC-manifold $M\times F$. The Dirac operator $D=\sD\otimes\1+\gamma_5\otimes D_F$ consists of two terms, and hence we can also split the inner fluctuation $A=a[D,b]$ in two terms. The first term is given by
\begin{align}
\label{eq:fluc_gauge}
a[\sD\otimes\1,b] = -i\gamma^\mu \otimes a\partial_\mu b =: \gamma^\mu \otimes A_\mu ,
\end{align}
where $A_\mu:=-ia\partial_\mu b\in i\A$ must be hermitian. The second term yields 
\begin{align}
\label{eq:fluc_higgs}
a[\gamma_5\otimes D_F,b] = \gamma_5 \otimes a[D_F,b] =: \gamma_5\otimes\phi ,
\end{align}
for hermitian $\phi:=a[D_F,b]$. Thus, the inner fluctuations of an even almost-commutative manifold $M\times F$ take the form
\begin{align}
\label{eq:acm_inner_fluc}
A = \gamma^\mu\otimes A_\mu + \gamma_5\otimes\phi ,
\end{align}
for hermitian operators $A_\mu\in i\A$ \footnote{Note that $i\A=\A$ for complex algebras only.} and $\phi\in\Gamma\big(\End(E)\big)$, where $E$ is the trivial bundle $E=M\times\mH_F$. In the context of the Standard Model (\cref{sec:gauge_SM} below), we will see that the field $\phi$ describes the Higgs field, explaining the notation. 

The fluctuated Dirac operator is given by $D_A = D + A + JAJ^*$. We then calculate
\begin{align}
\label{eq:fluc_Gauge}
\gamma^\mu \otimes A_\mu + J\gamma^\mu \otimes A_\mu J^* = \gamma^\mu \otimes \big( A_\mu - J_FA_\mu J_F^* \big) =: \gamma^\mu \otimes B_\mu ,
\end{align}
where we have defined $ B_\mu \in \Gamma\big(\End(E)\big)$. We define the twisted connection $\nabla^E$ on the bundle $S\otimes E$ by
\begin{align*}
\nabla^E_\mu = \nabla^S_\mu\otimes\1 + i \1\otimes B_\mu .
\end{align*} 
We then see that we can rewrite $\sD\otimes\1 + \gamma^\mu\otimes B_\mu = -i\gamma^\mu\nabla^E_\mu$. For the remainder of the fluctuated Dirac operator, we define $\Phi\in\Gamma\big(\End(E)\big)$ such that
\begin{align}
\label{eq:fluc_Higgs}
\gamma_5\otimes D_F + \gamma_5\otimes \phi + J(\gamma_5\otimes \phi) J^* &=: \gamma_5\otimes \Phi . 
\end{align}
The fluctuated Dirac operator of a real even AC-manifold then takes the form
\begin{align}
\label{eq:acm_fluc_Dirac}
D_A = \sD\otimes\1 + \gamma^\mu\otimes B_\mu + \gamma_5\otimes\Phi = -i\gamma^\mu\nabla^E_\mu + \gamma_5\otimes\Phi .
\end{align}

In \cref{sec:gauge_group} we have obtained the local gauge group $\G(F)$ with Lie algebra $\g(F)$. For consistency we should now check that the gauge field $A_\mu$ arising from the inner fluctuation indeed corresponds to this same gauge group. 

The demand that $A_\mu$ is hermitian is equivalent to $(iA_\mu)^* = -iA_\mu$. Since $A_\mu$ is of the form $-ia\partial_\mu b$ for $a,b\in\A$ (see \eqref{eq:fluc_gauge}), we see that $iA_\mu$ is an element of the algebra $\A$ (also if $\A$ is only a real algebra). Thus by \eqref{eq:lu}, we have $A_\mu(x)\in i\,\lu(A_F)$. 

The only way in which $A_\mu$ affects the results is through the action of $A_\mu - J_FA_\mu J_F^*$. If we take $A_\mu^\prime = A_\mu - a_\mu$ for some $a_\mu\in i\h_F=i\,\lu\big((\til\A_F)_{J_F}\big)$ (which commutes with $J_F$), we see that $A_\mu^\prime - J_FA_\mu^\prime J_F^* = A_\mu - J_FA_\mu J_F^*$. Therefore we can, without any loss of generality, assume that $A_\mu(x)$ is an element of the quotient $i\g(F) = i(\lu(\A_F) / \h_F \big))$. Since $\g(F)$ is the Lie algebra of the gauge group $\G(F)$, we have confirmed that 
\begin{align}
\label{eq:gauge_field}
A_\mu \in C^\infty(M,i\,\g(F)) 
\end{align}
is indeed a gauge field for the local gauge group $\G(F)$. For the field $B_\mu$ found in \eqref{eq:acm_fluc_Dirac}, we can also write $B_\mu = \ad(A_\mu)$, where $\ad$ has been defined in \eqref{eq:ad_J}. So, we conclude that $B_\mu$ is given by the adjoint action of a gauge field $A_\mu$ for the gauge group $\G(F)$ with Lie algebra $\g(F)$.

If the finite space $F$ has a grading $\gamma_F$, the field $\phi$ satisfies $\phi\gamma_F = - \gamma_F\phi$ and the field $\Phi$ satisfies $\Phi\gamma_F = - \gamma_F\Phi$ and $\Phi J_F = J_F\Phi$. 
These relations follow directly from the definitions of $\phi$ and $\Phi$ and the commutation relations for $D_F$.

Using the cyclic property of the trace, it is easy to see that the traces of the fields $B_\mu$, $\phi$ and $\Phi$ over the finite Hilbert space $\mH_F$ vanish identically. For $B_\mu$ we find
\begin{align*}
\Tr_{\mH_F}\big(B_\mu\big) &= \Tr_{\mH_F}\big(A_\mu - J_FA_\mu J_F^*\big) = \Tr_{\mH_F}\big(A_\mu - A_\mu J_F^*J_F\big) = 0 .
\end{align*}
For the field $\phi$ we find
\begin{align*}
\Tr_{\mH_F}\big(\phi\big) = \Tr_{\mH_F}\big(a[D_F,b]\big) = \Tr_{\mH_F}\big([b,a]D_F\big) .
\end{align*}
Using the fact that the grading commutes with the algebra and anticommutes with the Dirac operator, one can show that this trace also vanishes. It then automatically follows that $\Phi = D_F + \phi + J_F\phi J_F^*$ is also traceless. 

\begin{example}
\label{ex:YM_fields}
For the Yang-Mills manifold of \cref{ex:YM}, the inner fluctuations take the form $A = \gamma^\mu \otimes A_\mu$ for a traceless hermitian field $A_\mu = A_\mu^* \in C^\infty(M,i\su(N))$. Since $J_FA_\mu J_F^* m = m A_\mu$ for $m\in M_N(\C)$, we see that for the field $B_\mu = A_\mu - J_FA_\mu J_F^*$ we obtain the action $B_\mu m = A_\mu m - m A_\mu = [A_\mu,m] = (\ad A_\mu) m$. Thus $A_\mu$ is a $PSU(N)$ gauge field which acts by the adjoint representation on the fermions in $L^2(M,S)\otimes M_N(\C)$. 
\end{example}

\subsubsection{Gauge transformations}

In \cref{sec:gauge_group} we have seen that an element $U\in \G(M\times F)$ transforms the Dirac operator as $D\rightarrow UDU^*$. Let us now consider the effect of this transformation on the fluctuated Dirac operator $D_A = D + A + \epsilon'JAJ^*$. Using the commutation rules $[a,b^0]=0$, $\big[[D,a],b^0\big]=0$ and $JD = \epsilon'DJ$, we calculate that
\begin{align*}
UDU^* &= uJuJ^*DJu^*J^*u^* = \epsilon'uJuDu^*J^*u^* = \epsilon'uJ(D+u[D,u^*])J^*u^* \\
&= uDu^* + \epsilon'JJ^*uJu[D,u^*]J^*u^* = D + u[D,u^*] + \epsilon'Ju[D,u^*]J^* .
\end{align*}
Because of the commutation rules \labelcref{eq:order0,eq:order1}, we immediately find $[A,JaJ^*]=0$, so we see that
\begin{align*}
UAU^* &= uJuJ^*AJu^*J^*u^* = uAu^* 
\intertext{and}
U\epsilon'JAJ^*U^* &= \epsilon'uJuJ^*JAJ^*Ju^*J^*u^* = \epsilon'uJuAu^*J^*u^*JJ^* \\
&= \epsilon'uJJ^*u^*JuAu^*J^* = \epsilon'JuAu^*J^* .
\end{align*}
Combining these three relations, we find that
\begin{align}
\label{eq:gauge_transf_ACM}
UD_AU^* &= D_{A^u} \quad,\quad\text{for } A^u := uAu^* + u[D,u^*] . 
\end{align}
Thus, the transformed operator $UD_AU^*$ can also be written as a fluctuated Dirac operator $D_{A^u}$, for a new fluctuation $A^u$. This only works because we consider the unitary transformation $U=uJuJ^*$ given by the \emph{adjoint action} of $u\in U(\A)$, to make sure that the conjugation operator $J$ remains unchanged. The resulting transformation on the inner fluctuation $A\rightarrow A^u$ shall be interpreted in physics as the gauge transformation of the gauge field. 

Note that for an element $U = uJuJ^*$ in the gauge group $\G(M\times F)$, there is an ambivalence in the corresponding transformation of $A$. Namely, for $u\in U(\A)$ and $h\in U(\til\A_J)$, we can also write $U = uhJuhJ^*$. By replacing $u$ with $uh$ we then obtain, using \eqref{eq:order0,eq:order1}, that
\begin{align*}
A^{uh} = uAu^* + u[D,u^*] + h[D,h^*] .
\end{align*}
However, when considering the total inner fluctuation $A^{uh}+JA^{uh}J^*$, the extra term $h[D,h^*]$ will be cancelled:
\begin{align*}
h[D,h^*] + Jh[D,h^*]J^* = h[D,h^*] + h^*[D,h]JJ^* 
= [D,hh^*] = 0 .
\end{align*}
Hence the transformation of $D_A = D + A + JAJ^*$ is well-defined. 

For an AC-manifold $M\times F$, we can write $A = \gamma^\mu\otimes A_\mu + \gamma_5\otimes\phi$ (by \eqref{eq:acm_inner_fluc}) and $D=-i\gamma^\mu\nabla^S_\mu\otimes\1 +\gamma_5\otimes D_F$, and by using that $[\nabla^S_\mu,u^*] = \partial_\mu u^*$, we thus obtain
\begin{align}
\label{eq:acm_transf_inner_fluc}
A_\mu &\rightarrow uA_\mu u^* - i u \partial_\mu u^* , \notag\\
\phi &\rightarrow u\phi u^* + u[D_F,u^*] .
\end{align}
Let us rewrite the hermitian field $A_\mu$ as the anti-hermitian field $\omega_\mu := iA_\mu\in C^\infty(M,\g(F))$. The above transformation property of the field $A_\mu$ then corresponds to
\begin{align}
\label{eq:gauge_transf}
\omega_\mu &\rightarrow u\omega_\mu u^* + u \partial_\mu u^* .
\end{align}
This is precisely the gauge transformation for a gauge field $\omega_\mu$, as desired. 

However, the transformation property of the field $\phi$ is more surprising. In the usual setup in physics, a Higgs field transforms linearly under the gauge group. The transformation for $\phi$ derived above on the other hand is non-linear. From the framework of noncommutative geometry this is no surprise, since both bosonic fields $A_\mu$ and $\phi$ are obtained from the inner fluctuations of the Dirac operator, and are thereby expected to transform in a similar manner. It might be though that for particular choices of the finite space $F$, the transformation property of $\phi$ reduces to a linear transformation. An example of this will be discussed in \cref{chap:ex_GWS}, where we derive the electroweak sector of the Standard Model as an almost-commutative manifold.

\subsection{The action functional}
We shall now continue to introduce interesting functionals on AC-manifolds, that are invariant under the action of unitary elements of the algebra. 

For an AC-manifold $M\times F$ given by the data $(\A, \mH, D)$, we define the \emph{spectral action} as \cite{CC96,CC97}
\begin{equation}
\label{eq:spectral_action}
S_b := \Tr \left(f\Big(\frac{D_A}{\Lambda}\Big)\right) ,
\end{equation}
where $f$ is a positive even function, $\Lambda$ is a cut-off parameter and $D_A$ is the fluctuated Dirac operator. The function $f$ may be considered as a smooth approximation of a cut-off function and as such counts the number of eigenvalues of $D_A$ smaller than $\Lambda$. However, such a restriction is not necessary and we will not do so.

The spectral action accounts only for the purely bosonic part of the action. For the terms involving fermions and their coupling to the bosons, we need something else. The precise form of the fermionic action depends on the KO-dimension of the AC-manifold. We will only consider the case of KO-dimension $2$ and give the fermionic action for this case. Referring to the sign table of \cref{defn:real_structure}, we thus have the relations
\begin{align}
\label{eq:real_KO2}
J^2 &= -1 , & JD &= DJ , & J\gamma &= -\gamma J .
\end{align}
We use the decomposition $\mH = \mH^+ \oplus \mH^-$ by the grading $\gamma$. Following \cite{CCM07} (cf.\ \cite[Ch.~1, \S16.2-3]{CM07}), the relations above yield a natural construction of an antisymmetric form on $\mH^+$, namely we define
\begin{align*}
\mA_D(\xi,\xi') = \langle J\xi,D\xi'\rangle
\end{align*}
for $\xi,\xi'\in\mH^+$, where $\langle\;,\;\rangle$ is the inner product on $\mH$. This inner product is antilinear in the first variable, and since $J$ is also antilinear, $\mA_D$ is a bilinear form. We check that it is antisymmetric:
\begin{align*}
\mA_D(\xi,\xi') &= 
- \langle J\xi,J^2D\xi'\rangle = -\langle JD\xi',\xi\rangle = - \langle DJ\xi',\xi\rangle = -\langle J\xi',D\xi\rangle = -\mA_D(\xi',\xi),
\end{align*}
where we have used the relations of \eqref{eq:real_KO2} and the fact that $J$ is antiunitary, i.e.\ $\langle J\xi,J\xi'\rangle = \langle\xi',\xi\rangle$ for all $\xi,\xi'\in\mH$. Furthermore, we can restrict $\mA_D$ to $\mH^+$ without automatically getting zero, since we have $\gamma JD = JD\gamma$. For $\xi = \gamma\xi, \xi' = \gamma\xi' \in\mH^+$, we have
\begin{align*}
\langle J\xi,D\xi'\rangle &= \langle J\gamma\xi,D\xi'\rangle = - \langle \gamma J\xi,D\xi'\rangle = - \langle J\xi,\gamma D\xi'\rangle = \langle J\xi,D\gamma\xi'\rangle = \langle J\xi,D\xi'\rangle . 
\end{align*}
We define the set of \emph{classical fermions} corresponding to $\mH^+$, 
\begin{align*}
\mH^+_\text{cl} := \{\til\xi \mid \xi\in\mH^+\} ,
\end{align*}
as the set of Grassmann variables $\til\xi$ for $\xi\in\mH^+$. 

\begin{defn}
\label{defn:act_funct}
For a real even AC-manifold $M\times F$ of KO-dimension $2$ we define the full \emph{action functional} by
\begin{align*}
S := S_b + S_f := \Tr \left(f\Big(\frac{D_A}{\Lambda}\Big)\right) + \frac12 \langle J\til\xi,D_A\til\xi\rangle ,
\end{align*}
for $\til\xi\in\mH^+_\text{cl}$. The factor $\frac12$ in front of the \emph{fermionic action} $S_f$ has been chosen for future convenience.
\end{defn}

\begin{remark}
The above formulas for the bosonic and fermionic action look rather different for both cases. A more symmetrically looking proposal was put forward recently in \cite{Sit08}.

\end{remark}

One should note that we have incorporated two restrictions in the fermionic action $S_f$. The first is that we restrict ourselves to even vectors in $\mH^+$, instead of considering all vectors in $\mH$. The second restriction is that we do not consider the inner product $\langle J\til\xi',D_A\til\xi\rangle$ for two independent vectors $\xi$ and $\xi'$, but instead use the same vector $\xi$ on both sides of the inner product. Each of these restrictions reduces the number of degrees of freedom in the fermionic action by a factor $2$, yielding a factor $4$ in total. It is precisely this approach that solves the problem of fermion doubling pointed out in \cite{LMMS97} (see also the discussion in \cite[Ch.~1, \S16.3]{CM07}). We shall discuss this in more detail in \cref{chap:ex_ED}, where we calculate the fermionic action for electrodynamics.

\subsubsection{Invariance of the action functional}
\label{sec:act_invar}

Above we have defined the action functional $S_b + S_f$ for an AC-manifold $M\times F$, and of course we want this action to be invariant under the gauge group $\G(M\times F)$. Therefore, let us now check that both the bosonic action $S_b$ and the fermionic action $S_f$ are indeed invariant functionals, and can thus suitably be used in the description of a gauge theory. 

Let us first consider the spectral action $S_b$. The transformation of the fluctuated Dirac operator is given by $D_A\rightarrow UD_AU^*$ for $U\in \G(M\times F)$, so the spectral action transforms as
\begin{align*}
\Tr \left(f\Big(\frac{D_A}{\Lambda}\Big)\right) \mapsto \Tr \left(f\Big(\frac{UD_AU^*}{\Lambda}\Big)\right) .
\end{align*}
The trace depends only on the discrete spectrum of the fluctuated Dirac operator $D_A$. The unitary transformation has no effect on this spectrum. Namely, if we let $\psi_n$ be the eigenvectors of $D_A$ with eigenvalues $\lambda_n$, then the vectors $\psi_n' := U\psi_n$ are the eigenvectors of $D_A' := UD_AU^*$ with the same eigenvalues $\lambda_n$:
\begin{align*}
D_A'\psi_n' = UD_AU^* U\psi_n = U D_A\psi_n = U\lambda_n\psi_n = \lambda_n U\psi_n = \lambda_n \psi_n' .
\end{align*}
For the spectral action, we thus obtain
\begin{align*}
\Tr \left(f\Big(\frac{D_A}{\Lambda}\Big)\right) &= \sum_n f\Big(\frac{\lambda_n}{\Lambda}\Big) = \Tr \left(f\Big(\frac{UD_AU^*}{\Lambda}\Big)\right) . 
\end{align*}

Next, consider the fermionic action $S_f$. The transformation of the fluctuated Dirac operator is given by $D_A\rightarrow UD_AU^*$ for $U\in \G(M\times F)$, whereas the conjugation operator remains unchanged since $UJU^*=J$. From the unitarity of $U$ we then easily see that
\begin{align*}
\langle J\til\xi,D_A\til\xi\rangle \mapsto &\langle JU\til\xi,UD_AU^*U\til\xi\rangle = \langle UJ\til\xi,UD_A\til\xi\rangle = \langle J\til\xi,D_A\til\xi\rangle . 
\end{align*}
So, we have confirmed that the total action functional $S_b + S_f$ is indeed invariant under the gauge group $\G(M\times F)$.

\subsection{Gauge theories from almost-commutative manifolds}
\label{sec:gauge_ACM}

In this section we have used the data $(\A, \mH, D, \gamma, J)$ describing an almost-commu\-ta\-tive manifold $M\times F$ to describe a gauge group, gauge fields, gauge transformations as well as invariant action functionals. These results can now be summarized as follows:

\begin{thm}
\label{thm:gauge_acm}
A real even almost-commutative manifold $M\times F$ describes a gauge theory on $M$ with gauge group $\G(M\times F) = C^\infty(M,\G(F))$. 
\end{thm}
\begin{proof}
In \eqref{eq:gauge_field} we have obtained that $iA_\mu(x) \in \g(F) = \lu(\A_F) / \h_F$. The total algebra is given by $\A = C^\infty(M,\A_F)$, and this is by construction the space of smooth sections of the trivial bundle $M\times\A_F$. Therefore the gauge field $A_\mu$ defines a global smooth $\g(F)$-valued $1$-form $\omega=iA_\mu dx^\mu$. Consider the trivial principal bundle $P = M\times \G(F)$. Because of the transformation property \eqref{eq:gauge_transf}, we see that $\omega$ is a connection form on $P$. 

The group of gauge transformations for a trivial principal fibre bundle $P=M\times \G(F)$ is given by $C^\infty(M,\G(F))$, and by \cref{prop:acm_gauge} this group is equal to $\G(M\times F)$. This means that the gauge group of this principal bundle $P$ is identical to the gauge group of the almost-commutative manifold, as defined in \cref{defn:gauge_group_NCG}. We have seen in \cref{sec:act_invar} that the total Lagrangian we obtain from the bosonic and fermionic action functionals is invariant under this gauge group. 

Since the representation of $\A_F$ on $\mH_F$ induces a representation of $\G(F)$ on $\mH_F$, we see that $M\times\mH_F$ is an associated vector bundle of the principal bundle $P=M\times \G(F)$. We have thus seen that, from an almost-commutative manifold, we can recover all the ingredients of a gauge theory. 
\end{proof}

In the above theorem, we have used the gauge field $A_\mu$ to construct a connection on a (trivial) principal $\G(F)$-bundle $P=M\times \G(F)$. We have seen that $E=M\times\mH_F$ is an associated vector bundle of $P$, and this provides an action of the gauge group on the fermionic particle fields. One should note however that the total Hilbert space of an AC-manifold is given by $\mH = L^2(M,S)\otimes \mH_F = L^2(M,S\otimes E)$, so the particle fields on an AC-manifold are sections of the total bundle $S\otimes E$, and this total bundle is not an associated vector bundle of $P$.

\newpage
\section{The Spectral Action on AC-manifolds}
\label{chap:spectral}

In this section we shall derive, from the spectral action of \eqref{eq:spectral_action}, an explicit formula for the bosonic Lagrangian of an almost-commutative manifold $M\times F$. We will start by calculating a generalized Lichnerowicz formula for the square of the fluctuated Dirac operator. Then, we will show how we can use this formula to obtain the heat expansion of the spectral action. We will explicitly calculate this heat expansion, allowing for a derivation of the general form of the Lagrangian for an almost-commutative manifold.

\subsection{The heat expansion of the spectral action}

\subsubsection{A generalized Lichnerowicz formula}

Suppose we have a vector bundle $E\rightarrow M$. An important example of a second order differential operator is the Laplacian $\Delta^E$ of a connection $\nabla^E$ on $E$. We say that a second order differential operator $H$ is a \emph{generalized Laplacian} if it is of the form $H = \Delta^E - F$, for some $F\in\Gamma(\End(E))$. For more details on generalized Laplacians we refer the reader to \cite[\S2.1]{BGV92}. 

We can then define a (generalized) \emph{Dirac operator} on a $\Z_2$-graded vector bundle $E$ as a first order differential operator on $E$ of odd parity, i.e.\ $D\colon \Gamma(M,E^\pm) \rightarrow \Gamma(M,E^\mp)$, such that $D^2$ is a generalized Laplacian (see \cite[section 3.3]{BGV92}).

Our first task is to show that the fluctuated Dirac operator $D_A$ of an almost-commutative manifold, is indeed a (generalized) Dirac operator. In other words, we would like to show that ${D_A}^2$ can be written in the form $\Delta^E - F$. Before we prove this, let us first have a closer look at some explicit formulas for the fluctuated Dirac operator. Recall from \eqref{eq:acm_fluc_Dirac} that we can write
\begin{align*}
D_A = -i\gamma^\mu\nabla^E_\mu + \gamma_5\otimes\Phi ,
\end{align*}
for the connection $\nabla^E_\mu = \nabla^S_\mu\otimes\1 + i \1\otimes B_\mu$ on $S\otimes E$ and for the Higgs field $\Phi\in\Gamma(\End(E))$. Let us evaluate the relations between the connection, its curvature and their adjoint actions. We define the operator $D_\mu$ as the adjoint action of the connection $\nabla^E_\mu$, i.e.\ $D_\mu = \ad\big(\nabla^E_\mu\big)$. In other words, we have
\begin{align}
\label{eq:cov_der}
D_\mu\Phi = [\nabla^E_\mu,\Phi] = \partial_\mu\Phi + i[ B_\mu,\Phi] .
\end{align}
We shall define the curvature $F_{\mu\nu}$ of the gauge field $B_\mu$ by
\begin{align}
\label{eq:field_curv}
F_{\mu\nu} := \partial_\mu B_\nu - \partial_\nu B_\mu + i[ B_\mu, B_\nu] .
\end{align}
The curvature of the connection $\nabla^E$ is defined as 
\begin{align}
\label{eq:defn_curv}
\Omega^E(X,Y) = \nabla^E_X \nabla^E_Y - \nabla^E_Y \nabla^E_X - \nabla^E_{[X,Y]}
\end{align}
for two vector fields $X,Y$. Since in local coordinates we have $[\partial_\mu,\partial_\nu]=0$, we find
\begin{align*}
\Omega^E_{\mu\nu} &= \nabla^E_\mu\nabla^E_\nu - \nabla^E_\nu\nabla^E_\mu \\
&= (\nabla^S_\mu\otimes\1 + i \1 \otimes B_\mu)(\nabla^S_\nu\otimes\1 + i \1 \otimes B_\nu) \\
&\quad- (\nabla^S_\nu\otimes\1 + i \1 \otimes B_\nu)(\nabla^S_\mu\otimes\1 + i \1 \otimes B_\mu) \\
&= \Omega^S_{\mu\nu} \otimes\1 + i \1 \otimes \partial_\mu B_\nu - i \1 \otimes \partial_\nu B_\mu - \1 \otimes [B_\mu,B_\nu] .
\end{align*}
By inserting \eqref{eq:field_curv}, we obtain the formula 
\begin{align}
\label{eq:acm_curv}
\Omega^E_{\mu\nu} = \big[\nabla^E_\mu,\nabla^E_\nu\big] = \Omega^S_{\mu\nu}\otimes\1 + i \1 \otimes F_{\mu\nu} .
\end{align}

Next, let us have a look at the commutator $\big[D_\mu,D_\nu\big]$. By using the definition of $D_\mu$ and the Jacobi identity, we obtain
\begin{align*}
[D_\mu,D_\nu]\Phi &= \ad\big(\nabla^E_\mu\big)\ad\big(\nabla^E_\nu\big)\Phi - \ad\big(\nabla^E_\nu\big)\ad\big(\nabla^E_\mu\big)\Phi \\
&= \big[\nabla^E_\mu,[\nabla^E_\nu,\Phi]\big] - \big[\nabla^E_\nu,[\nabla^E_\mu,\Phi]\big] \\
&= \big[[\nabla^E_\mu,\nabla^E_\nu],\Phi]\big] = \big[\Omega^E_{\mu\nu},\Phi\big] = \ad\big(\Omega^E_{\mu\nu}\big) \Phi .
\end{align*}
Since $\Omega^S_{\mu\nu}$ commutes with $\Phi$, we obtain the relation
\begin{align*}
\big[D_\mu,D_\nu\big] = i \ad\big(F_{\mu\nu}\big) .
\end{align*}
Note that this relation simply reflects the fact that $\ad$ is a Lie algebra homomorphism.

In local coordinates, the Laplacian is given by $\Delta^E = -g^{\mu\nu} \left( \nabla^E_\mu \nabla^E_\nu - \Gamma^\rho_{\phantom{\rho}\mu\nu} \nabla^E_\rho \right)$. We can then calculate the explicit formula
\begin{align}
\Delta^E &= -g^{\mu\nu} \left( \nabla^E_\mu\nabla^E_\nu - \Gamma^\rho_{\phantom{\rho}\mu\nu} \nabla^E_\rho \right) \notag\\
&= \Delta^S\otimes\1 \notag\\
&\quad- g^{\mu\nu} \Big( i(\nabla^S_\mu\otimes\1) (\1 \otimes B_\nu) + i (\1 \otimes B_\mu)(\nabla^S_\nu\otimes\1) - \1 \otimes B_\mu B_\nu - i\Gamma^\rho_{\phantom{\rho}\mu\nu} \otimes B_\rho \Big) \notag\\
\label{eq:acm_Lap}
&= \Delta^S\otimes\1 - 2 i (\1 \otimes B^\mu)(\nabla^S_\mu\otimes\1) \notag\\
&\quad- ig^{\mu\nu}(\1 \otimes \partial_\mu B_\nu) + \1 \otimes B_\mu B^\mu + ig^{\mu\nu}\Gamma^\rho_{\phantom{\rho}\mu\nu} \otimes B_\rho .
\end{align}

We are now ready to prove that the fluctuated Dirac operator $D_A$ of an almost-commutative manifold satisfies the following \emph{generalized Lichnerowicz formula} or \emph{Weitzenb\"ock formula}. First, for the canonical Dirac operator $\sD$ on a compact Riemannian spin manifold $M$ we have the Lichnerowicz formula (see, for instance, \cite[Theorem 9.16]{GVF01}) 
\begin{equation}
\label{eq:lich}
\sD^2 = \Delta^S + \frac14 s ,
\end{equation}
where $\Delta^S$ is the Laplacian of the spin connection $\nabla^S$, and $s$ is the scalar curvature of $M$. 

\begin{prop}
\label{prop:acm_Dirac_sq}
The square of the fluctuated Dirac operator of an almost-commu\-ta\-tive manifold is a generalized Laplacian of the form
\begin{align*}
{D_A}^2 &= \Delta^E - Q .
\end{align*}
The endomorphism $Q$ is given by
\begin{align*}
Q = -\frac14s\otimes\1 - \1\otimes\Phi^2 + \frac12 i \gamma^\mu\gamma^\nu \otimes F_{\mu\nu} - i\gamma^\mu\gamma_5 \otimes D_\mu \Phi ,
\end{align*}
where $D_\mu$ and $F_{\mu\nu}$ are defined in \eqref{eq:cov_der,eq:field_curv}. 
\end{prop}
\begin{proof}
Rewriting the formula for $D_A$, we have
\begin{align*}
{D_A}^2 &= \left( \sD\otimes\1 + \gamma^\mu \otimes B_\mu + \gamma_5\otimes\Phi \right)^2 \notag\\
&= \sD^2\otimes\1 + \gamma^\mu\gamma^\nu \otimes B_\mu B_\nu + \1\otimes\Phi^2 + (\sD\gamma^\mu\otimes\1) (\1 \otimes B_\mu) \\
&\quad+ (\1 \otimes B_\mu) (\gamma^\mu\sD\otimes\1) + (\sD\otimes\1)(\gamma_5\otimes\Phi) + (\gamma_5\otimes\Phi)(\sD\otimes\1) \\
&\quad+ (\gamma^\mu \otimes B_\mu)(\gamma_5\otimes\Phi) + (\gamma_5\otimes\Phi)(\gamma^\mu \otimes B_\mu) .
\end{align*}
For the first term we use the Lichnerowicz formula of \eqref{eq:lich}. We rewrite the second term into 
\begin{multline*}
\gamma^\mu\gamma^\nu \otimes B_\mu B_\nu = \frac12 \gamma^\mu\gamma^\nu\otimes \left( B_\mu B_\nu +  B_\nu B_\mu + [ B_\mu, B_\nu]\right) \\
= \1 \otimes B_\mu B^\mu + \frac12 \gamma^\mu\gamma^\nu \otimes [ B_\mu, B_\nu] .
\end{multline*}
where we have used the Clifford relation to obtain the second equality. For the fourth and fifth terms we use the local formula $\sD = -i\gamma^\nu\nabla^S_\nu$ to obtain
\begin{multline*}
(\sD\gamma^\mu\otimes\1) (\1 \otimes B_\mu) + (\1 \otimes B_\mu) (\gamma^\mu\sD\otimes\1) \\
= - (i\gamma^\nu\nabla^S_\nu \gamma^\mu\otimes\1) (\1 \otimes B_\mu) - (\1 \otimes B_\mu) (\gamma^\mu i\gamma^\nu\nabla^S_\nu\otimes\1) .
\end{multline*}
Using the identity $[\nabla^S_\nu,c(\alpha)] = c(\nabla_\nu \alpha)$ for the spin connection, we find $[\nabla^S_\nu\otimes\1,(\gamma^\mu\otimes\1) (\1 \otimes B_\mu)] = c\big(\nabla_\nu (\theta^\mu \otimes B_\mu)\big)$. We thus obtain
\begin{align*}
(\sD&\gamma^\mu\otimes\1) (\1 \otimes B_\mu) + (\1 \otimes B_\mu) (\gamma^\mu\sD\otimes\1) \\
&= - i(\gamma^\nu\otimes\1) c\big(\nabla_\nu ( \theta^\mu \otimes B_\mu)\big) \\
&\quad- i(\gamma^\nu\gamma^\mu\otimes\1)(\1 \otimes B_\mu) (\nabla^S_\nu\otimes\1) - i (\1 \otimes B_\mu)(\gamma^\mu \gamma^\nu\nabla^S_\nu\otimes\1) \\
&= - i(\gamma^\nu\otimes\1) c\big(\theta^\mu \otimes (\partial_\nu  B_\mu) - \Gamma^\rho_{\phantom{\rho}\mu\nu}\theta^\mu \otimes B_\rho\big) - 2i (\1 \otimes B^\nu) (\nabla^S_\nu\otimes\1) \\
&= - i(\gamma^\nu \gamma^\mu\otimes\1) \Big(\1 \otimes \partial_\nu  B_\mu - \Gamma^\rho_{\phantom{\rho}\mu\nu} \otimes B_\rho \Big) - 2i (\1 \otimes B^\nu) (\nabla^S_\nu\otimes\1) \\
&= -i(\gamma^\nu \gamma^\mu\otimes\1) (\1 \otimes \partial_\nu B_\mu) + i g^{\mu\nu} \Gamma^\rho_{\phantom{\rho}\mu\nu} \otimes B_\rho - 2i (\1 \otimes B^\nu) (\nabla^S_\nu\otimes\1) .
\end{align*}
The sixth and seventh terms are rewritten into
\begin{align*}
&(\sD\otimes\1)(\gamma_5\otimes \Phi) + (\gamma_5\otimes \Phi)(\sD\otimes\1) = - (\gamma_5\otimes\1) \big[\sD\otimes\1,\1\otimes\Phi\big] \\
&\quad= (\gamma_5\otimes\1) (i\gamma^\mu \otimes \partial_\mu \Phi) = i\gamma_5\gamma^\mu \otimes \partial_\mu \Phi . 
\end{align*}
The eighth and ninth terms are rewritten as 
\begin{align*}
(\gamma^\mu \otimes B_\mu)(\gamma_5\otimes\Phi) + (\gamma_5\otimes\Phi)(\gamma^\mu \otimes B_\mu) = - \gamma_5\gamma^\mu\otimes[B_\mu,\Phi] .
\end{align*}
Summing all these terms then yields the formula 
\begin{align*}
{D_A}^2 &= (\Delta^S + \frac14s)\otimes\1 + (\1 \otimes B_\mu B^\mu) + \frac12 \gamma^\mu\gamma^\nu \otimes [ B_\mu, B_\nu] + \1\otimes\Phi^2 \notag\\
&\quad- i(\gamma^\nu \gamma^\mu\otimes\1) (\1 \otimes \partial_\nu  B_\mu) + i g^{\mu\nu} \Gamma^\rho_{\phantom{\rho}\mu\nu} \otimes B_\rho - 2i (\1 \otimes B^\nu) (\nabla^S_\nu\otimes\1) \notag\\
&\quad+ i\gamma_5\gamma^\mu \otimes \partial_\mu \Phi - \gamma_5\gamma^\mu\otimes[B_\mu,\Phi] .
\end{align*}
Inserting the formula for $\Delta^E$ obtained in \eqref{eq:acm_Lap} we obtain
\begin{align*}
{D_A}^2 &= \Delta^E + \frac14s\otimes\1 + \frac12 \gamma^\mu\gamma^\nu \otimes [ B_\mu, B_\nu] + \1\otimes\Phi^2 - i(\gamma^\nu \gamma^\mu\otimes\1) (\1 \otimes \partial_\nu  B_\mu) \notag\\
&\quad+ i g^{\mu\nu}(\1 \otimes \partial_\mu  B_\nu) + i\gamma_5\gamma^\mu \otimes \partial_\mu \Phi - \gamma_5\gamma^\mu\otimes[B_\mu,\Phi] .
\end{align*}
Using \eqref{eq:field_curv}, we shall rewrite
\begin{align*}
-i(\gamma^\nu& \gamma^\mu\otimes\1) (\1 \otimes \partial_\nu  B_\mu) + i g^{\mu\nu}(\1 \otimes \partial_\mu  B_\nu) \\
&= -i (\gamma^\nu \gamma^\mu\otimes\1) (\1 \otimes \partial_\nu  B_\mu) + \frac12 i (\gamma^\mu\gamma^\nu+\gamma^\nu\gamma^\mu) \otimes (\partial_\mu  B_\nu) \\
&= -\frac12 i \gamma^\mu\gamma^\nu \otimes (\partial_\mu  B_\nu) + \frac12 i \gamma^\nu\gamma^\mu \otimes (\partial_\mu B_\nu) \\
&= -\frac12 i \gamma^\mu\gamma^\nu \otimes F_{\mu\nu} - \frac12 \gamma^\mu\gamma^\nu \otimes [ B_\mu, B_\nu].
\end{align*}
Using \eqref{eq:cov_der}, we thus finally obtain
\begin{align*}
{D_A}^2 &= \Delta^E + \frac14s\otimes\1 + \1\otimes\Phi^2 - \frac12 i \gamma^\mu\gamma^\nu \otimes F_{\mu\nu} + i\gamma_5\gamma^\mu \otimes D_\mu \Phi ,
\end{align*}
from which we can read off the formula for $Q$. 
\end{proof}

\subsubsection{The heat expansion}

Below we present two important theorems (without proof) which we will need to calculate the spectral action of almost-commutative manifolds. The first of these theorems states that there exists a heat expansion for a generalized Laplacian. The second theorem gives explicit formulas for the first three non-zero coefficients of this expansion. Next, we will show how these theorems can be applied to obtain a perturbative expansion of the spectral action for an almost-commutative manifold. 

\begin{thm}[{\cite[\S1.7]{Gilkey84}}]
\label{thm:heat_expansion}
For a generalized Laplacian $H$ on $E$ we have the following expansion in $t$, known as the \emph{heat expansion:} 
\begin{equation}
\Tr\left(e^{-tH}\right) \sim \sum_{k\geq0} t^{\frac{k-n}2} a_k(H) ,
\end{equation}
where $n$ is the dimension of the manifold, the trace is taken over the Hilbert space $L^2(M,E)$ and the coefficients of the expansion are given by
\begin{equation}
a_k(H) := \int_M a_k(x,H) \sqrt{|g|} d^4x. 
\end{equation}
The coefficients $a_k(x,H)$ are called the Seeley-DeWitt coefficients.
\end{thm}
For a more physicist friendly approach, we refer to \cite{Vas03}. We also state here without proof Theorem 4.8.16 from Gilkey \cite{Gilkey84}. Note that the conventions used by Gilkey for the Riemannian curvature $R$ are such that $g^{\mu\rho}g^{\nu\sigma}R_{\mu\nu\rho\sigma}$ is negative for a sphere, in contrast to our own conventions. Therefore we have replaced $s=-R$. Furthermore, we have used that $f_{;\mu}^{\phantom{\mu};\mu} = -\Delta f$ for $f\in C^\infty(M)$. 

\begin{thm}[{\cite[Theorem 4.8.16]{Gilkey84}}]
\label{thm:seeley-dewitt}
For a generalized Laplacian $H = \Delta^E - F$ the Seeley-DeWitt coefficients are given by
\begin{align*} 
a_0(x,H) &= (4\pi)^{-\frac n2} \Tr(\Id) , \qquad\qquad a_2(x,H) = (4\pi)^{-\frac n2} \Tr\left(\frac s6 + F\right) ,\\ 
a_4(x,H) &= (4\pi)^{-\frac n2} \frac1{360} 
\begin{aligned}[t] 
\Tr\big(&-12 \Delta s 
+ 5 s^2 - 2 R_{\mu\nu} R^{\mu\nu} + 2 R_{\mu\nu\rho\sigma} R^{\mu\nu\rho\sigma} \\
&+ 60 sF + 180 F^2 - 60 \Delta F + 30 \Omega^E_{\mu\nu} (\Omega^E)^{\mu\nu} \big) ,
\end{aligned}
\end{align*}
where the traces are now taken over the fibre $E_x$. Here $s$ is the scalar curvature of the Levi-Civita connection $\nabla$, $\Delta$ is the scalar Laplacian and $\Omega^E$ is the curvature of the connection $\nabla^E$ corresponding to $\Delta^E$. All $a_k(x,H)$ with odd $k$ vanish.
\end{thm}

We have seen in \cref{prop:acm_Dirac_sq} that the square of the fluctuated Dirac operator of an almost-commutative manifold is a generalized Laplacian. Applying \cref{thm:heat_expansion} on ${D_A}^2$ then yields the heat expansion: 
\begin{equation}
\Tr\left(e^{-t{D_A}^2}\right) \sim \sum_{k\geq0} t^{\frac{k-4}2} a_k({D_A}^2) ,
\end{equation}
where the Seeley-DeWitt coefficients are given by \cref{thm:seeley-dewitt}. In the following proposition, we use this heat expansion for ${D_A}^2$ to obtain an expansion of the spectral action. 

\begin{prop} 
\label{prop:action_expansion}
For an almost-commutative manifold, the spectral action given by \eqref{eq:spectral_action} can be expanded in powers of $\Lambda$ in the form
\begin{equation}
\Tr \left(f\Big(\frac {D_A}\Lambda\Big)\right) \sim a_4({D_A}^2) f(0) + 2 \sum_{\substack{0\leq k<4 \\ k\text{ even}}} f_{4-k}\Lambda^{4-k} a_k({D_A}^2) \;\frac1{\Gamma\big(\frac{4-k}2\big)} + O(\Lambda^{-1}),
\end{equation}
where $f_j = \int_0^\infty f(v) v^{j-1} dv$ are the moments of the function $f$ for $j>0$.
\end{prop}
\begin{proof}
This proof is partly based on \cite[Theorem 1.145]{CM07}.
Consider a function $g(u)$ and its Laplace transform
\begin{equation*}
g(v) = \int_0^\infty e^{-sv} h(s) ds .
\end{equation*}
We can then formally write
\begin{equation*}
g(t{D_A}^2) = \int_0^\infty e^{-st{D_A}^2} h(s) ds .
\end{equation*}
We now take the trace and use the heat expansion of ${D_A}^2$ to obtain
\begin{align*}
\Tr\big( g(t{D_A}^2) \big) &= \int_0^\infty \Tr\big(e^{-st{D_A}^2}\big) h(s) ds \sim \int_0^\infty \sum_{k\geq0} (st)^{\frac{k-4}2} a_k({D_A}^2) h(s) ds \\
&=  \sum_{k\geq0} t^{\frac{k-4}2} a_k({D_A}^2) \;\int_0^\infty s^{\frac{k-4}2} h(s) ds .
\end{align*}
The parameter $t$ is considered to be a small expansion parameter. From here on we will therefore drop the terms with $k>4$. The term with $k=4$ equals 
\begin{equation*}
a_4({D_A}^2) \;\int_0^\infty s^0 h(s) ds = a_4({D_A}^2) g(0) .
\end{equation*}
We can rewrite the terms with $k<4$ using the definition of the $\Gamma$-function as the analytic continuation of
\begin{equation}
\Gamma(z) = \int_0^\infty r^{z-1} e^{-r} dr
\end{equation}
for $z\in\C$, and by inserting $r=sv$, we see that (for $k<4$)
\begin{align*}
\Gamma\Big(\frac{4-k}2\Big) &= \int_0^\infty (sv)^{\frac{4-k}2-1} e^{-sv} d(sv) = s^{\frac{4-k}2} \int_0^\infty v^{\frac{4-k}2-1} e^{-sv} dv .
\end{align*}
From this we obtain an expression for $s^{\frac{k-4}2}$, which we insert into the equation for $\Tr\big( g(t{D_A}^2) \big)$, and then we perform the integration over $s$ to obtain
\begin{multline*}
\Tr\big( g(t{D_A}^2) \big) \sim \; a_4({D_A}^2) f(0) \\
+ \sum_{0\leq k<4} t^{\frac{k-4}2} a_k({D_A}^2) \;\frac1{\Gamma\big(\frac{4-k}2\big)} \int_0^\infty  v^{\frac{4-k}2-1} g(v) dv + O(\Lambda^{-1}) .
\end{multline*}
Now we choose the function $g$ such that $g(u^2) = f(u)$. We rewrite the integration over $v$ by substituting $v=u^2$ and obtain
\begin{align*}
\int_0^\infty  v^{\frac{4-k}2-1} g(v) dv &= \int_0^\infty  u^{4-k-2} g(u^2) d(u^2) = 2 \int_0^\infty  u^{4-k-1} f(u) du ,
\end{align*}
which by definition equals $2f_{4-k}$. Upon writing $t=\Lambda^{-2}$ we have modulo $\Lambda^{-1}$:
\begin{align*}
\Tr \left(f\Big(\frac {D_A}\Lambda\Big)\right) &= \Tr \left(g(\Lambda^{-2}{D_A}^2)\right) \\
&\sim a_4({D_A}^2) f(0) + 2 \sum_{0\leq k<4} f_{4-k}\Lambda^{4-k} a_k({D_A}^2) \;\frac1{\Gamma\big(\frac{4-k}2\big)}.
\end{align*}
Using $a_k({D_A}^2)=0$ for odd $k$, the proof follows.
\end{proof}

\subsection{The spectral action of almost-commutative manifolds}

In the previous section, we have obtained a perturbative expansion of the spectral action for an almost-commutative manifold. We will now explicitly calculate the coefficients in this expansion, first for the canonical triple (yielding the Einstein-Hilbert action of General Relativity) and then for a general almost-commutative manifold.

By \cref{prop:action_expansion} we have
\begin{equation}
\label{eq:canonical_expansion}
\Tr \left(f\Big(\frac {D_A}\Lambda\Big)\right) \sim 2f_4\Lambda^4 a_0({D_A}^2) + 2f_2\Lambda^2 a_2({D_A}^2) + f(0) a_4({D_A}^2) + O(\Lambda^{-1}) .
\end{equation}
Recall the Lichnerowicz formula from \eqref{eq:lich}, which says $\sD^2 = \Delta^S + \frac14 s$, where $\Delta^S$ is the Laplacian of the spin connection $\nabla^S$, and $s$ is the scalar curvature of the Levi-Civita connection. Using this formula, we can calculate the Seeley-DeWitt coefficients from \cref{thm:seeley-dewitt}. 

\begin{prop} 
\label{prop:canon_spec_act}
For the canonical triple $(C^\infty(M),L^2(M,S),\sD)$, the spectral action is given by:
\begin{equation}
\Tr \left(f\Big(\frac \sD\Lambda\Big)\right) \sim \int_M \L_M(g_{\mu\nu}) \sqrt{|g|} d^4x + O(\Lambda^{-1}) ,
\end{equation}
where the Lagrangian is defined by
\begin{align*}
\L_M(g_{\mu\nu}) := \frac{f_4\Lambda^4}{2\pi^2} - \frac{f_2\Lambda^2}{24\pi^2} s + \frac{f(0)}{16\pi^2} \Big(\frac1{30} \Delta s - \frac1{20} C_{\mu\nu\rho\sigma} C^{\mu\nu\rho\sigma} + \frac{11}{360}R^*R^* \Big) .
\end{align*}
\end{prop}
\begin{proof}
We have $m=\dim M=4$, and $\Tr(\Id) = \dim S = 2^{m/2} = 4$. Inserting this into \cref{thm:seeley-dewitt} gives
\begin{equation*}
a_0(\sD^2) = \frac1{4\pi^2} \int_M \sqrt{|g|} d^4x .
\end{equation*}
From the Lichnerowicz formula we see that $F = -\frac14 s\,\Id$, so
\begin{equation*}
a_2(\sD^2) = -\frac1{48\pi^2} \int_M s \sqrt{|g|} d^4x.
\end{equation*}
Using $F = -\frac14 s\,\Id$ we calculate
\begin{align*}
5s^2\Id + 60sF + 180F^2 = \frac54 s^2 \Id .
\end{align*}
Inserting this into $a_4(\sD^2)$ gives
\begin{align*}
a_4(\sD^2) = \frac1{16\pi^2} \frac1{360} \int_M \Tr\big(& 3 \Delta s\,\Id +\frac54 s^2\Id - 2 R_{\mu\nu} R^{\mu\nu}\Id \\
&+ 2 R_{\mu\nu\rho\sigma} R^{\mu\nu\rho\sigma}\Id + 30 \Omega^S_{\mu\nu} {\Omega^S}^{\mu\nu} \big) \sqrt{|g|} d^4x .
\end{align*}
The curvature $\Omega^S$ of the spin connection is defined as in \eqref{eq:defn_curv}, and its components are $\Omega^S_{\mu\nu} = \Omega^S(\partial_\mu,\partial_\nu)$. The spin curvature $\Omega^S$ is related to the Riemannian curvature tensor by (see, for instance, \cite[p.395]{GVF01})
\begin{align}
\label{eq:spin_curv}
\Omega^S_{\mu\nu} = \frac14 R_{\mu\nu\rho\sigma} \gamma^\rho \gamma^\sigma . 
\end{align}
We use this and the trace identity $\Tr(\gamma^\mu\gamma^\nu\gamma^\rho\gamma^\sigma) = 4(g^{\mu\nu}g^{\rho\sigma} - g^{\mu\rho}g^{\nu\sigma} + g^{\mu\sigma}g^{\nu\rho})$ to calculate the last term of $a_4(\sD^2)$: 
\begin{align}
\label{eq:tr_spin_curv_sq}
\Tr(\Omega^S_{\mu\nu} {\Omega^S}^{\mu\nu}) &= \frac1{16} R_{\mu\nu\rho\sigma}R^{\mu\nu}_{\phantom{\mu\nu}\lambda\kappa}\;\Tr(\gamma^\rho \gamma^\sigma \gamma^\lambda \gamma^\kappa ) \notag\\
&= \frac14 R_{\mu\nu\rho\sigma}R^{\mu\nu}_{\phantom{\mu\nu}\lambda\kappa}\;(g^{\rho\sigma}g^{\lambda\kappa} - g^{\rho\lambda}g^{\sigma\kappa} + g^{\rho\kappa}g^{\sigma\lambda}) = -\frac12 R_{\mu\nu\rho\sigma}R^{\mu\nu\rho\sigma} ,
\end{align}
where on the second line because of the antisymmetry of $R_{\mu\nu\rho\sigma}$ in $\rho$ and $\sigma$, the first term vanishes and the other two terms contribute equally. We thus obtain 
\begin{align}
\label{eq:canon_a4}
a_4(\sD^2) = \frac1{16\pi^2} \frac1{360} \int_M \big( 12 \Delta s +5 s^2 - 8 R_{\mu\nu} R^{\mu\nu} -7 R_{\mu\nu\rho\sigma} R^{\mu\nu\rho\sigma} \big) \sqrt{|g|} d^4x .
\end{align}
We shall rewrite this into a more convenient form. First let us consider the Weyl tensor $C_{\mu\nu\rho\sigma}$, which is the traceless part of the Riemann tensor. The square of the Weyl tensor can be written as
\begin{align}
\label{eq:weyl-sq}
C_{\mu\nu\rho\sigma} C^{\mu\nu\rho\sigma} &=  R_{\mu\nu\rho\sigma}R^{\mu\nu\rho\sigma} - 2R_{\nu\sigma}R^{\nu\sigma} + \frac13 s^2 .
\end{align}
Next, we shall also consider the Pontryagin class $R^*R^*$ given by
\begin{align}
\label{eq:pontryagin}
R^*R^* = s^2-4R_{\mu\nu}R^{\mu\nu}+R_{\mu\nu\rho\sigma}R^{\mu\nu\rho\sigma} .
\end{align}
Using \eqref{eq:weyl-sq,eq:pontryagin} we calculate:
\begin{align*}
-\frac1{20} C_{\mu\nu\rho\sigma} C^{\mu\nu\rho\sigma} + \frac{11}{360}R^*R^* &= -\frac1{20} R_{\mu\nu\rho\sigma}R^{\mu\nu\rho\sigma} +\frac1{10} R_{\nu\sigma}R^{\nu\sigma} - \frac1{60} s^2 \notag\\
&\quad+ \frac{11}{360}R_{\mu\nu\rho\sigma}R^{\mu\nu\rho\sigma} - \frac{44}{360}R_{\nu\sigma}R^{\nu\sigma} + \frac{11}{360}s^2 \notag\\
&= \frac1{360} \big( -7R_{\mu\nu\rho\sigma}R^{\mu\nu\rho\sigma} -8R_{\nu\sigma}R^{\nu\sigma} +5s^2\big) .
\end{align*}
Therefore we can rewrite \eqref{eq:canon_a4} and obtain
\begin{align*}
a_4(\sD^2) &= \frac1{16\pi^2}  \int_M \Big( \frac1{30} \Delta s -\frac1{20} C_{\mu\nu\rho\sigma} C^{\mu\nu\rho\sigma} + \frac{11}{360}R^*R^* \Big) \sqrt{|g|} d^4x .
\end{align*}
Inserting the obtained formulas for $a_0(\sD^2)$, $a_2(\sD^2)$ and $a_4(\sD^2)$ into \eqref{eq:canonical_expansion} proves the proposition.
\end{proof}

\begin{remark}
In general, an expression of the form $as^2 + bR_{\nu\sigma}R^{\nu\sigma} + cR_{\mu\nu\rho\sigma}R^{\mu\nu\rho\sigma}$, for constants $a,b,c\in\R$, can always be rewritten in the form $\alpha s^2 + \beta C_{\mu\nu\rho\sigma} C^{\mu\nu\rho\sigma} + \gamma R^*R^*$, for new constants $\alpha,\beta,\gamma\in\R$. One should note here that the term $s^2$ is not present in the spectral action of the canonical triple, as calculated in \cref{prop:canon_spec_act}. The only higher-order gravitational term that arises is the conformal gravity term $C_{\mu\nu\rho\sigma} C^{\mu\nu\rho\sigma}$. This feature of the spectral action will later allow us in \cref{sec:conf_spec_act} to derive the conformal symmetry of the spectral action. 

Note that alternatively, using only \eqref{eq:pontryagin}, we could also have written
\begin{align*}
a_4(\sD^2) &= \frac1{16\pi^2} \frac1{30} \int_M \big( \Delta s +s^2 - 3 R_{\mu\nu} R^{\mu\nu} -\frac7{12} R^*R^* \big) \sqrt{|g|} d^4x .
\end{align*}
The integral over $\Delta s$ only yields a boundary term, so if the manifold $M$ is compact without boundary, we can discard the term with $\Delta s$. Furthermore, for a $4$-dimensional compact orientable manifold $M$ without boundary, we have the formula
\begin{align*}
\int_M R^*R^* \nu_g = 8\pi^2 \chi(M) , 
\end{align*}
where $\chi(M)$ is Euler's characteristic. Hence the term with $R^*R^*$ only yields a topological contribution, which we will also disregard. From here on, we will therefore consider the Lagrangian
\begin{align}
\label{eq:lagr_M_simple}
\L_M(g_{\mu\nu}) &= \frac{f_4\Lambda^4}{2\pi^2} - \frac{f_2\Lambda^2}{24\pi^2} s - \frac{f(0)}{320\pi^2} C_{\mu\nu\rho\sigma} C^{\mu\nu\rho\sigma} \\
\intertext{or}
\L_M(g_{\mu\nu}) &= \frac{f_4\Lambda^4}{2\pi^2} - \frac{f_2\Lambda^2}{24\pi^2} s + \frac{f(0)}{480\pi^2} \Big( s^2 - 3 R_{\mu\nu} R^{\mu\nu} \Big) .
\end{align}
\end{remark}

\begin{prop} 
\label{prop:acm_spec_act}
The spectral action of the fluctuated Dirac operator of an almost-commutative manifold is given by 
\begin{align*}
\Tr \left(f\Big(\frac{D_A}\Lambda\Big)\right) &\sim \int_M \L(g_{\mu\nu}, B_\mu, \Phi) \sqrt{|g|} d^4x + O(\Lambda^{-1}) ,
\end{align*}
for
\begin{align*}
\L(g_{\mu\nu}, B_\mu, \Phi) := N \L_M(g_{\mu\nu}) + \L_B(B_\mu) + \L_H(g_{\mu\nu}, B_\mu, \Phi) .
\end{align*}
Here $\L_M(g_{\mu\nu})$ is defined in \cref{prop:canon_spec_act}, and $N$ is the dimension of the finite Hilbert space $\mH_F$. $\L_B$ gives the kinetic term of the gauge field and equals
\begin{align*}
\L_B(B_\mu) := \frac{f(0)}{24\pi^2} \Tr(F_{\mu\nu}F^{\mu\nu}) .
\end{align*}
$\L_H$ gives the Higgs Lagrangian including its interactions plus a boundary term given by
\begin{multline}
\L_H(g_{\mu\nu}, B_\mu, \Phi) := -\frac{2f_2\Lambda^2}{4\pi^2} \Tr(\Phi^2) + \frac{f(0)}{8\pi^2} \Tr(\Phi^4) + \frac{f(0)}{24\pi^2} \Delta\big(\Tr(\Phi^2)\big) \\
+ \frac{f(0)}{48\pi^2} s\Tr(\Phi^2) + \frac{f(0)}{8\pi^2} \Tr\big((D_\mu \Phi)(D^\mu \Phi)\big) .
\end{multline}
\end{prop}
\begin{proof}
The proof is very similar to \cref{prop:canon_spec_act}, but we now use the formula for ${D_A}^2$ given by \cref{prop:acm_Dirac_sq}. The trace over the Hilbert space $\mH_F$ yields an overall factor $N:=\Tr(\1_{\mH_F})$, so we have
\begin{align*}
a_0({D_A}^2) = N a_0(\sD^2) .
\end{align*}
The square of the Dirac operator now contains three extra terms. The trace of $\gamma^\mu\gamma_5$ vanishes, since the trace of a product of any odd number of gamma matrices vanishes. Since $\Tr(\gamma^\mu\gamma^\nu) = 4 g^{\mu\nu}$ and $F_{\mu\nu}$ is anti-symmetric, the trace of $\gamma^\mu\gamma^\nu F_{\mu\nu}$ also vanishes. Thus we find that 
\begin{align*}
a_2({D_A}^2) = N a_2(\sD^2) - \frac{1}{4\pi^2} \int_M \Tr(\Phi^2) \sqrt{|g|} d^4x .
\end{align*}
Furthermore we obtain several new terms from the formula for $a_4({D_A}^2)$. First we calculate 
\begin{align*}
\frac{1}{360}\Tr(60sF) = -\frac16 s \left( Ns + 4\Tr(\Phi^2) \right) .
\end{align*}
The next contribution arises from the trace over $F^2$, which (ignoring traceless terms) equals
\begin{multline*}
F^2 = \frac{1}{16}s^2\otimes\1 + \1\otimes\Phi^4 - \frac14 \gamma^\mu\gamma^\nu\gamma^\rho\gamma^\sigma \otimes F_{\mu\nu}F_{\rho\sigma} \\
+ \gamma^\mu\gamma^\nu\otimes(D_\mu\Phi)(D_\nu\Phi) + \frac12 s\otimes\Phi^2 + \text{ traceless terms} .
\end{multline*}
Taking the trace then yields
\begin{align*}
\frac{1}{360} \Tr(180 F^2) = \frac{N}{8}s^2 + 2\Tr(\Phi^4) + \Tr(F_{\mu\nu}F^{\mu\nu}) + 2 \Tr\big((D_\mu\Phi)(D^\mu\Phi)\big) + s \Tr(\Phi^2) . 
\end{align*}
Another contribution arises from $-\Delta F$. Again we can simply ignore the traceless terms and obtain
\begin{align*}
 \frac{1}{360} \Tr(-60 \Delta F) = \frac16 \Delta \left( Ns + 4\Tr(\Phi^2) \right) .
\end{align*}
The final contribution comes from the term $\Omega^E_{\mu\nu}{\Omega^E}^{\mu\nu}$, where the curvature $\Omega^E$ is given by \eqref{eq:acm_curv}. We have
\begin{align*}
\Omega^E_{\mu\nu}{\Omega^E}^{\mu\nu} = \Omega^S_{\mu\nu}{\Omega^S}^{\mu\nu}\otimes\1 - \1\otimes F_{\mu\nu}F^{\mu\nu} + 2i \Omega^S_{\mu\nu}\otimes F^{\mu\nu} .
\end{align*}
Using \eqref{eq:spin_curv}, we find
\begin{align*}
\Tr(\Omega^S_{\mu\nu}) = \frac14 R_{\rho\sigma\mu\nu} \Tr(\gamma^\rho\gamma^\sigma) = \frac14 R_{\rho\sigma\mu\nu} g^{\rho\sigma} = 0 
\end{align*}
by the anti-symmetry of $R_{\rho\sigma\mu\nu}$, so the trace over the cross-terms in $\Omega^E_{\mu\nu}{\Omega^E}^{\mu\nu}$ vanishes. From \eqref{eq:tr_spin_curv_sq} we then obtain
\begin{align*}
\frac{1}{360} \Tr(30 \Omega^E_{\mu\nu}{\Omega^E}^{\mu\nu}) = \frac1{12} \left(-\frac N2 R_{\mu\nu\rho\sigma}R^{\mu\nu\rho\sigma} - 4 \Tr(F_{\mu\nu}F^{\mu\nu})  \right) . 
\end{align*}
Gathering all terms, we obtain
\begin{align*}
a_4(x,{D_A}^2) &= \frac{1}{(4\pi)^2} \frac{1}{360} \Bigg( -48N\Delta s + 20Ns^2 - 8NR_{\mu\nu}R^{\mu\nu} \\
&\quad+ 8NR_{\mu\nu\rho\sigma}R^{\mu\nu\rho\sigma} - 60 s \left( Ns + 4\Tr(\Phi^2) \right) \\
&\quad+ 360 \begin{aligned}[t] \bigg( &\frac{N}{8}s^2 + 2\Tr(\Phi^4) + \Tr(F_{\mu\nu}F^{\mu\nu}) \\&+ 2 \Tr\big((D_\mu\Phi)(D^\mu\Phi)\big) + s \Tr(\Phi^2) \bigg) \end{aligned} \\
&\quad+ 60 \Delta\left( Ns + 4\Tr(\Phi^2) \right) - 30 \left(\frac N2 R_{\mu\nu\rho\sigma}R^{\mu\nu\rho\sigma} + 4 \Tr(F_{\mu\nu}F^{\mu\nu})  \right) \Bigg) \displaybreak[0] \\ 
&= \frac{1}{(4\pi)^2} \frac{1}{360} \Bigg( 12N\Delta s + 5Ns^2 - 8NR_{\mu\nu}R^{\mu\nu} - 7NR_{\mu\nu\rho\sigma}R^{\mu\nu\rho\sigma} \\
&\quad+ 120s\Tr(\Phi^2) + 360\bigg(2\Tr(\Phi^4) + 2 \Tr\big((D_\mu\Phi)(D^\mu\Phi)\big) \bigg) \\
&\quad+ 240 \Delta\left(\Tr(\Phi^2)\right) + 240\Tr(F_{\mu\nu}F^{\mu\nu}) \Bigg) .
\end{align*}
By comparing the first line of the second equality to \eqref{eq:canon_a4}, we see that we can write
\begin{multline*}
a_4(x,{D_A}^2) = N a_4(x,\sD^2) + \frac{1}{4\pi^2}  \Bigg( \frac1{12} s\Tr(\Phi^2) + \frac12\Tr(\Phi^4) \\ + \frac12 \Tr\big((D_\mu\Phi)(D^\mu\Phi)\big) + \frac16 \Delta\left(\Tr(\Phi^2)\right) + \frac16 \Tr(F_{\mu\nu}F^{\mu\nu}) \Bigg) .
\end{multline*}
Inserting these Seeley-DeWitt coefficients into \eqref{eq:canonical_expansion} proves the proposition.
\end{proof} 

\begin{example}
\label{ex:YM_spectral_action}
Let us return to the Yang-Mills manifold $M\times F\Sub{YM}$ of \cref{ex:YM}. We have already seen in \cref{ex:YM_fields} that we have a $PSU(N)$ gauge field $A_\mu$, which acts by the adjoint representation $B_\mu = \ad A_\mu$ on the fermions. There is no Higgs field $\phi$, so $\Phi = D_F = 0$. We can insert these fields into the result of \cref{prop:acm_spec_act}. The dimension of the Hilbert space $\mH_F = M_N(\C)$ is $N^2$. We then find that the Lagrangian of the Yang-Mills manifold is given by
\begin{align*}
\L(g_{\mu\nu}, B_\mu) := N^2 \L_M(g_{\mu\nu}) + \frac{f(0)}{24\pi^2} \L\Sub{YM}(B_\mu)
\end{align*}
Here $\L\Sub{YM}$ is the Yang-Mills Lagrangian given by
\begin{align*}
\L\Sub{YM}(B_\mu) := \Tr(F_{\mu\nu}F^{\mu\nu}) ,
\end{align*}
where $F_{\mu\nu}$ denotes the curvature of $B_\mu$. This was first derived in \cite{CC96,CC97}.
\end{example}

\newpage
\section{Electrodynamics}
\label{chap:ex_ED}

In the previous sections we have described the general framework for the description of gauge theories on almost-commutative manifolds. The present section serves two purposes. First, we describe abelian gauge theories within the framework of noncommutative geometry, which for long was thought impossible.
In \cite[Chapter 9]{Landi97}, a proof is given for the claim that the inner fluctuation $A + JAJ^*$ vanishes for commutative algebras. The proof is based on the claim that the left and right action can be identified, i.e.\ $a=a^0$, for a commutative algebra. Though this claim holds in the case of the canonical triple describing a spin manifold, it need not be true for arbitrary commutative algebras. The almost-commutative manifold given in \cref{sec:U1} provides a counter-example. 

What can be said for a commutative algebra, is that there exist no non-trivial inner automorphisms. It is thus an important insight here that the gauge group $\G(\A)$, as defined in \cref{defn:gauge_group_NCG}, is larger than the group of inner automorphisms, so that a commutative algebra may still lead to a non-trivial gauge group. In fact, we will show that our example given below describes an abelian $U(1)$ gauge theory. 

Second, in \cref{sec:ED} we will show how this example can be modified to provide a description of one of the simplest examples of a gauge field theory in physics, namely electrodynamics. Because of its simplicity, it helps in gaining an understanding of the formulation of gauge theories in terms of almost-commutative manifolds, and it provides a first stepping stone towards the derivation of the Standard Model from noncommutative geometry in \cref{chap:ex_SM}.

\subsection{The two-point space} 
\label{sec:U1}

\subsubsection{A two-point space}
\label{sec:two-point}

In this section we will discuss one of the simplest possible spaces, namely the two-point space $X = \{x,y\}$. A complex function on this space is simply determined by two complex numbers. The algebra of functions on $X$ is then given by $C(X) = \C^2$. Let us construct an \emph{even} finite space $F\Sub{X}$, corresponding to the two-point space $X$, given by (see \cref{sec:acm})
\begin{align*}
F\Sub{X} := \left( C(X), \mH_F, D_F, \gamma_F \right) .
\end{align*}
We require that the action of $C(X)$ on the finite-dimensional Hilbert space $\mH_F$ is faithful, which implies that $\mH_F$ must be at least $2$-dimensional. For now we will restrict ourselves to the simplest case, and thus we will take $\mH_F = \C^2$. We use the $\Z_2$-grading $\gamma_F$ to decompose $\mH_F = \mH_F^+ \oplus \mH_F^- = \C\oplus\C$ into the two eigenspaces $\mH_F^\pm = \{\psi\in\mH_F \mid \gamma_F\psi = \pm\psi \}$. Hence, we can decompose accordingly
\begin{align*}
\gamma_F = \mattwo{1}{0}{0}{-1} .
\end{align*}
Because the grading must satisfy the relations $[\gamma_F,a]=0$ and $D_F\gamma_F=-\gamma_FD_F$, the hermitian Dirac operator $D_F$ must be off-diagonal and the action of an element $a\in\A_F$ on $\psi\in\mH_F$ can be written as 
\begin{align}
\label{eq:rep_U1}
a\psi = \mattwo{a_+}{0}{0}{a_-} \vectwo{\psi_+}{\psi_-} .
\end{align}
Thus, the \emph{even finite space} $F\Sub{X}$ we will study in this section is given by
\begin{align}
\label{eq:triple_U1}
\left(\A_F, \mH_F, D_F, \gamma_F \right) = \left( \C^2, \C^2, \mattwo{0}{t}{\bar t}{0}, \mattwo{1}{0}{0}{-1} \right) ,
\end{align}
where $D_F$ is determined by some complex parameter $t\in\C$, and where the action of $\A_F$ on $\mH_F$ given by \eqref{eq:rep_U1}.

Next, we want to introduce a \emph{real structure} (or conjugation operator) on the finite space $F\Sub{X}$, so we must give an antilinear isomorphism $J_F$ on $\C^2$ which satisfies the conditions of \cref{defn:real_structure}. 

\begin{prop}
\label{prop:zero_D_F}
The finite space $F\Sub{X}$ for the two-point space, given by \eqref{eq:triple_U1}, can only have a real structure $J_F$ if $D_F = 0$. 
\end{prop}
\begin{proof}
We must have ${J_F}^2 = \varepsilon$ and $J_F\gamma_F = \varepsilon''\gamma_FJ_F$, and we shall consider all possible (even) KO-dimensions separately. Thus, we apply Lemma \ref{lem:class_real} to the finite space $F\Sub{X}$ given above and, for each even KO-dimension, also impose the relations $[a,b^0]=0$ and $\big[[D_F,a],b^0\big]=0$. This gives:
\begin{description}
\item[KO-dimension $0$] \mbox{}\\
We have $J_F = \mattwo{j_+}{0}{0}{j_-} C$ for $j_\pm\in U(1)$. For $b=\mattwo{b_+}{0}{0}{b_-}$ we then obtain 
$$
b^0 = \mattwo{j_+b_+\bar{j_+}}{0}{0}{j_-b_-\bar{j_-}} = b ,
$$
and see that this indeed commutes with the left action of $a\in\C^2$. Next, we check the order one condition
\begin{align*}
0 = \big[[D_F,a],b^0\big] &= 
 (a_+-a_-)(b_+-b_-)D_F .
\end{align*}
Since this must hold for all $a,b\in\C^2$, we conclude that we must require $D_F = 0$. 

\item[KO-dimension $2$] \mbox{}\\
We have $J_F = \mattwo{0}{j}{-j}{0} C$ for $j\in U(1)$. We now obtain 
$$
b^0 = \mattwo{jb_-\bar{j}}{0}{0}{jb_+\bar{j}} = \mattwo{b_-}{0}{0}{b_+} ,
$$
and see that this indeed commutes with the left action of $a\in\C^2$. Next, we check the order one condition
\begin{align*}
0 = \big[[D_F,a],b^0\big] &= 
(a_+-a_-)(b_--b_+)D_F .
\end{align*}
Again we conclude that we must require $D_F = 0$. 

\item[KO-dimension $4$] \mbox{}\\
We have $J_F$ 
of the same form as in KO-dimension 0, but now with $j_\pm = -j_\pm^T\in U(1)$. This implies that $j_\pm = 0$, so the given finite space cannot have a real structure in KO-dimension $4$.

\item[KO-dimension $6$] \mbox{}\\
We have $J_F = \mattwo{0}{j}{j}{0} C$ for $j\in U(1)$. We again obtain
$$
b^0 = \mattwo{jb_-\bar{j}}{0}{0}{jb_+\bar{j}} = \mattwo{b_-}{0}{0}{b_+} ,
$$
just as for KO-dimension $2$. Hence again the commutation rules are only satisfied for $D_F = 0$. \qedhere
\end{description}
\end{proof}

\subsubsection{The product space}
\label{sec:product_space_U1}

Let $M$ be a compact $4$-dimensional Riemannian spin manifold. We will now consider the almost-commutative manifold $M\times F\Sub{X}$ given by the product of $M$ with the even finite space $F\Sub{X}$ corresponding to the two-point space, as given in \eqref{eq:triple_U1}. Thus we consider the almost-commutative manifold given by the data
\begin{align*}
M\times F\Sub{X} := \Big( C^\infty(M,\C^2), L^2(M,S)\otimes\C^2, \sD\otimes\1, \gamma_5\otimes\gamma_F, J_M\otimes J_F \Big) ,
\end{align*}
where we still need to make a choice for $J_F$. The algebra of this almost-commutative manifold is given by $C^\infty(M,\C^2) \simeq C^\infty(M)\oplus C^\infty(M)$. By the Gelfand-Naimark theorem (see, for instance, \cite[Theorem 1.4]{GVF01}), this algebra corresponds to the space $N := M\times X \simeq M \sqcup M$, which consists of the disjoint union of two identical copies of the space $M$, and we can write $C^\infty(N) = C^\infty(M)\oplus C^\infty(M)$. We can also decompose the total Hilbert space as $\mH = L^2(M,S) \oplus L^2(M,S)$. For $a,b\in C^\infty(M)$ and $\psi,\phi\in L^2(M,S)$, an element $(a,b)\in C^\infty(N)$ then simply acts on $(\psi,\phi)\in\mH$ as $(a,b) (\psi,\phi) = (a\psi,b\phi)$. 

\paragraph{Distances}

In \eqref{eq:distance_ACM} we have given a formula for a generalized notion of distance on almost-commutative manifolds. We can straightforwardly restrict this formula to our finite space $F\Sub{X}$, and we write
\begin{align*}
d_{D_F}(x,y) = \sup \left\{ |a(x) - a(y)| \colon a\in\A_F, \|[D_F,a]\|\leq1 \right\} .
\end{align*}
Note that we now have only two distinct points $x$ and $y$ in the space $X$, and we shall calculate the distance between these points. An element $a\in\C^2=C(X)$ is specified by two complex numbers $a(x)$ and $a(y)$, so the commutator with $D_F$ becomes
\begin{align*}
[D_F,a] &= \mattwo{0}{t}{\bar t}{0}\mattwo{a(x)}{0}{0}{a(y)} - \mattwo{a(x)}{0}{0}{a(y)}\mattwo{0}{t}{\bar t}{0} \\
&= \big(a(y)-a(x)\big) \mattwo{0}{t}{-\bar t}{0} .
\end{align*}
The norm of this commutator is given by $|a(y)-a(x)|\,|t|$, so $\|[D_F,a]\| \leq 1$ implies $|a(y)-a(x)| \leq \frac{1}{|t|}$. We thus obtain that the distance between the two points $x$ and $y$ is given by
\begin{align*}
d_{D_F}(x,y) = \frac{1}{|t|} .
\end{align*}
If there is a real structure $J_F$, we have $t=0$ by \cref{prop:zero_D_F}, so then the distance between the two points becomes infinite. 

Let $p$ be a point in $M$, and write $(p,x)$ or $(p,y)$ for the two corresponding points in $N=M\times X$. A function $a\in C^\infty(N)$ is then determined by two functions $a_x,a_y\in C^\infty(M)$, given by $a_x(p) := a(p,x)$ and $a_y(p) := a(p,y)$. Now consider the distance function on $N$ given by 
\begin{align*}
d_{\sD\otimes\1}(n_1,n_2) = \sup \left\{ |a(n_1) - a(n_2)| \colon a\in\A, \|[\sD\otimes\1,a]\|\leq1 \right\} .
\end{align*}
If $n_1$ and $n_2$ are points in the same copy of $M$, for instance if $n_1=(p,x)$ and $n_2=(q,x)$ for points $p,q\in M$, then their distance is determined by $|a_x(p) - a_x(q)|$, for functions $a_x\in C^\infty(M)$ for which $\|[\sD,a_x]\|\leq1$. Thus, in this case we obtain that we recover the geodesic distance on $M$, i.e.\ $d_{\sD\otimes\1}(n_1,n_2) = d_g(p,q)$. 

However, if $n_1$ and $n_2$ are points in a different copy of $M$, for instance if $n_1=(p,x)$ and $n_2=(q,y)$, then their distance is determined by $|a_x(p) - a_y(q)|$ for two functions $a_x,a_y\in C^\infty(M)$, such that $\|[\sD,a_x]\|\leq1$ and $\|[\sD,a_y]\|\leq1$. These latter requirements however yield no restriction on $|a_x(p) - a_y(q)|$, so in this case the distance between $n_1$ and $n_2$ is infinite. We thus find that the space $N$ is given by two disjoint copies of $M$, which are separated by an infinite distance. 

It should be noted that the only way in which the distance between the two copies of $M$ could have been finite, is when the commutator $[D_F,a]$ would be nonzero. This same commutator generates the Higgs field $\phi$ of \eqref{eq:fluc_higgs}, hence the finiteness of the distance is related to the existence of a Higgs field.

\subsubsection{\texorpdfstring{$U(1)$}{U(1)} gauge theory}

We would now like to determine the gauge theory that corresponds to the almost-commutative manifold $M\times F\Sub{X}$. The gauge group $\G(\A)$ as defined in \cref{defn:gauge_group_NCG} is given by the quotient $U(\A) / U(\til\A_J)$, so if we wish to obtain a nontrivial gauge group, we need to choose $J$ such that $U(\til\A_J) \neq U(\A)$. By looking at the form of $J_F$ for the different (even) KO-dimensions, as given in \cref{sec:two-point}, we conclude that we need to have KO-dimension $2$ or $6$. 
It has been observed independently by Barrett \cite{Barrett06} and Connes \cite{Connes06} that in the noncommutative description of the Standard Model, the correct signature for the internal space should be KO-dimension $6$. Therefore, we choose to work in KO-dimension 6 here as well. 
%
%
The almost-commutative manifold $M\times F\Sub{X}$ then has KO-dimension $6+4\mod 8 = 2$. This means that we can use \cref{defn:act_funct} to calculate the fermionic action. Therefore, we will consider the finite space $F\Sub{X}$ given by the data 
\begin{align*}
F\Sub{X} := \left( \C^2, \C^2, 0, \gamma_F=\mattwo{1}{0}{0}{-1}, J_F=\mattwo{0}{C}{C}{0} \right) ,
\end{align*}
which defines a real even finite space of KO-dimension $6$. Now, let us derive the gauge group.

\begin{prop}
\label{prop:gauge_group_U1}
The gauge group $\G(\A_F)$ of the two-point space is given by $U(1)$. 
\end{prop}
\begin{proof}
First, note that $U(\A_F) = U(1) \times U(1)$. We will show that $U( (\til\A_F)_{J_F} ) \equiv U(\A_F) \cap (\til\A_F)_{J_F} \simeq U(1)$ so that the quotient $\G(\A_F) \simeq U(1)$ as claimed. Indeed, for $a \in \C^2$ to be in $(\til \A_F)_{J_F}$ it has to satisfy $J_F a^* J_F =a $. Since 
$$
 J_Fa^*J_F^* = \mattwo{0}{C}{C}{0} \mattwo{\bar a_1}{0}{0}{\bar a_2} \mattwo{0}{C}{C}{0} = \mattwo{a_2}{0}{0}{a_1} ,
$$
this is the case if and only if $a_1 = a_2$. Thus, $(\til \A_F)_{J_F} \simeq \C$ whose unitary elements form the group $U(1)$, contained in $U(\A_F)$ as the diagonal subgroup.
\end{proof}

In \cref{prop:acm_spec_act} we have calculated the spectral action of an almost-commu\-ta\-tive manifold. Before we can apply this to the two-point space, we need to find the exact form of the field $B_\mu$. Since we have $(\til\A_F)_{J_F} \simeq \C$, we find that $\h_F = \lu\big((\til\A_F)_{J_F}\big) \simeq i\R$. From \cref{prop:unimod,eq:gauge_field} we then see that the gauge field $A_\mu(x) \in i\g_F = i \big(\lu(\A_F) / (i\R)\big) = i\,\su(\A_F) \simeq \R$ becomes traceless. 

Let us consider in detail how we obtain this $U(1)$ gauge field. An arbitrary hermitian field of the form $A_\mu=-ia\partial_\mu b$ would be given by two $U(1)$ gauge fields $X^1_\mu,X^2_\mu\in C^\infty(M,\R)$. However, because $A_\mu$ only appears in the combination $A_\mu - J_FA_\mu J_F^{-1}$, we obtain
\begin{align*}
B_\mu = A_\mu - J_FA_\mu J_F^{-1} = \mattwo{X^1_\mu}{0}{0}{X^2_\mu} - \mattwo{X^2_\mu}{0}{0}{X^1_\mu} =: \mattwo{Y_\mu}{0}{0}{-Y_\mu} = Y_\mu\otimes\gamma_F ,
\end{align*}
where we have defined the $U(1)$ gauge field $Y_\mu := X^1_\mu-X^2_\mu \in C^\infty(M,\R) = C^\infty(M,i\,\lu(1))$. Thus, the fact that we only have the combination $A+JAJ^*$ effectively identifies the $U(1)$ gauge fields on the two copies of $M$, so that $A_\mu$ is determined by only one $U(1)$ gauge field. This ensures that we can take the quotient of the Lie algebra $\lu(\A_F)$ with $\h_F$. We can then write 
\begin{align*}
A_\mu = \frac12 \mattwo{Y_\mu}{0}{0}{-Y_\mu} = \frac12 Y_\mu\otimes \gamma_F ,
\end{align*}
which yields the same result:
\begin{align}
\label{eq:u1_field}
B_\mu = A_\mu - J_FA_\mu J_F^{-1} = 2A_\mu = Y_\mu\otimes\gamma_F .
\end{align}
We summarize: 
\begin{prop}
\label{prop:gauge-field}
The inner fluctuations of the almost-commutative manifold $M \times F_X$ described above are parametrized by a $U(1)$-gauge field $Y_\mu$ as
$$
D \mapsto D' = D + \gamma^\mu Y_\mu \otimes \gamma_F.
$$
The action of the gauge group $\G(\A) \simeq C^\infty(M, U(1))$ on $D'$, as in \eqref{eq:gauge_transf_ACM}, is implemented by
$$
Y_\mu \mapsto Y_\mu - i u \partial_\mu u^*, \qquad (u \in \G(\A)). 
$$ 
\end{prop}

\begin{remark}
In \cite{BDDD11} it was observed that a $U(1)$ gauge theory can also be described by a spectral triple based on the algebra $\A_F = \C$, but with a real representation on the Hilbert space $\mH_F = \C^2$, which leads to the same action of the gauge group $U(1)$. 
\end{remark}

\subsection{Electrodynamics}
\label{sec:ED}

Inspired by the previous section, which shows that one can use the framework of noncommutative geometry to describe a gauge theory with the abelian gauge group $U(1)$, we shall now attempt to describe the full theory of electrodynamics. Our approach provides a unified description of gravity and electromagnetism, albeit at the classical level. Earlier attempts at such a unified description have originated from the work of Kaluza \cite{Kaluza21} and Klein \cite{Klein26} in the 1920's. In their approach, a new (compact) fifth dimension is added to the $4$-dimensional spacetime $M$. The additional components in the $5$-dimensional metric tensor are then identified with the electromagnetic gauge potential. Subsequently it can be shown that the Einstein equations of the $5$-dimensional spacetime can be reduced to the Einstein equations plus Maxwell equations on $4$-dimensional spacetime. We note that our approach via almost-commutative manifolds is fundamentally different from Kaluza-Klein theory. Instead of adding new dimensions, we expand our $4$-dimensional manifold $M$ by a discrete internal two-point space $X$. Thus, we consider the new space $N = M\times X = M \sqcup M$ consisting of two disjoint copies of $M$. In our case, the gauge group $U(1)$ does not arise from an additional compact dimension, but instead from the algebra of functions on the discrete space $X$. 

We have seen that the almost-commutative manifold $M\times F\Sub{X}$ describes a gauge theory with local gauge group $U(1)$, where the inner fluctuations of the Dirac operator provide the $U(1)$ gauge field $Y_\mu$. There appear to be two problems if one wishes to use this model for a description of (classical) electrodynamics. First, by \cref{prop:zero_D_F}, the finite Dirac operator $D_F$ must vanish. However, we want our electrons to be massive, and for this purpose we need a finite Dirac operator that is non-zero. 

Second, from \cite[Ch.7, \S5.2]{Coleman85}, we find the usual Euclidean action for a free Dirac field to be of the form
\begin{align}
\label{eq:coleman}
S = - \int i \bar\psi (\gamma^\mu\partial_\mu - m) \psi d^4x ,
\end{align}
where the fields $\psi$ and $\bar\psi$ (as usual for a Euclidean field theory) must be considered as \emph{totally independent variables}. Thus, we require that the fermionic action $S_f$ should also yield two \emph{independent} Dirac spinors. Let us write $\left\{e, \bar e\right\}$ for the set of orthonormal basis vectors of $\mH_F$, where $e$ is the basis element of $\mH_F^+$ and $\bar e$ of $\mH_F^-$. Note that on this basis, we have $J_Fe=\bar e$, $J_F\bar e=e$, $\gamma_Fe=e$ and $\gamma_F\bar e=-\bar e$. The total Hilbert space $\mH$ is given by $L^2(M,S)\otimes\mH_F$. Since we can also decompose $L^2(M,S) = L^2(M,S)^+ \oplus L^2(M,S)^-$ by means of $\gamma_5$, we obtain that the positive eigenspace $\mH^+$ of $\gamma=\gamma_5\otimes\gamma_F$ is given by
\begin{align*}
\mH^+ = L^2(M,S)^+\otimes\mH_F^+ \oplus L^2(M,S)^-\otimes\mH_F^- .
\end{align*}
An arbitrary vector $\xi\in\mH^+$ can then uniquely be written as 
\begin{align*}
\xi =  \psi_L \otimes e + \psi_R \otimes \bar e ,
\end{align*}
for two Weyl spinors $\psi_L\in L^2(M,S)^+$ and $\psi_R\in L^2(M,S)^-$. One should note here that this vector $\xi$ is completely determined by only one Dirac spinor $\psi := \psi_L + \psi_R$, instead of the required two independent spinors. Thus, the restrictions that are incorporated into the fermionic action of \cref{defn:act_funct} are such that the finite space $F\Sub{X}$ is in fact too restricted.

\subsubsection{The finite space}

It turns out that both problems sketched above can be simply solved by doubling our finite Hilbert space. Hence, we will start with the same algebra $C^\infty(M,\C^2)$ that corresponds to the space $N = M\times X \simeq M\sqcup M$. The finite Hilbert space will now be used to describe four particles, namely both the left-handed and the right-handed electrons and positrons. We will choose the orthonormal basis $\left\{e_R, e_L, \bar{e_R}, \bar{e_L}\right\}$ for $\mH_F = \C^4$, with respect to the standard inner product. The subscript $L$ denotes left-handed particles, and the subscript $R$ denotes right-handed particles, and we have $\gamma_F e_L = e_L$ and $\gamma_F e_R = -e_R$. 

We will choose $J_F$ such that it interchanges particles with their antiparticles, so $J_Fe_R = \bar{e_R}$ and $J_Fe_L = \bar{e_L}$. We will again choose the real structure such that is has KO-dimension $6$, so we have $J_F^2 = \1$ and $J_F\gamma_F = -\gamma_FJ_F$. This last relation implies that the element $\bar{e_R}$ is left-handed and $\bar{e_L}$ is right-handed. Hence, the grading $\gamma_F$ and the conjugation operator $J_F$ are given by
\begin{align*}
\gamma_F &= \matfour{-1&0&0&0}{0&1&0&0}{0&0&1&0}{0&0&0&-1} , & J_F &= \matfour{0&0&C&0}{0&0&0&C}{C&0&0&0}{0&C&0&0} .
\end{align*}

The grading $\gamma_F$ decomposes the Hilbert space $\mH_F$ into $\mH_L\oplus\mH_R$, where the bases of $\mH_L$ and $\mH_R$ are given by $\{e_L,\bar{e_R}\}$ and $\{e_R,\bar{e_L}\}$, respectively. We can also decompose the Hilbert space into $\mH_e\oplus\mH_{\bar e}$, where $\mH_e$ contains the electrons $\{e_R,e_L\}$, and $\mH_{\bar e}$ contains the positrons $\{\bar{e_R},\bar{e_L}\}$. 

The elements $a\in\A_F=\C^2$ now act on the basis $\left\{e_R, e_L, \bar{e_R}, \bar{e_L}\right\}$ as
\begin{align}
\label{eq:rep_ED}
a = \vectwo{a_1}{a_2} \rightarrow \matfour{a_1&0&0&0}{0&a_1&0&0}{0&0&a_2&0}{0&0&0&a_2} .
\end{align}
Note that this action commutes with the grading, as it should. We can also easily check that $[a,b^0] = 0$ for $b^0 := J_Fb^*J_F^*$, since both the left and the right action are given by diagonal matrices. For now, we will still take $D_F = 0$, and hence the order one condition is trivially satisfied. We have now obtained the following result:
\begin{prop}
The finite space 
$$
F\Sub{ED} := (\C^2, \C^4, 0, \gamma_F, J_F)
$$
as given above defines a real even finite space of KO-dimension $6$. 
\end{prop} 

\subsubsection{A non-trivial finite Dirac operator}

Let us now consider the possibilities for adding a non-zero Dirac operator to the finite space $F\Sub{ED}$. Since $D_F\gamma_F = -\gamma_FD_F$, the Dirac operator obtains the form 
\begin{align*}
D_F = \matfour{0&d_1&d_2&0}{\bar d_1&0&0&d_3}{\bar d_2&0&0&d_4}{0&\bar d_3&\bar d_4&0} .
\end{align*}
Next, we impose the commutation relation $D_FJ_F = J_FD_F$, which implies $d_1 = \bar d_4$.
For the order one condition, we calculate
\begin{align*}
[D_F,a] &= 
(a_1-a_2) \matfour{0&0&-d_2&0}{0&0&0&-d_3}{\bar d_2&0&0&0}{0&\bar d_3&0&0} .
\end{align*}
which then imposes the condition
\begin{align*}
0 = \big[[D_F,a],b^0\big] 
&= (a_1-a_2)(b_2-b_1) \matfour{0&0&d_2&0}{0&0&0&d_3}{\bar d_2&0&0&0}{0&\bar d_3&0&0} . 
\end{align*}
Since this must hold for all $a,b\in\C^2$, we must require that $d_2=d_3=0$. To conclude, the Dirac operator only depends on one complex parameter and is given by
\begin{align}
\label{eq:Dirac}
D_F = \matfour{0&d&0&0}{\bar d&0&0&0}{0&0&0&\bar d}{0&0&d&0} .
\end{align}
From here on, we will consider the finite space $F\Sub{ED}$ given by
\begin{align*}
F\Sub{ED} := (\C^2, \C^4, D_F, \gamma_F, J_F) .
\end{align*}

\subsubsection{The almost-commutative manifold}

By taking the product with the canonical triple, our almost-commutative manifold (of KO-dimension $2$) under consideration is given by
\begin{align*}
M\times F\Sub{ED} := \left( C^\infty(M,\C^2), L^2(M,S)\otimes\C^4, \sD\otimes\I + \gamma_5\otimes D_F, \gamma_5\otimes\gamma_F, J_M\otimes J_F \right) .
\end{align*}
As in \cref{sec:U1}, the algebra decomposes as $C^\infty(M,\C^2) = C^\infty(M)\oplus C^\infty(M)$, and we now decompose the Hilbert space as $\mH = (L^2(M,S)\otimes\mH_e)\oplus(L^2(M,S)\otimes\mH_{\bar e})$. The action of the algebra on $\mH$, given by \eqref{eq:rep_ED}, is then such that one component of the algebra acts on the electron fields $L^2(M,S)\otimes\mH_e$, and the other component acts on the positron fields $L^2(M,S)\otimes\mH_{\bar e}$. 

The derivation of the gauge group for $F\Sub{ED}$ is exactly the same as in \cref{prop:gauge_group_U1}, so again we have the finite gauge group $\G(\A_F) \simeq U(1)$. The field $B_\mu := A_\mu - J_FA_\mu J_F^*$ now takes the form
\begin{align}
\label{eq:gauge_field_ED}
B_\mu = \matfour{Y_\mu&0&0&0}{0&Y_\mu&0&0}{0&0&-Y_\mu&0}{0&0&0&-Y_\mu} \qquad\text{for } Y_\mu(x) \in \R .
\end{align}
Thus, we again obtain a single $U(1)$ gauge field $Y_\mu$, carrying an action of the gauge group $\G(\A) \simeq C^\infty(M, U(1))$ (as in Proposition \ref{prop:gauge-field}). 

As mentioned before, our space $N$ consists of two copies of $M$, and the distance between these two copies is infinite (cf.\ \cref{sec:product_space_U1}). Now, we have introduced a non-zero Dirac operator, but it commutes with the algebra, i.e.\ $[D_F,a]=0$ for all $a\in\A$. Therefore, the distance between the two copies of $M$ is still infinite. 

To summarize, the $U(1)$ gauge theory arises from the geometric space $N=M \sqcup M$ as follows. On one copy of $M$, we have the vector bundle $S\otimes(M\times\mH_e)$, and on the other copy the vector bundle $S\otimes(M\times\mH_{\bar e})$. The gauge fields on each copy of $M$ are identified with each other. The electrons $e$ and positrons $\bar e$ are then both coupled to the same gauge field, and as such the gauge field provides an interaction between electrons and positrons. Note the different role that is played by the internal space with Kaluza-Klein theories.

\subsubsection{The Lagrangian}

We are now ready to explicitly calculate the Lagrangian that corresponds to the almost-commutative manifold $M\times F\Sub{ED}$, and we will show that this yields the usual Lagrangian for electrodynamics (on a curved background manifold), as well as a purely gravitational Lagrangian. The action functional for an almost-commutative manifold, as defined in \cref{defn:act_funct}, consists of the spectral action $S_b$ and the fermionic action $S_f$, which we will calculate separately. 

\paragraph{The spectral action}

The spectral action for an almost-commutative manifold has been calculated in \cref{prop:acm_spec_act}, and we only need to insert the fields $B_\mu$ (given by \eqref{eq:gauge_field_ED}) and $\Phi = D_F$. We obtain the following result:

\begin{prop}
\label{prop:spec_act_ED}
The spectral action of the almost-commutative manifold
\begin{align*}
M\times F\Sub{ED} = \left( C^\infty(M,\C^2), L^2(M,S)\otimes\C^4, \sD\otimes\1 + \gamma_5\otimes D_F, \gamma_5\otimes\gamma_F, J_M\otimes J_F \right) 
\end{align*}
is given by
\begin{align*}
\Tr \left(f\Big(\frac {D_A}\Lambda\Big)\right) &\sim \int_M \L(g_{\mu\nu}, Y_\mu) \sqrt{|g|} d^4x + O(\Lambda^{-1}) ,
\end{align*}
for the Lagrangian
\begin{align*}
\L(g_{\mu\nu}, Y_\mu) := 4\L_M(g_{\mu\nu}) + \L_Y(Y_\mu) + \L_H(g_{\mu\nu},d) .
\end{align*}
Here $\L_M(g_{\mu\nu})$ is defined in \cref{prop:canon_spec_act}. The term $\L_Y$ gives the kinetic term of the $U(1)$ gauge field $Y_\mu$ and equals 
\begin{align*}
\L_Y(Y_\mu) := \frac{f(0)}{6\pi^2} \F_{\mu\nu}\F^{\mu\nu} ,
\end{align*}
where we have defined the curvature $\F_{\mu\nu}$ of the field $Y_\mu$ as $\F_{\mu\nu} := \partial_\mu Y_\nu - \partial_\nu Y_\mu$. The Higgs potential $\L_H$ (ignoring the boundary term) only gives two constant terms which add to the cosmological constant, plus an extra contribution to the Einstein-Hilbert action:
\begin{align*}
\L_H(g_{\mu\nu}) := -\frac{2f_2\Lambda^2}{\pi^2} |d|^2 + \frac{f(0)}{2\pi^2} |d|^4 + \frac{f(0)}{12\pi^2} s|d|^2 .
\end{align*}
\end{prop}
\begin{proof}
The trace over the Hilbert space $\C^4$ yields an overall factor $N=4$. The field $B_\mu$ is given by \eqref{eq:gauge_field_ED}, and we obtain $\Tr(F_{\mu\nu}F^{\mu\nu}) = 4\F_{\mu\nu}\F^{\mu\nu}$. Inserting this into \cref{prop:acm_spec_act} provides the Lagrangian $\L_Y$. In addition, we have $\Phi^2 = {D_F}^2 = |d|^2$, and the Higgs Lagrangian $\L_H$ only yields extra contributions to the cosmological constant and the Einstein-Hilbert action. 
\end{proof} 

\subsubsection{The fermionic action}
We have written the set of basis vectors of $\mH_F$ as $\left\{e_R, e_L, \bar{e_R}, \bar{e_L}\right\}$, and the subspaces $\mH_F^+$ and $\mH_F^-$ are spanned by $\left\{e_L, \bar{e_R}\right\}$ and $\left\{e_R, \bar{e_L}\right\}$, respectively. The total Hilbert space $\mH$ is given by $L^2(M,S)\otimes\mH_F$. Since we can also decompose $L^2(M,S) = L^2(M,S)^+ \oplus L^2(M,S)^-$ by means of $\gamma_5$, we obtain
\begin{align*}
\mH^+ = L^2(M,S)^+\otimes\mH_F^+ \oplus L^2(M,S)^-\otimes\mH_F^- .
\end{align*}
A spinor $\psi\in L^2(M,S)$ can be decomposed as $\psi = \psi_L + \psi_R$. Each subspace $\mH_F^\pm$ is now spanned by two basis vectors. A generic element of the tensor product of two spaces consists of sums of tensor products, so an arbitrary vector $\xi\in\mH^+$ can uniquely be written as 
\begin{align}
\label{eq:right_fermion_ED}
\xi = \chi_R\otimes e_R + \chi_L\otimes e_L + \psi_L\otimes\bar{e_R} + \psi_R\otimes\bar{e_L} ,
\end{align}
for Weyl spinors $\chi_L,\psi_L\in L^2(M,S)^+$ and $\chi_R,\psi_R\in L^2(M,S)^-$. Note that this vector $\xi\in\mH^+$ is now completely determined by two Dirac spinors $\chi := \chi_L + \chi_R$ and $\psi := \psi_L + \psi_R$. 

\begin{prop}
\label{prop:fermion_act_ED}
The fermionic action of the almost-commutative manifold
\begin{align*}
M\times F\Sub{ED} = \left( C^\infty(M,\C^2), L^2(M,S)\otimes\C^4, \sD\otimes\I + \gamma_5\otimes D_F, \gamma_5\otimes\gamma_F, J_M\otimes J_F \right) 
\end{align*}
is given by
\begin{align*}
S_f = -i\big\langle J_M\til\chi,\gamma^\mu(\nabla^S_\mu - i Y_\mu)\til\psi\big\rangle + \langle J_M\til\chi_L,\bar d\til\psi_L\rangle - \langle J_M\til\chi_R,d\til\psi_R\rangle .
\end{align*}
\end{prop}
\begin{proof}
The fluctuated Dirac operator is given by 
\begin{align*}
D_A = \sD\otimes\I + \gamma^\mu\otimes B_\mu + \gamma_5\otimes D_F .
\end{align*}
An arbitrary $\xi\in\mH^+$ has the form of \eqref{eq:right_fermion_ED}, and then we obtain the following expressions:
\begin{align*}
J\xi &= J_M\chi_R\otimes\bar{e_R} + J_M\chi_L\otimes\bar{e_L} + J_M\psi_L\otimes e_R + J_M\psi_R\otimes e_L, \\
(\sD\otimes\I)\xi &= \sD\chi_R\otimes e_R + \sD\chi_L\otimes e_L + \sD\psi_L\otimes\bar{e_R} + \sD\psi_R\otimes\bar{e_L} , \\
(\gamma^\mu\otimes B_\mu)\xi &= \gamma^\mu\chi_R\otimes Y_\mu e_R + \gamma^\mu\chi_L\otimes Y_\mu e_L - \gamma^\mu\psi_L\otimes Y_\mu\bar{e_R} - \gamma^\mu\psi_R\otimes Y_\mu\bar{e_L} , \\
(\gamma_5\otimes D_F)\xi &= \gamma_5\chi_L\otimes \bar de_R + \gamma_5\chi_R\otimes d e_L + \gamma_5\psi_R\otimes d \bar{e_R} + \gamma_5\psi_L\otimes \bar d\bar{e_L} .
\end{align*}
We decompose the fermionic action into the three terms
\begin{align*}
\frac12 \langle J\til\xi,D_A\til\xi\rangle &= \frac12 \langle J\til\xi,(\sD\otimes\I)\til\xi\rangle + \frac12 \langle J\til\xi,(\gamma^\mu\otimes B_\mu)\til\xi\rangle + \frac12 \langle J\til\xi,(\gamma_5\otimes D_F)\til\xi\rangle ,
\end{align*}
and then continue to calculate each term separately. The first term is given by 
\begin{multline*}
\frac12 \langle J\til\xi,(\sD\otimes\I)\til\xi\rangle = \frac12 \langle J_M\til\chi_R,\sD\til\psi_L\rangle + \frac12 \langle J_M\til\chi_L,\sD\til\psi_R\rangle \\
+ \frac12 \langle J_M\til\psi_L,\sD\til\chi_R\rangle + \frac12 \langle J_M\til\psi_R,\sD\til\chi_L\rangle .
\end{multline*}
Using the fact that $\sD$ changes the chirality of a Weyl spinor, and that the subspaces $L^2(M,S)^+$ and $L^2(M,S)^-$ are orthogonal, we can rewrite this term as
\begin{align*}
\frac12 \langle J\til\xi,(\sD\otimes\I)\til\xi\rangle = \frac12 \langle J_M\til\chi,\sD\til\psi\rangle + \frac12 \langle J_M\til\psi,\sD\til\chi\rangle .
\end{align*}
Using the symmetry of the form $\langle J_M\til\chi,\sD\til\psi\rangle$, we obtain
\begin{align*}
\frac12 \langle J\til\xi,(\sD\otimes\I)\til\xi\rangle = \langle J_M\til\chi,\sD\til\psi\rangle = -i \langle J_M\til\chi,\gamma^\mu\nabla^S_\mu\til\psi\rangle .
\end{align*}
Note that the factor $\frac12$ has now disappeared in the result, and this is the reason why this factor is included in the definition of the fermionic action. The second term is given by
\begin{multline*}
\frac12 \langle J\til\xi,(\gamma^\mu\otimes B_\mu)\til\xi\rangle = - \frac12 \langle J_M\til\chi_R,\gamma^\mu Y_\mu\til\psi_L\rangle - \frac12 \langle J_M\til\chi_L,\gamma^\mu Y_\mu\til\psi_R\rangle \\
+ \frac12 \langle J_M\til\psi_L,\gamma^\mu Y_\mu\til\chi_R\rangle + \frac12 \langle J_M\til\psi_R,\gamma^\mu Y_\mu\til\chi_L\rangle .
\end{multline*}
In a similar manner as above, we obtain
\begin{align*}
\frac12 \langle J\til\xi,(\gamma^\mu\otimes B_\mu)\til\xi\rangle = -\langle J_M\til\chi,\gamma^\mu Y_\mu\til\psi\rangle ,
\end{align*}
where we have used that the form $\langle J_M\til\chi,\gamma^\mu Y_\mu\til\psi\rangle$ is anti-symmetric. The third term is given by
\begin{multline*}
\frac12 \langle J\til\xi,(\gamma_5\otimes D_F)\til\xi\rangle = \frac12 \langle J_M\til\chi_R,d\gamma_5\til\psi_R\rangle + \frac12 \langle J_M\til\chi_L,\bar d\gamma_5\til\psi_L\rangle \\
+ \frac12 \langle J_M\til\psi_L,\bar d\gamma_5\til\chi_L\rangle + \frac12 \langle J_M\til\psi_R,d\gamma_5\til\chi_R\rangle .
\end{multline*}
The bilinear form $\langle J_M\til\chi,\gamma_5\til\psi\rangle$ is again symmetric, but we now have the extra complication that two terms contain the parameter $d$, while the other two terms contain $\bar{d}$. Therefore we are left with two distinct terms:
\begin{equation*}
\frac12 \langle J\til\xi,(\gamma_5\otimes D_F)\til\xi\rangle = \langle J_M\til\chi_L,\bar d\til\psi_L\rangle - \langle J_M\til\chi_R,d\til\psi_R\rangle . \qedhere
\end{equation*}
\end{proof}

\begin{remark}
\label{remark:real_mass}
It is interesting to note that the fermions acquire mass terms without being coupled to a Higgs field. However, it seems we obtain a complex mass parameter $d$, where we would desire a real parameter $m$. By simply requiring that our result should be similar to \eqref{eq:coleman}, we will choose $d:=-im$, so that
\begin{align*}
\langle J_M\til\chi_L,\bar d\til\psi_L\rangle - \langle J_M\til\chi_R,d\til\psi_R\rangle = i \big\langle J_M\til\chi,m\til\psi\big\rangle .
\end{align*}
\end{remark}

The results obtained in this section can now be summarized into the following theorem.

\begin{thm}
\label{thm:ED}
The full Lagrangian of the almost-commutative manifold
\begin{align*}
M\times F\Sub{ED} = \left( C^\infty(M,\C^2), L^2(M,S)\otimes\C^4, \sD\otimes\1 + \gamma_5\otimes D_F, \gamma_5\otimes\gamma_F, J_M\otimes J_F \right) 
\end{align*}
as defined in this section, can be written as the sum of a purely gravitational Lagrangian,
\begin{align*}
\L\Sub{\text{grav}}(g_{\mu\nu}) = 4 \L_M(g_{\mu\nu}) + \L_H(g_{\mu\nu}) ,
\end{align*}
and a Lagrangian for electrodynamics,
\begin{align*}
\L\Sub{\text{ED}} = -i\Big( J_M\til\chi,(\gamma^\mu(\nabla^S_\mu - i Y_\mu)-m)\til\psi\Big) + \frac{f(0)}{6\pi^2} \F_{\mu\nu}\F^{\mu\nu} .
\end{align*}
\end{thm}
\begin{proof}
The spectral action $S_B$ and the fermionic action $S_F$ are given by \cref{prop:spec_act_ED,prop:fermion_act_ED}. This immediately yields $\L\Sub{\text{grav}}$. To obtain $\L\Sub{\text{ED}}$, we need to rewrite the fermionic action $S_F$ as the integral over a Lagrangian. The inner product $\langle\;,\;\rangle$ on the Hilbert space $L^2(S)$ is given by 
\begin{align*}
\langle\xi,\psi\rangle = \int_M (\xi,\psi) \sqrt{|g|} d^4x ,
\end{align*}
where the hermitian pairing $(\;,\;)$ is given by the pointwise inner product on the fibres. Choosing $d=-im$ as in \cref{remark:real_mass}, we can then rewrite the fermionic action into
\begin{align*}
S_F &= -\int_M i\Big( J_M\til\chi,\big(\gamma^\mu(\nabla^S_\mu - i Y_\mu)-m\big)\til\psi\Big) \sqrt{|g|} d^4x . \qedhere
\end{align*}
\end{proof}

\subsubsection{Fermionic degrees of freedom}

To conclude this section, let us make a final remark on the fermionic degrees of freedom in the Lagrangian derived above. For this purpose, we will first give a short introduction to Grassmann variables, and use this to find the relation between the Pfaffian and the determinant of an antisymmetric matrix. For more details we refer the reader to \cite[\S1.5]{Nakahara03}. Subsequently, we shall use the Grassmann integrals to briefly study the path integral of the fermionic action for electrodynamics. 

For a set of anticommuting Grassmann variables $\theta_i$, we have $\theta_i\theta_j = - \theta_j\theta_i$, and in particular, $\theta_i^2 = 0$. On these Grassmann variables $\theta_j$, we define an integral by
\begin{align*}
\int 1 d\theta_j &= 0 , & \int \theta_j d\theta_j &= 1 .
\end{align*}
If we have a Grassmann vector $\theta$ consisting of $N$ components, we define the integral over $D[\theta]$ as the integral over $d\theta_1\cdots d\theta_N$. Suppose we have two Grassmann vectors $\eta$ and $\theta$ of $N$ components. We then define the integration element as $D[\eta,\theta] = d\eta_1d\theta_1\cdots d\eta_Nd\theta_N$. 

Consider the Grassmann integral over a function of the form $e^{\theta^TA\eta}$ for Grassmann vectors $\theta$ and $\eta$ of $N$ components. The $N\times N$-matrix $A$ can be considered as a bilinear form on these Grassmann vectors. In the case where $\theta$ and $\eta$ are independent variables, we find 
\begin{align}
\label{eq:int_Grassmann_det}
\int e^{\theta^TA\eta} D[\eta,\theta] = \det A ,
\end{align}
where the determinant of $A$ is given by the formula
\begin{align*}
\det(A) = \frac1{N!} \sum_{\sigma,\tau\in\Pi_N} (-1)^{|\sigma|+|\tau|} A_{\sigma(1)\tau(1)}\cdots A_{\sigma(N)\tau(N)} ,
\end{align*}
where $\Pi_N$ denotes the set of all permutations of $1,2,\ldots,N$. Now let us assume that $A$ is an antisymmetric $N\times N$-matrix $A$ for $N=2l$. 
If we then take $\theta=\eta$, we find
\begin{align}
\label{eq:int_Grassmann_Pf}
\int e^{\frac12\eta^TA\eta} D[\eta] = \Pf(A) ,
\end{align}
where the \emph{Pfaffian} of $A$ is given by 
\begin{align*}
\Pf(A) = \frac{(-1)^l}{2^ll!} \sum_{\sigma\in\Pi_{2l}} (-1)^{|\sigma|} A_{\sigma(1)\sigma(2)} \cdots A_{\sigma(2l-1)\sigma(2l)} .
\end{align*}
Finally, using these Grassmann integrals, one can show that the determinant of a $2l\times2l$ skewsymmetric matrix $A$ is the square of the Pfaffian:
\begin{align*}
\det A = \Pf(A)^2 .
\end{align*}
So, by simply considering one instead of two independent Grassmann variables in the Grassmann integral of $e^{\theta^TA\eta}$, we are in effect taking the square root of a determinant. 

As mentioned after \cref{defn:act_funct}, the number of degrees of freedom of the fermion fields in the fermionic action is related to the restrictions that are incorporated into the definition of the fermionic action. These restrictions make sure that in this case we obtain two independent Dirac spinors in the fermionic action. 

In quantum field theory, one would consider the functional integral of $e^{S}$ over the fields. Let us now denote $\mA$ for the antisymmetric bilinear form on $\mH^+$ and $\mB$ for the bilinear form on $L^2(M,S)$, given by
\begin{align*}
\mA(\xi,\zeta) &:= \langle J\xi,D_A\zeta\rangle ,&&\text{ for } \xi,\zeta\in\mH^+ , \\
\mB(\chi,\psi) &:= -i\Big\langle J_M\chi,\big(\gamma^\mu(\nabla^S_\mu - i Y_\mu)-m\big)\psi\Big\rangle ,&&\text{ for } \chi,\psi\in L^2(M,S) .
\end{align*}
We have shown in \cref{prop:fermion_act_ED} that for $\xi = \chi_L\otimes e_L + \chi_R\otimes e_R + \psi_R\otimes\bar{e_L} + \psi_L\otimes\bar{e_R}$, where we can define two Dirac spinors by $\chi := \chi_L + \chi_R$ and $\psi := \psi_L + \psi_R$, we obtain 
\begin{align*}
\frac12\mA(\xi,\xi) = \mB(\chi,\psi) . 
\end{align*}
Using the Grassmann integrals that were calculated in \eqref{eq:int_Grassmann_det,eq:int_Grassmann_Pf}, we then obtain for the bilinear forms $\mA$ and $\mB$ the equality
\begin{align*}
\Pf(\mA) = \int e^{\frac12 \mA(\til\xi,\til\xi)} D[\til\xi] = \int e^{\mB(\til\chi,\til\psi)} D[\til\psi,\til\chi] = \det(\mB) .
\end{align*}

\newpage
\section{The Glashow-Weinberg-Salam Model} 
\label{chap:ex_GWS}

In the previous section we have described the theory of electrodynamics on an almost-commutative manifold. It has been shown in \cite{CCM07} (see also \cite{CM07}) that for a suitable choice of the finite space, the corresponding almost-commutative manifold gives rise to the full Standard Model (see \cref{chap:ex_SM}). The present section serves as an intermediate step between these two models. We will modify the finite space $F\Sub{ED}$ for electrodynamics such that it will incorporate the weak interactions. In other words, we will reproduce the Glashow-Weinberg-Salam Model, which describes the electroweak interactions for one generation of the leptonic sector of the Standard Model. An important feature of the Standard Model already occurs in this electroweak theory, namely the Higgs mechanism. The main purpose of this section is to show how this Higgs mechanism arises from an almost-commutative manifold, without worrying about the quark sector present in the Standard Model. 

Although it is perfectly possible to derive the fermionic action for this model, by exactly the same approach as for electrodynamics in \cref{chap:ex_ED}, we will refrain from doing so. The Higgs mechanism is given solely in the bosonic part of the Lagrangian, and for now we will therefore only focus on the spectral action. In \cref{chap:ex_SM} we will discuss the full Standard Model, and we shall derive the fermionic action there.

\subsection{The finite space}
\label{sec:GWS_finite_space}

We start by constructing a finite space $F\Sub{GWS}$, starting with the finite space $F\Sub{ED}$ for electrodynamics from the previous section. In the latter case, the finite Hilbert space was given by the basis $\left\{e_R, e_L, \bar{e_R}, \bar{e_L}\right\}$. Similarly, we will now also incorporate the left- and right-handed neutrinos and anti-neutrinos given by $\left\{\nu_R, \nu_L, \bar{\nu_R}, \bar{\nu_L}\right\}$. Together, these particles form the first generation of the leptons and anti-leptons in the Standard Model. We write $\mH_l = \C^4$ for the space of leptons, given by the basis $(\nu_R, e_R, \nu_L, e_L)$. The space of anti-leptons $\mH_{\bar l} = \C^4$ then has the basis $(\bar{\nu_R}, \bar{e_R}, \bar{\nu_L}, \bar{e_L})$. The total finite Hilbert space is given by 
$$
\mH_F = \mH_l\oplus\mH_{\bar l} .
$$ 
In the case of electrodynamics, the algebra was given by $\C\oplus\C$. We must expand this algebra such that it will describe the weak interactions as well. We do this by replacing the second copy of $\C$ by the quaternions $\qH$, so we shall take 
$$
\A_F = \C\oplus\qH .
$$
We can write $q\in\qH$ as $q=\alpha+\beta j$ for $\alpha,\beta\in\C$. We shall write $q_\lambda$ for the embedding of $\C$ in $\qH$. The quaternions $\qH$ can be embedded into $M_2(\C)$ as
\begin{align}
\label{eq:quaternion_M2}
q_\lambda &= \mattwo{\lambda}{0}{0}{\bar\lambda}, & q &= \mattwo{\alpha}{\beta}{-\bar\beta}{\bar\alpha} . 
\end{align}
Note that the embedding $\qH\subset M_2(\C)$ is real-linear, but not complex-linear, and consequently the algebra $\A_F$ should be considered as a real algebra. An element $a=(\lambda,q)\in\A_F$ acts on the space of leptons $\mH_l$ by multiplication with the matrix
\begin{align}
\label{eq:alg_rep_GWS}
a = (\lambda,q) \mapsto \mattwo{q_\lambda}{0}{0}{q} = \matfour{\lambda&0&0&0}{0&\bar\lambda&0&0}{0&0&\alpha&\beta}{0&0&-\bar\beta&\bar\alpha} .
\end{align}
For the action of $a$ on an antilepton $\bar l\in\mH_{\bar l}$ we set $a\bar l = \lambda\bar l$. 

The $\Z_2$-grading $\gamma_F$ and the real structure $J_F$ are chosen in the same way as for electrodynamics, such that we will again obtain a finite space of KO-dimension $6$. The antilinear conjugation operator $J_F$ interchanges particles with their anti-particles, so $J_F l = \bar l$ and $J_F \bar l = l$. The $\Z_2$-grading $\gamma_F$ is chosen such that left-handed particles have positive eigenvalue and right-handed particles have negative eigenvalue. As before, $C$ stands for complex conjugation, so on the decomposition $\mH = \mH_{l_R}\oplus\mH_{l_L}\oplus\mH_{\bar {l_R}}\oplus\mH_{\bar {l_L}}$ 
we can write
\begin{align*}
\gamma_F &= \matfour{-1&0&0&0}{0&1&0&0}{0&0&1&0}{0&0&0&-1} , & J_F &= \matfour{0&0&C&0}{0&0&0&C}{C&0&0&0}{0&C&0&0} .
\end{align*}

\subsubsection{The finite Dirac operator}
\label{sec:D_F_GWS}

We are left only with deriving the most general form of the Dirac operator $D_F$ that is consistent with the above definitions. First, $D_F$ must be hermitian, which implies that we can write 
$$
D_F = \mattwo{S}{T^*}{T}{S'}
$$
on the decomposition $\mH = \mH_l\oplus\mH_{\bar l}$, for hermitian $S,S'$. Since the finite space is even, we need that $D_F$ commutes with $J_F$:
\begin{align*}
0 = [D_F,J_F] = \mattwo{C(T^T-T)}{C(\bar S-S')}{C(\bar{S'}-S)}{C(\bar T-T^*)} .
\end{align*}
This imposes the relations $S'=\bar S$ and $T=T^T$. We also require that $D_F$ anticommutes with $\gamma_F$, which yields
\begin{align*}
0 &= D_F\gamma_F + \gamma_FD_F \\
&= \mattwo{S\mattwo{-1}{0}{0}{1}}{T^*\mattwo{1}{0}{0}{-1}}{T\mattwo{-1}{0}{0}{1}}{\bar S\mattwo{1}{0}{0}{-1}} + \mattwo{\mattwo{-1}{0}{0}{1}S}{\mattwo{-1}{0}{0}{1}T^*}{\mattwo{1}{0}{0}{-1}T}{\mattwo{1}{0}{0}{-1}\bar S} .
\end{align*}
This means that we can write
\begin{align*}
S &= \mattwo{0}{Y_0^*}{Y_0}{0} , & T &= \mattwo{T_R}{0}{0}{T_L} ,
\end{align*}
where $T_R$ and $T_L$ are required to be symmetric. We will consider the restriction that $T\nu_R = Y_R\bar{\nu_R}$ for some complex parameter $Y_R$, and $Tl=0$ for all other leptons $l\neq\nu_R$. As will be shown below, this restriction makes sure that the order one condition \labelcref{eq:order1} is satisfied. The mass matrix $Y_0$ can be written as a diagonal matrix by simply requiring that the basis elements of $\mH_F$ are mass eigenstates. Hence we shall take 
\begin{align*}
Y_0 = \mattwo{Y_\nu}{0}{0}{Y_e} ,
\end{align*}
for two complex parameters $Y_\nu$ and $Y_e$. We now arrive at the following result. 

\begin{prop}
\label{prop:spec_trip_GWS}
The data 
\begin{align*}
F\Sub{GWS} := \left( \A_F, \mH_F, D_F, \gamma_F, J_F \right)
\end{align*}
as given above define a real even finite space of KO-dimension $6$.
\end{prop}
\begin{proof}
One immediately sees that $\gamma_F$ commutes with the algebra $\A_F$. We have already shown that $D_FJ_F=J_FD_F$ and $[D_F,\gamma_F]=0$. We also have $J_F^2=1$ and $J_F\gamma_F = -\gamma_FJ_F$. From the table in \cref{defn:real_structure} we then see that we have KO-dimension $6$. It remains to check the order one condition $\big[[D_F,a],b\big]=0$. The action of the algebra on $\mH_{\bar l}$ is by scalar multiplication, so we find that $[\bar S, a] = 0$ on $\mH_{\bar l}$. On $\mH_l$, the right action $a^0 = Ja^*J^* = \lambda$ is also just scalar multiplication, so we obtain that 
$$
\left[\left[\mattwo{S}{0}{0}{\bar S},a\right],b^0\right] = 0 .
$$
Since $a\nu_R = \lambda\nu_R$ and also $a \bar{\nu_R} = \lambda\bar{\nu_R}$, the action of $a$ commutes with $T$, and the order one condition is indeed satisfied.
\end{proof}

\subsection{The gauge theory} 

Let us describe the gauge theory corresponding to the almost-commutative manifold $M\times F\Sub{GWS}$. We will frequently make use of the Pauli matrices $\sigma^a$, which are given by
\begin{align}
\label{eq:Pauli}
\sigma^1 &= \mattwo{0}{1}{1}{0} , & \sigma^2 &= \mattwo{0}{-i}{i}{0} , & \sigma^3 &= \mattwo{1}{0}{0}{-1} .
\end{align}

\subsubsection{The gauge group}

First, we need to derive the local gauge group from the finite space $F\Sub{GWS}$. Let us start by examining the subalgebra $(\til\A_F)_{J_F}$ of the algebra $\A_F = \C\oplus\qH$, as defined in \cref{sec:subalgs}. This subalgebra is determined by the relation $aJ_F=J_Fa^*$. An element $a=(\lambda,q)\in\C\oplus\qH$ satisfies this relation if $\lambda=\bar\lambda=\alpha=\bar\alpha$ and $\beta=0$, so if $a=(x,x)$ for $x\in\R$. Hence we find that 
\begin{align}
\label{eq:subalg_GWS}
(\til\A_F)_{J_F} \simeq \R .
\end{align}

Next, let us consider the Lie algebra $\h_F = \lu\big((\til\A_F)_{J_F}\big)$ of \eqref{eq:unitary_subLiealg_F}. The anti-hermitian elements $u\in\lu(\A_F)$ are given by $u=(\lambda,q)$ for $\lambda\in i\R$ and for $iq$ a linear combination of the Pauli matrices of \eqref{eq:Pauli}. In particular this means that $\bar\lambda = -\lambda$. Hence in the cross-section $\h_F = \lu\big((\til\A_F)_{J_F}\big)$ we find $\alpha=\bar\alpha=\lambda=\bar\lambda=-\lambda=0$. Hence $\h_F$ is given by the trivial subset 
\begin{align}
\label{eq:h_F_triv}
\h_F = \{0\} .
\end{align}

\begin{prop}
\label{prop:gauge_group_GWS}
The local gauge group of the finite space $F\Sub{GWS}$ is given by
\begin{align*}
\G(F\Sub{GWS}) \simeq \big(U(1)\times SU(2)\big) / \{1,-1\} .
\end{align*}
\end{prop}
\begin{proof}
The unitary elements of the algebra form the group $U(\A_F) \simeq U(1)\times U(\qH)$. The quaternions are spanned by the identity matrix $\1$ and the anti-hermitian matrices $i\sigma_j$, where $\sigma_j$ ($j=1,2,3$) are the Pauli matrices. A quaternion $q = q_0\1 + iq_1\sigma_1 + iq_2\sigma_2 + iq_3\sigma_3$ is unitary if and only if $|q|^2 = {q_0}^2 + {q_1}^2 + {q_2}^2 + {q_3}^2 = 1$. By using the embedding of $\qH$ in $M_2(\C)$, we find $|q|^2 = \det(q) = 1$, and this yields the isomorphism $U(\qH)\simeq SU(2)$ (for more details, see, for instance, \cite[\S1.2.B]{DK00}). Note that if $q$ is unitary, then so is $-q$. 

From \eqref{eq:subalg_GWS} we know that $(\til\A_F)_{J_F} = \R$. The group $H_F = U\big((\til\A_F)_{J_F}\big)$ is then given by $H_F = \{1,-1\}$. From \eqref{eq:gauge_group}, we find the gauge group $\G(F\Sub{GWS})$ to be the quotient $U(\A_F) / H_F$. 
\end{proof}

Note that, although we obtain the gauge group $\big(U(1)\times SU(2)\big) / \{1,-1\}$, this is very similar to $U(1)\times SU(2)$, since both groups have the same Lie algebra $\lu(1)\oplus\su(2)$. 

As shown in \cref{prop:unimod}, the unimodularity condition is satisfied naturally only for complex algebras. In this case however we only have a real-linear representation of the algebra, so the unimodularity condition is not satisfied. Indeed, in \eqref{eq:h_F_triv} we found that the Lie subalgebra $\h_F$ is trivial, and hence the gauge field $A_\mu$ takes values in the Lie algebra $\g_F = \lu(\A_F) / \h_F = \lu(\A_F) = \lu(1)\oplus\su(2)$, which is obviously not unimodular, because of the presence of the $\lu(1)$ part. Note that in this particular case we also would not want the unimodularity condition to be satisfied, because that would mean that our electromagnetic $U(1)$ gauge field would vanish.

\subsubsection{The gauge fields and the Higgs field}
\label{sec:fields_GWS}

Let us now derive the precise form of the gauge field $A_\mu$ of \eqref{eq:fluc_gauge} and the Higgs field $\phi$ of \eqref{eq:fluc_higgs}. We take two elements $a=(\lambda,q)$ and $b=(\lambda',q')$ of the algebra $\A = C^\infty(M,\C\oplus\qH)$. The inner fluctuations $A_\mu = -ia\partial_\mu b$ are obtained from \eqref{eq:alg_rep_GWS} to be $\Lambda_\mu := -i\lambda\partial_\mu\lambda'$ on $\nu_R$, $\Lambda_\mu' := -i\bar\lambda\partial_\mu\bar\lambda'$ on $e_R$, and $Q_\mu := -iq\partial_\mu q'$ on $(\nu_l,e_L)$. Demanding the hermiticity of $\Lambda_\mu=\Lambda_\mu^*$ implies $\Lambda_\mu\in\R$, and also automatically yields $\Lambda_\mu' = -\Lambda_\mu$. Furthermore, $Q_\mu = Q_\mu^*$ implies that $Q_\mu$ is a real-linear combination of the Pauli matrices, which span $i\,\su(2)$. 

Next, we calculate the field $\phi = a[D_F,b]$. In the proof of \cref{prop:spec_trip_GWS} we have already noted that the only part of $D_F$ that does not commute with the algebra is given by $S$. Therefore, we start by calculating the commutator on $\mH_l$ given by
\begin{align*}
[S,b] &= 
\matfour{0&0&\bar Y_\nu(\alpha'-\lambda')&\bar Y_\nu\beta'}{0&0&-\bar Y_e\bar\beta'&\bar Y_e(\bar\alpha'-\bar\lambda')}{Y_\nu(\lambda'-\alpha')&-Y_e\beta'&0&0}{Y_\nu\bar\beta'&Y_e(\bar\lambda'-\bar\alpha')&0&0} .
\end{align*}
By multiplying this with the element $a$, we obtain
\begin{align}
\label{eq:higgs_field_GWS}
\phi = \matfour{0&0&\bar Y_\nu\phi_1'&\bar Y_\nu\phi_2'}{0&0&-\bar Y_e\bar\phi_2'&\bar Y_e\bar\phi_1'}{Y_\nu\phi_1&-Y_e\bar\phi_2&0&0}{Y_\nu\phi_2&Y_e\bar\phi_1&0&0} ,
\end{align}
where we define\footnote{This notation looks very similar to the notation of $\varphi_1,\varphi_2,\varphi_1',\varphi_2'$ that is used in \cite{CCM07} (see also \cite[Ch.~1, \S15.2]{CM07}), but we have taken $\phi_1 = \varphi_1'$ and $\phi_2 = -\bar\varphi_2'$. The reason for this change in notation is that we obtain a prettier formula for the gauge transformation in \cref{prop:gauge_transf_GWS}.}
\begin{align*}
\phi_1 &= \alpha(\lambda'-\alpha')+\beta\bar\beta' , \qquad & \phi_2 &= \bar\alpha\bar\beta'-\bar\beta(\lambda'-\alpha') ,\\
\phi_1' &= \lambda(\alpha'-\lambda') , \qquad & \phi_2' &= \lambda\beta' .
\end{align*}
By demanding $\phi=\phi^*$, we obtain $\phi_1' = \bar\phi_1$ and $\phi_2' = \bar\phi_2$. Hence we find that the field $\phi$ is completely determined by the complex doublet $(\phi_1,\phi_2)$. 

In general, an inner fluctuation is given by a sum of terms, of the form $A = \sum_j a_j [D,b_j]$. For such a general inner fluctuation, we simply need to redefine $\Lambda_\mu := -\sum_j i \lambda_j\partial_\mu\lambda_j'$ and $Q_\mu := -\sum_j i q_j\partial_\mu q_j'$, as well as 
\begin{align*}
\phi_1 &= \sum_j \alpha_j(\lambda_j'-\alpha_j')+\beta_j\bar\beta_j' , \qquad & \phi_2 &= \sum_j \bar\alpha_j\bar\beta_j'-\bar\beta_j(\lambda_j'-\alpha_j') ,\\
\phi_1' &= \sum_j \lambda_j(\alpha_j'-\lambda_j') , \qquad & \phi_2' &= \sum_j \lambda_j\beta_j' .
\end{align*}
To summarize, the fields $A_\mu$ and $\phi$ are given on $\mH_l$ by 
\begin{align}
\label{eq:fields_GWS}
A_\mu &= \matthree{\Lambda_\mu&0&}{0&-\Lambda_\mu&}{&&Q_\mu} , \qquad \text{for } \Lambda_\mu\in\R, \quad Q_\mu\in i\su(2) ;\notag\\
\phi &= \mattwo{0}{Y^*}{Y}{0} , \qquad \text{for } Y=\mattwo{Y_\nu\phi_1}{-Y_e\bar\phi_2}{Y_\nu\phi_2}{Y_e\bar\phi_1} , \quad \phi_1,\phi_2\in\C .
\end{align}
On $\mH_{\bar l}$ we have $A_\mu=\Lambda_\mu$ and $\phi=0$. The gauge field $B_\mu = A_\mu - J_FA_\mu J_F^*$ is then given by
\begin{align}
\label{eq:Gauge_field_GWS}
B_\mu \Big|_{\mH_l} &= \matthree{0&0&}{0&-2\Lambda_\mu&}{&&Q_\mu-\Lambda_\mu \1_2} , &
B_\mu \Big|_{\mH_{\bar l}} &= \matthree{0&0&}{0&2\Lambda_\mu&}{&&\Lambda_\mu \1_2-\bar Q_\mu} .
\end{align}
Note that the coefficients in front of $\Lambda_\mu$ in the above formulas, are precisely the well-known hypercharges of the corresponding particles, as given by the following table: 
\begin{align*}
\begin{array}{l|cccc}
\text{Particle} & \nu_R & e_R & \nu_L & e_L \\
\hline
\text{Hypercharge} & 0 & -2 & -1 & -1 \\
\end{array}
\end{align*}
The Higgs field $\Phi$ is given in matrix-form as
\begin{align}
\label{eq:Higgs_field_GWS}
\Phi = D_F + \mattwo{\phi}{0}{0}{0} + J_F\mattwo{\phi}{0}{0}{0}J_F^* = \mattwo{S+\phi}{T^*}{T}{\bar{(S+\phi)}} ,
\end{align}
where $\phi$ is the matrix given by \eqref{eq:higgs_field_GWS}. 

\begin{prop}
\label{prop:gauge_transf_GWS}
The action of the gauge group $\G(M\times F\Sub{GWS})$ on the fluctuated Dirac operator
\begin{align*}
D_A = \sD\otimes\1 + \gamma^\mu\otimes B_\mu + \gamma_5\otimes\Phi
\end{align*}
is implemented by
\begin{align*}
\Lambda_\mu &\rightarrow \Lambda_\mu - i \lambda\partial_\mu\bar\lambda , \notag\\
Q_\mu &\rightarrow qQ_\mu q^* - iq\partial_\mu q^* , \\
\vectwo{\phi_1}{\phi_2} &\rightarrow 
\bar\lambda\,q \vectwo{\phi_1}{\phi_2} + (\bar\lambda\,q-1) \vectwo{1}{0} , \notag
\end{align*}
for $\lambda\in C^\infty\big(M,U(1)\big)$ and $q\in C^\infty\big(M,SU(2)\big)$. 
\end{prop}
\begin{proof}
We simply insert the formulas for the fields obtained in \eqref{eq:fields_GWS} into the transformations given by \eqref{eq:acm_transf_inner_fluc}. We shall write $u=(\lambda,q)\in C^\infty\big(M,U(1)\times SU(2)\big)$. On $\mH_{\bar l}$ we see that $A_\mu$ commutes with $\lambda$. On $(\nu_R,e_R)$ we see that $A_\mu$ also commutes with $q_\lambda$. Hence the only effect of the term $uA_\mu u^*$ is to replace $Q_\mu$ by $qQ_\mu q^*$. Secondly, we see that the term $-iu\partial_\mu u^*$ is given by $- i \lambda\partial_\mu\bar\lambda$ on $\nu_R$ and $\mH_{\bar l}$, by $- i \bar\lambda\partial_\mu\lambda = i \lambda\partial_\mu\bar\lambda$ on $e_R$, and by $- iq\partial_\mu q^*$ on $(\nu_L,e_L)$. We thus obtain the desired transformation for $\Lambda_\mu$ and $Q_\mu$. 

For the transformation of $\phi$, we separately calculate $u\phi u^*$ and $u[D_F,u^*]$. Since $\phi=0$ on $\mH_{\bar l}$, we can restrict our calculation of $u\phi u^*$ to $\mH_l$ and find 
\begin{align*}
u\phi u^* = \mattwo{q_\lambda}{0}{0}{q} \mattwo{0}{Y^*}{Y}{0} \mattwo{q_\lambda^*}{0}{0}{q^*} = \mattwo{0}{q_\lambda Y^*q^*}{qYq_\lambda^*}{0} ,
\end{align*}
which is still hermitian. We then calculate that
\begin{align*}
qYq_\lambda^* &= \mattwo{\alpha}{\beta}{-\bar\beta}{\bar\alpha} \mattwo{Y_\nu\phi_1}{-Y_e\bar\phi_2}{Y_\nu\phi_2}{Y_e\bar\phi_1} \mattwo{\bar\lambda}{0}{0}{\lambda} \\
&= \mattwo{\bar\lambda Y_\nu (\alpha\phi_1 + \beta\phi_2)}{\lambda Y_e (\beta\bar\phi_1-\alpha\bar\phi_2)}{\bar\lambda Y_\nu (-\bar\beta\phi_1+\bar\alpha\phi_2)}{\lambda Y_e (\bar\alpha\bar\phi_1+\bar\beta\bar\phi_2)} .
\end{align*}
Now let us calculate the second term $u[D_F,u^*]$, where $D_F$ is given in \cref{sec:D_F_GWS}. The operator $T$ only acts on $\nu_R$ and therefore commutes with the algebra. On the restriction to $\mH_{\bar l}$, the operator $\bar S$ commutes with the algebra. Hence again we can restrict our calculation to $\mH_l$. The term $u[S,u^*]$ splits into $uSu^* - S$, and (similarly to $u\phi u^*$) we find
\begin{align*}
uSu^* = \mattwo{0}{q_\lambda Y_0^*q^*}{qY_0q_\lambda^*}{0}
\end{align*}
and
\begin{align*}
qY_0q_\lambda^* &= \mattwo{\alpha}{\beta}{-\bar\beta}{\bar\alpha} \mattwo{Y_\nu}{0}{0}{Y_e} \mattwo{\bar\lambda}{0}{0}{\lambda} = \mattwo{\bar\lambda Y_\nu \alpha}{\lambda Y_e \beta}{-\bar\lambda Y_\nu \bar\beta}{\lambda Y_e \bar\alpha} .
\end{align*}
Combining the two contributions to the transformation, we find that the transformation $u\phi u^* + u[S,u^*]$ yields
\begin{multline*}
Y = \mattwo{Y_\nu\phi_1}{-Y_e\bar\phi_2}{Y_\nu\phi_2}{Y_e\bar\phi_1} \rightarrow Y' = \mattwo{Y_\nu\phi_1'}{-Y_e\bar\phi_2'}{Y_\nu\phi_2'}{Y_e\bar\phi_1'} \\ 
= \mattwo{\bar\lambda Y_\nu (\alpha\phi_1 + \beta\phi_2)}{\lambda Y_e (\beta\bar\phi_1-\alpha\bar\phi_2)}{\bar\lambda Y_\nu (-\bar\beta\phi_1+\bar\alpha\phi_2)}{\lambda Y_e (\bar\alpha\bar\phi_1+\bar\beta\bar\phi_2)} + \mattwo{\bar\lambda Y_\nu (\alpha-1)}{\lambda Y_e \beta}{-\bar\lambda Y_\nu \bar\beta}{\lambda Y_e (\bar\alpha-1)} .
\end{multline*}
where we have defined $\phi_1' := \bar\lambda(\alpha\phi_1+\beta\phi_2 + \alpha)-1$ and $\phi_2' := \bar\lambda(-\bar\beta\phi_1 + \bar\alpha\phi_2 - \bar\beta)$. Rewriting this in terms of $q$ then proves the proposition.
\end{proof}

In \eqref{eq:acm_transf_inner_fluc}, we have seen that in general the transformation of the field $\phi$ is not a linear transformation. In the present model, \cref{prop:gauge_transf_GWS} shows that it can be reduced to an affine transformation of the doublet $\phi_1,\phi_2$. This can be rewritten in the linear form
\begin{align*}
\vectwo{\phi_1+1}{\phi_2} \rightarrow \bar\lambda\,q \vectwo{\phi_1+1}{\phi_2} .
\end{align*}
One should note here that, whereas the complex doublet $(\phi_1,\phi_2)$ corresponds to the field $\phi$, the doublet $(1,0)$ corresponds to the operator $S$, which is a part of $D_F$. We thus see that the combination $S+\phi$ has a linear transformation under the gauge group.

\subsection{The spectral action}
\label{sec:spec_act_GWS}

In this section we will calculate the bosonic part of the Lagrangian of the Glashow-Weinberg-Salam Model. The general form of this Lagrangian has already been calculated in \cref{prop:acm_spec_act} so we only need to insert the expressions \eqref{eq:Gauge_field_GWS,eq:Higgs_field_GWS} for the fields $B_\mu$ and $\Phi$. We first start with a few lemmas, in which we capture the rather tedious calculations that are needed to obtain the traces of $F_{\mu\nu}F^{\mu\nu}$, $\Phi^2$, $\Phi^4$ and $(D_\mu\Phi)(D^\mu\Phi)$. 

\begin{lem}
\label{lem:curv_GWS}
The trace of the square of the curvature of $B_\mu$ is given by
\begin{align*}
\Tr(F_{\mu\nu}F^{\mu\nu}) = 12\Lambda_{\mu\nu}\Lambda^{\mu\nu} + 2\Tr(Q_{\mu\nu}Q^{\mu\nu}) . 
\end{align*}
\end{lem}
\begin{proof}
Let us define the curvatures of the $U(1)$ and $SU(2)$ gauge fields by
\begin{align}
\label{eq:curv_GWS}
\Lambda_{\mu\nu} &:= \partial_\mu\Lambda_\nu - \partial_\nu\Lambda_\mu , \notag\\
Q_{\mu\nu} &:= \partial_\mu Q_\nu - \partial_\nu Q_\mu + i[Q_\mu,Q_\nu] .
\end{align}
Using \eqref{eq:Gauge_field_GWS}, we then find that the curvature $F_{\mu\nu}$ of $B_\mu$ can be written as
\begin{align*}
F_{\mu\nu} \Big|_{\mH_l} &= \matthree{0&0&}{0&-2\Lambda_{\mu\nu}&}{&&Q_{\mu\nu}-\Lambda_{\mu\nu}\1_2} , \\
F_{\mu\nu} \Big|_{\mH_{\bar l}} &= \matthree{0&0&}{0&2\Lambda_{\mu\nu}&}{&&\Lambda_{\mu\nu}\1_2-(\bar Q)_{\mu\nu}} .
\end{align*}
The curvature squared thus becomes
\begin{align*}
F_{\mu\nu}F^{\mu\nu} \Big|_{\mH_l} &= \matthree{0&0&}{0&4\Lambda_{\mu\nu}\Lambda^{\mu\nu}&}{&&Q_{\mu\nu}Q^{\mu\nu}+\Lambda_{\mu\nu}\Lambda^{\mu\nu}\1_2-2\Lambda_{\mu\nu}Q^{\mu\nu}} , \\
F_{\mu\nu}F^{\mu\nu} \Big|_{\mH_{\bar l}} &= \matthree{0&0&}{0&4\Lambda_{\mu\nu}\Lambda^{\mu\nu}&}{&&(\bar Q)_{\mu\nu}(\bar Q)^{\mu\nu}+\Lambda_{\mu\nu}\Lambda^{\mu\nu}\1_2-2\Lambda_{\mu\nu}(\bar Q)^{\mu\nu}} .
\end{align*}
Since $Q_{\mu\nu}$ is traceless, the cross-term $-2\Lambda_{\mu\nu}Q^{\mu\nu}$ will drop out after taking the trace. Note that since $Q_\mu$ is hermitian we have $\bar Q_\mu = Q_\mu^T$. We then see that $Q_{\mu\nu}$ is also hermitian, since
\begin{align*}
\bar{(Q_{\mu\nu})} &= \partial_\mu\bar Q_\nu - \partial_\nu\bar Q_\mu + i [\bar Q_\mu,\bar Q_\nu] = \partial_\mu Q_\nu^T - \partial_\nu Q_\mu^T + i [Q_\mu^T,Q_\nu^T] \\ 
&= \big( \partial_\mu Q_\nu - \partial_\nu Q_\mu - i [Q_\mu,Q_\nu] \big)^T = (Q_{\mu\nu})^T .
\end{align*}
This implies that
\begin{align*}
\Tr\big(\bar{(Q_{\mu\nu})}\bar{(Q^{\mu\nu})}\big) = \Tr\big((Q_{\mu\nu})^T(Q^{\mu\nu})^T\big) = \Tr\Big(\big(Q^{\mu\nu}Q_{\mu\nu}\big)^T\Big) = \Tr\big(Q_{\mu\nu}Q^{\mu\nu}\big) .
\end{align*}
We thus obtain that 
\begin{align*}
\Tr(F_{\mu\nu}F^{\mu\nu}) &= 12\Lambda_{\mu\nu}\Lambda^{\mu\nu} + 2\Tr(Q_{\mu\nu}Q^{\mu\nu}) . \qedhere 
\end{align*}
\end{proof}

\begin{lem}
\label{lem:Higgs_terms_GWS}
The traces of $\Phi^2$ and $\Phi^4$ are given by
\begin{align*}
\Tr\big(\Phi^2\big) &= 4a |H'|^2 +2c , \\
\Tr\big(\Phi^4\big) &= 4b|H'|^4 + 8e|H'|^2 + 2 d ,
\end{align*}
where $H'$ denotes the complex doublet $(\phi_1+1,\phi_2)$ and, following \cite{CCM07} (see also \cite[Ch.~1, \S15.2]{CM07}), 
\begin{align}
\label{eq:abcde_GWS}
a &= |Y_\nu|^2 + |Y_e|^2 , \notag\\
b &= |Y_\nu|^4 + |Y_e|^4 , \notag\\
c &= |Y_R|^2 , \\
d &= |Y_R|^4 , \notag\\
e &= |Y_R|^2 |Y_\nu|^2 . \notag
\end{align}
\end{lem}
\begin{proof}
The field $\Phi$ is given by \eqref{eq:Higgs_field_GWS}, and its square equals
\begin{align*}
\Phi^2 = \mattwo{(S+\phi)^2+T^*T}{(S+\phi)T^*+T^*\bar{(S+\phi)}}{T(S+\phi)+\bar{(S+\phi)}T}{\bar{(S+\phi)}^2+TT^*} . 
\end{align*}
The square of the off-diagonal part yields $T^*T = TT^* = |Y_R|^2$ on $\nu_R$ and $\bar{\nu_R}$, and zero on $l\neq\nu_R,\bar{\nu_R}$. The component $S+\phi$ is given by $$S+\phi = \mattwo{0}{Y^*+Y_0^*}{Y+Y_0}{0}.$$ We then calculate
\begin{align*}
X := (Y+Y_0)^*(Y+Y_0) 
&= |H'|^2 \mattwo{|Y_\nu|^2}{0}{0}{|Y_e|^2} .
\end{align*}
where we have defined the complex doublet $H' := (\phi_1+1,\phi_2)$.
Similarly, we define $X' := (Y+Y_0)(Y+Y_0)^*$ and note that $\Tr(X) = \Tr(X')$ by the cyclic property of the trace. Since $X=X^*$ and $\Tr(X)=\Tr(X^T)$, we also have $\Tr(\bar X) = \Tr(X)$. We thus obtain that 
\begin{align*}
\Tr\big(\Phi^2\big) &= \Tr(X+X'+\bar X+\bar X') +2|Y_R|^2 \\
&= 4\Tr(X) +2|Y_R|^2 = 4(|Y_\nu|^2 + |Y_e|^2) |H'|^2 +2|Y_R|^2 .
\end{align*}
In order to find the trace of $\Phi^4$, we calculate
\begin{align*}
(X+T^*T)^2 = |H'|^4 \mattwo{|Y_\nu|^4}{0}{0}{|Y_e|^4} + 2|H'|^2 \mattwo{|Y_R|^2 |Y_\nu|^2}{0}{0}{0} + \mattwo{|Y_R|^4}{0}{0}{0} .
\end{align*}
We now also obtain a contribution from the off-diagonal part of $\Phi^2$. Any term of the form $(S+\phi)T^*\bar{(S+\phi)}T$ (or a cyclic permutation thereof) vanishes. We do obtain contributions from $\Tr\big((S+\phi)T^*T(S+\phi)\big)$ and three other similar terms, which each yield the contribution $|H'|^2 |Y_R|^2 |Y_\nu|^2$. 
We thus obtain
\begin{align*}
\Tr\big(\Phi^4\big) &= \Tr\big( (X+T^*T)^2 + (X')^2 + (\bar X+TT^*)^2 + (\bar X')^2 \big) + 4|H'|^2 |Y_R|^2 |Y_\nu|^2\\
&= \Tr\big( 4X^2 + 4XT^*T + 2(T^*T)^2 \big) +4|H'|^2 |Y_R|^2 |Y_\nu|^2\\
&= 4|H'|^4 \big(|Y_\nu|^4 + |Y_e|^4\big) + 8|H'|^2 |Y_R|^2 |Y_\nu|^2 + 2 |Y_R|^4 . \qedhere
\end{align*}
\end{proof}

\begin{lem}
\label{lem:Higgs_kin_min_coupl_GWS}
The trace of $(D_\mu\Phi)(D^\mu\Phi)$ is given by
\begin{align*}
\Tr\big((D_\mu\Phi)(D^\mu\Phi)\big) = 4 a |\til D_\mu H'|^2 ,
\end{align*}
where $H'$ denotes the complex doublet $(\phi_1+1,\phi_2)$, and the covariant derivative $\til D_\mu$ on $H'$ is defined as
\begin{align*}
\til D_\mu H' = \partial_\mu H' + i Q_\mu^a \sigma^a H' - i \Lambda_\mu H' .
\end{align*}
\end{lem}
\begin{proof}
We need to calculate the commutator $[B_\mu,\Phi]$. We note that $B_\mu$ commutes with the off-diagonal part of $D_F$. It is thus sufficient to calculate the commutator $[B_\mu,S+\phi]$ on $\mH_l$. We shall write $Q_\mu = Q_\mu^1\sigma^1+Q_\mu^2\sigma^2+Q_\mu^3\sigma^3$ as a superposition of the Pauli matrices of \eqref{eq:Pauli} for real coefficients $Q_\mu^a$. We then obtain by direct calculation
\begin{align*}
[B_\mu,S+\phi] 
&= \matfour{0&0&-\bar Y_\nu\bar\chi_1&-\bar Y_\nu\bar\chi_2}{0&0&-\bar Y_e\chi_2&\bar Y_e\chi_1}{Y_\nu\chi_1&Y_e\bar\chi_2&0&0}{Y_\nu\chi_2&-Y_e\bar\chi_1&0&0} ,
\end{align*}
where we have defined the new doublet $\chi = (\chi_1,\chi_2)$ by
\begin{align*}
\chi_1 &:= (\phi_1+1)(Q_\mu^3-\Lambda_\mu)+\phi_2(Q_\mu^1-iQ_\mu^2) , \\
\chi_2 &:= (\phi_1+1)(Q_\mu^1+iQ_\mu^2)+\phi_2(-Q_\mu^3-\Lambda_\mu) .
\end{align*}
We then obtain that
\begin{multline*}
D_\mu(S+\phi) = \partial_\mu\phi + i [B_\mu,S+\phi] \\
= \matfour{0&0&\bar Y_\nu(\partial_\mu\bar\phi_1-i\bar\chi_1)&\bar Y_\nu(\partial_\mu\bar\phi_2-i\bar\chi_2)}{0&0&-\bar Y_e(\partial_\mu\phi_2+i\chi_2)&\bar Y_e(\partial_\mu\phi_1+i\chi_1)}{Y_\nu(\partial_\mu\phi_1+i\chi_1)&-Y_e(\partial_\mu\bar\phi_2-i\bar\chi_2)&0&0}{Y_\nu(\partial_\mu\phi_2+i\chi_2)&Y_e(\partial_\mu\bar\phi_1-i\bar\chi_1)&0&0} .
\end{multline*}
We want to calculate the trace of the square of $D_\mu\Phi$, for this reason we only need to calculate the terms on the diagonal of $(D_\mu\Phi)(D^\mu\Phi)$. We thus find
\begin{align*}
\Tr_{\mH_l}\Big((D_\mu(S+\phi))(D^\mu(S+\phi)) \Big) = 2 a \Big( |\partial_\mu\phi_1+i\chi_1|^2 + |\partial_\mu\phi_2+i\chi_2|^2 \Big) ,
\end{align*}
where we have used $a = |Y_\nu|^2 + |Y_e|^2$ as in \eqref{eq:abcde_GWS}. The column vector $H'$ is given by the complex doublet $(\phi_1+1,\phi_2)$. We then note that $\partial_\mu \phi + i \chi$ is equal to the covariant derivative $\til D_\mu H'$, so that we obtain
\begin{align*}
\Tr_{\mH_l}\Big((D_\mu(S+\phi))(D^\mu(S+\phi)) \Big) = 2 a |\til D_\mu H'|^2 .
\end{align*}
The trace over $\mH_{\bar l}$ yields exactly the same contribution, so we need to multiply this by $2$ and thus obtain the desired result.
\end{proof}

\begin{prop}
\label{prop:spec_act_GWS}
The spectral action of the AC-manifold
\begin{multline*}
M\times F\Sub{GWS} = \\
\left( C^\infty(M,\C\oplus\qH), L^2(M,S)\otimes(\C^4\oplus\C^4), \sD\otimes\1 + \gamma_5\otimes D_F, \gamma_5\otimes\gamma_F, J_M\otimes J_F \right) 
\end{multline*}
defined in this section is given by
\begin{align*}
\Tr \left(f\Big(\frac {D_A}\Lambda\Big)\right) &\sim \int_M \L(g_{\mu\nu}, \Lambda_\mu, Q_\mu, H') \sqrt{|g|} d^4x + O(\Lambda^{-1}) ,
\end{align*}
for the Lagrangian
\begin{align*}
\L(g_{\mu\nu}, \Lambda_\mu, Q_\mu, H') := 8\L_M(g_{\mu\nu}) + \L_A(\Lambda_\mu,Q_\mu) + \L_H(g_{\mu\nu}, \Lambda_\mu, Q_\mu, H') .
\end{align*}
Here $\L_M(g_{\mu\nu})$ is defined in \cref{prop:canon_spec_act}. The term $\L_A$ gives the kinetic terms of the gauge fields and equals 
\begin{align*}
\L_A(\Lambda_\mu,Q_\mu) := \frac{f(0)}{12\pi^2} \Big( 6\Lambda_{\mu\nu}\Lambda^{\mu\nu} + \Tr(Q_{\mu\nu}Q^{\mu\nu}) \Big) .
\end{align*}
The Higgs Lagrangian $\L_H$ (ignoring the boundary term) gives
\begin{multline}
\L_H(g_{\mu\nu}, \Lambda_\mu, Q_\mu, H') := \frac{bf(0)}{2\pi^2} |H'|^4 + \frac{-2af_2\Lambda^2 + ef(0)}{\pi^2} |H'|^2 \\
-\frac{cf_2\Lambda^2}{\pi^2} + \frac{df(0)}{4\pi^2} + \frac{af(0)}{12\pi^2} s |H'|^2 + \frac{cf(0)}{24\pi^2} s + \frac{af(0)}{2\pi^2} |\til D_\mu H'|^2 .
\end{multline}
\end{prop}
\begin{proof}
We will use the general form of the spectral action of an almost-commutative manifold as calculated in \cref{prop:acm_spec_act}. From \cref{lem:curv_GWS} we immediately find the term $\L_A$. Combining the formulas of $\Tr\big(\Phi^2\big)$ and $\Tr\big(\Phi^4\big)$ obtained in \cref{lem:Higgs_terms_GWS} we find the Higgs potential
\begin{multline*}
-\frac{f_2\Lambda^2}{2\pi^2} \Tr(\Phi^2) + \frac{f(0)}{8\pi^2} \Tr(\Phi^4) \\
= \frac{bf(0)}{2\pi^2} |H'|^4 + \frac{-2af_2\Lambda^2 + ef(0)}{\pi^2} |H'|^2 -\frac{cf_2\Lambda^2}{\pi^2} + \frac{df(0)}{4\pi^2} .
\end{multline*}
Note that the last two constant terms yield a contribution to the cosmological constant term $\frac{4f_4\Lambda^4}{\pi^2}$ that arises from $\L_M$. The coupling of the Higgs field to the scalar curvature $s$ is given by
\begin{align*}
\frac{f(0)}{48\pi^2} s\Tr(\Phi^2) = \frac{af(0)}{12\pi^2} s |H'|^2 + \frac{cf(0)}{24\pi^2} s .
\end{align*}
Here the second term yields a contribution to the Einstein-Hilbert term $- \frac{f_2\Lambda^2}{3\pi^2} s$ of $\L_M$. The last term is the kinetic term of the Higgs field including the minimal coupling to the gauge fields, obtained from \cref{lem:Higgs_kin_min_coupl_GWS}, which gives 
\begin{align*}
\frac{f(0)}{8\pi^2} \Tr\big((D_\mu\Phi)(D^\mu\Phi)\big) &= \frac{af(0)}{2\pi^2} |\til D_\mu H'|^2 .\qedhere
\end{align*}
\end{proof}

\subsection{Normalization of kinetic terms}

In \cref{prop:spec_act_GWS} we have calculated the bosonic Lagrangian. We will now rescale the gauge fields $\Lambda_\mu$ and $Q_\mu$ and the Higgs field $H'$ in such a way that their kinetic terms are properly normalized. 

\subsubsection{Rescaling the Higgs field} 

We start with the Higgs field $H'\rightarrow H$, and we will require that its kinetic term is normalized as
\begin{align*}
\int_M \frac12 |\til D_\mu H|^2 \sqrt{|g|} d^4x .
\end{align*}
This normalization is achieved by rescaling the Higgs field as
\begin{align}
\label{eq:Higgs_rescale}
H := \sqrt{\frac{af(0)}{\pi^2}} H' ,
\end{align}

\subsubsection{The coupling constants}
\label{sec:coupl_const_GWS}

Next, let us consider the gauge fields $\Lambda_\mu$ and $Q_\mu = Q_\mu^a \sigma^a$. We shall now introduce coupling constants $g_1$ and $g_2$ into the model by rescaling these fields as
\begin{align*}
\Lambda_\mu &= \frac12 g_1 B_\mu , & Q_\mu^a &= \frac12 g_2 W_\mu^a .
\end{align*}
Note that we use the conventional notation $B_\mu$ for the $U(1)$ hypercharge field, which should not be confused with the gauge field we introduced in \eqref{eq:fluc_Gauge}. We define the curvatures $B_{\mu\nu}$ and $W_{\mu\nu}$ by setting
\begin{align*} 
\Lambda_{\mu\nu} &= \frac12 g_1 B_{\mu\nu} , & Q_{\mu\nu}^a &= \frac12 g_2 W_{\mu\nu}^a .
\end{align*}
Using \eqref{eq:curv_GWS}, this yields
\begin{align*}
B_{\mu\nu} &= \partial_\mu B_\nu - \partial_\nu B_\mu ,\\
W_{\mu\nu}^a &= \partial_\mu W_\nu^a - \partial_\nu W_\mu^a - g_2 \epsilon^{abc} W_\mu^b W_\nu^c ,
\end{align*}
where we have used the relation $[\sigma^b,\sigma^c] = 2i \epsilon^{abc}\sigma^a$ for the Pauli matrices. We then rewrite the trace of the square of the curvature, given by \cref{lem:curv_GWS}, to give
\begin{align}
\label{eq:gauge_kin}
\Tr(F_{\mu\nu}F^{\mu\nu}) = 3 {g_1}^2 B_{\mu\nu} B^{\mu\nu} + {g_2}^2 W_{\mu\nu}^a W^{\mu\nu,a} ,
\end{align}
where we have used the relation $\Tr(\sigma^a \sigma^b) = 2 \delta^{ab}$. Note that the covariant derivative $\til D_\mu H$ can be written as
\begin{align}
\label{eq:Higgs_kin_gauge}
\til D_\mu H = \partial_\mu H + \frac12 i g_2 W_\mu^a \sigma^a H - \frac12 i g_1 B_\mu H .
\end{align}

\subsubsection{Electroweak unification}
\label{sec:EW_uni}

It would be natural to require that the kinetic terms of the gauge fields, given by the squares of the curvatures, are properly normalized. That is, we require that both these squares of the curvatures have the coefficient $\frac14$. This imposes the relations
\begin{align}
\label{eq:ew_uni_rel}
\frac{f(0)}{8\pi^2} {g_1}^2 = \frac14 \qquad\text{and}\qquad \frac{f(0)}{24\pi^2} {g_2}^2 = \frac14 .
\end{align}
This then means that the coupling constants are related by ${g_2}^2 = 3 {g_1}^2$. The values of the coupling constants depend on the energy scale at which they are evaluated, and their scale-dependence is determined by the renormalization group equations. Let $\Lambda\Sub{EW}$ be the scale at which the equality ${g_2}^2 = 3 {g_1}^2$ holds. Our model of the electroweak theory is then naturally defined at this scale $\Lambda\Sub{EW}$, and one could use the renormalization group equations to `run down' this model to lower energies. We will not provide the details here. Instead, we will discuss this renormalization scheme for the full Standard Model in \cref{chap:RGE}.

\subsection{The Higgs mechanism}
\label{sec:Higgs_GWS}
When writing down a gauge theory with massive gauge bosons, one encounters the difficulty that the mass terms of these gauge bosons are not gauge invariant. The Higgs field plays a central role in obtaining these mass terms within a gauge theory. The Higgs mechanism provides a \emph{spontaneous breaking} of the gauge symmetry. In this section we will describe how the Higgs mechanism breaks the $U(1)\times SU(2)$ symmetry and introduces mass terms for the gauge bosons. 

From \cref{prop:spec_act_GWS} we have obtained the Higgs Lagrangian $\L_H$. If we drop all the terms that are independent of the Higgs field $H$, we obtain the Lagrangian
\begin{align}
\label{eq:Higgs_pot_full}
\L(g_{\mu\nu}, B_\mu, W_\mu^a, H) := \frac{b\pi^2}{2a^2f(0)} |H|^4 - \frac{2af_2\Lambda^2 - ef(0)}{af(0)} |H|^2 + \frac{1}{12} s |H|^2 + \frac12 |\til D_\mu H|^2 .
\end{align}
We wish to find the value of $H$ for which this Lagrangian obtains its minimum value. In order to simplify the following discussion, we shall from here on assume that the scalar curvature $s$ vanishes identically. We may thus consider the Higgs potential
\begin{align*}
\L_\text{pot}(H) := \frac{b\pi^2}{2a^2f(0)} |H|^4 - \frac{2af_2\Lambda^2 - ef(0)}{af(0)} |H|^2 .
\end{align*}
If $2af_2\Lambda^2 < ef(0)$, the minimum of this potential is obtained for $H=0$, and in this case there will be no symmetry breaking. We shall now assume that $2af_2\Lambda^2 > ef(0)$. The minimum of the Higgs potential is then obtained if the field $H$ satisfies 
\begin{align}
\label{eq:Higgs_vac_GWS}
|H|^2 = \frac{2a^2f_2\Lambda^2 - aef(0)}{b\pi^2} .
\end{align} 
The fields that satisfy this relation are called the \emph{vacuum states} of the Higgs field. We shall choose a vacuum state $(v,0)$, where the \emph{vacuum expectation value} $v$ is a real parameter such that $v^2$ is given by \eqref{eq:Higgs_vac_GWS}. 

We want to simplify the expression for the Higgs potential. First, we note that the potential only depends on the absolute value $|H|$. A transformation of the doublet $H$ by an element $u\in U(1)\times SU(2)$ is written as $H\rightarrow uH$. Since a unitary transformation preserves the absolute value, we see that $\L_\text{pot}(uH) = \L_\text{pot}(H)$ for any $u\in U(1)\times SU(2)$. We can use this \emph{gauge freedom} to transform the Higgs field into a simpler form. Consider elements of $U(1)\times SU(2)$ of the form $$\mattwo{a}{-\bar b}{b}{\bar a}$$ such that $|a|^2 + |b|^2 = 1$. The doublet $H$ can in general be written as $(h_1,h_2)$, for some $h_1,h_2\in\C$. We then see that we can write
\begin{align*}
\vectwo{h_1}{h_2} = \mattwo{a}{-\bar b}{b}{\bar a} \vectwo{|H|}{0} , \qquad a = \frac{h_1}{|H|} , \quad b = \frac{h_2}{|H|} .
\end{align*}
This means that we can always use the gauge freedom to write the doublet $H$ in terms of one real parameter. Let us define a new real-valued field $h$ by setting $h(x) := |H(x)|-v$. We then obtain
\begin{align}
\label{eq:Higgs_transf_vac}
H = u(x) \vectwo{v+h(x)}{0}, \qquad u(x) := \mattwo{a(x)}{-\bar{b(x)}}{b(x)}{\bar{a(x)}} .
\end{align}
Inserting this transformed Higgs field into the Higgs potential, we obtain an expression in terms of the real parameter $v$ and the real field $h(x)$:
\begin{align*}
\L_\text{pot}(h) &= \frac{bf(0)}{2\pi^2} (v+h)^4 - \frac{2af_2\Lambda^2 - ef(0)}{\pi^2} (v+h)^2 \notag\\
&= \frac{b\pi^2}{2a^2f(0)} (h^4 + 4vh^3 + 6v^2h^2 + 4v^3h + v^4) \\
&\quad- \frac{2af_2\Lambda^2 - ef(0)}{af(0)} (h^2 + 2vh + v^2) .
\end{align*}
Using \eqref{eq:Higgs_vac_GWS}, the value of $v^2$ is given by
\begin{align*}
v^2 = \frac{2a^2f_2\Lambda^2 - aef(0)}{b\pi^2} .
\end{align*}
We then see that in $\L_\text{pot}$ the terms linear in $h$ cancel each other. This is of course no surprise, since the change of variables $|H(x)|\rightarrow v+h(x)$ means that at $h(x)=0$ we are in the minimum of the potential, where the first order derivative of the potential with respect to $h$ must vanish. We thus obtain the simplified expression
\begin{align}
\label{eq:Higgs_pot}
\L_\text{pot}(h) &= \frac{b\pi^2}{2a^2f(0)} \Big(h^4 + 4vh^3 + 4v^2h^2 - v^4\Big)  .
\end{align}
We now observe that the field $h(x)$ has obtained a mass term and has two self-interactions given by $h^3$ and $h^4$. We also have another contribution to the cosmological constant given by $-v^4$.

\subsubsection{Massive gauge bosons}

Next, let us consider what this procedure entails for the remainder of the Higgs Lagrangian $\L_H$. We first consider the kinetic term of $H$, including its minimal coupling to the gauge fields, given by 
\begin{align*}
\L_\text{min}(B_\mu, W_\mu^a, H) := \frac12 |\til D_\mu H|^2 .
\end{align*}
The transformation of \eqref{eq:Higgs_transf_vac} is a gauge transformation, and to make sure that $\L_\text{min}$ is invariant under this transformation, we also need to transform the gauge fields. The field $B_\mu$ is unaffected by the local $SU(2)$-transformation $u(x)$. The transformation of $W_\mu = W_\mu^a\sigma^a$ is obtained from \cref{prop:gauge_transf_GWS} and is given by
\begin{align*}
W_\mu \rightarrow u W_\mu u^* - \frac{2i}{g_2} u\partial_\mu u^* .
\end{align*}
One then easily checks that we obtain the transformation $\til D_\mu H \rightarrow u \til D_\mu H$, so that $|\til D_\mu H|^2$ is invariant under such transformations. So we can just insert the doublet $(v+h,0)$ into \eqref{eq:Higgs_kin_gauge} and obtain
\begin{align*}
\til D_\mu H &= \partial_\mu \vectwo{v+h}{0} + \frac12 i g_2 W_\mu^a \sigma^a \vectwo{v+h}{0} - \frac12 i g_1 B_\mu \vectwo{v+h}{0} \notag\\
&= \partial_\mu \vectwo{h}{0} + \frac12 i g_2 W_\mu^1 \vectwo{0}{v+h} + \frac12 i g_2 W_\mu^2 \vectwo{0}{i(v+h)} \\
&\quad+ \frac12 i g_2 W_\mu^3 \vectwo{v+h}{0} - \frac12 i g_1 B_\mu \vectwo{v+h}{0} .
\end{align*}
We can then calculate its square as
\begin{align*}
|\til D_\mu H|^2 &= (\til D^\mu H)^\dagger (\til D_\mu H) \notag\\
&= (\partial^\mu h)(\partial_\mu h) + \frac14 {g_2}^2 (v+h)^2 (W^{\mu,1} W_\mu^1 + W^{\mu,2} W_\mu^2 + W^{\mu,3} W_\mu^3) \notag\\
&\quad+ \frac14 {g_1}^2 (v+h)^2 B^\mu B_\mu - \frac12 g_1g_2 (v+h)^2 B^\mu W_\mu^3 .
\end{align*}
Note that the last term yields a mixing of the gauge fields $B_\mu$ and $W_\mu^3$. The electroweak mixing angle $\theta_w$ is defined by
\begin{align*}
c_w := \cos \theta_w &= \frac{g_2}{\sqrt{{g_1}^2+{g_2}^2}} , & s_w := \sin \theta_w &= \frac{g_1}{\sqrt{{g_1}^2+{g_2}^2}} .
\end{align*}
Note that the relation ${g_2}^2 = 3 {g_1}^2$ for the coupling constants implies that we obtain the values $\cos^2\theta_w = \frac14$ and $\sin^2\theta_w = \frac34$ at the electroweak unification scale $\Lambda\Sub{EW}$. Let us now define new gauge fields by
\begin{align}
\label{eq:gauge_eigenstates}
W_\mu &:= \frac{1}{\sqrt{2}} (W_\mu^1 + iW_\mu^2) , &
W_\mu^* &:= \frac{1}{\sqrt{2}} (W_\mu^1 - iW_\mu^2) , \notag\\
Z_\mu &:= c_w W_\mu^3 - s_w B_\mu , &
A_\mu &:= s_w W_\mu^3 + c_w B_\mu .
\end{align}
We will show that the new fields $Z_\mu$ and $A_\mu$ become mass eigenstates. The fields $W_\mu^1$ and $W_\mu^2$ already were mass eigenstates, but the fields $W_\mu$ and $W_\mu^*$ are chosen such that they obtain a definite charge. We can write
\begin{align*}
W_\mu^1 &= \frac{1}{\sqrt{2}} (W_\mu + W_\mu^*) , &
W_\mu^2 &= \frac{-i}{\sqrt{2}} (W_\mu - W_\mu^*) , \notag\\
W_\mu^3 &= s_w A_\mu + c_w Z_\mu , &
B_\mu &= c_w A_\mu - s_w Z_\mu ,
\end{align*}
and inserting this into the expression for $|\til D_\mu H|^2$ yields
\begin{align}
\label{eq:Higgs_kin}
\frac12 |\til D_\mu H|^2 = \frac12 (\partial^\mu h)(\partial_\mu h) + \frac14 {g_2}^2 (v+h)^2 W^{\mu*} W_\mu + \frac18 \frac{{g_2}^2}{{c_w}^2} (v+h)^2 Z^\mu Z_\mu .
\end{align}
We thus see that the fields $W_\mu$, $W_\mu^*$ and $Z_\mu$ acquire a mass term (where $Z_\mu$ has a larger mass than $W_\mu,W_\mu^*$) and that the field $A_\mu$ is massless. The masses of the $W$-boson and $Z$-boson are given by
\begin{align}
\label{eq:gauge_masses}
M_W &= \frac12 v g_2 , & M_Z &= \frac12 v \frac{g_2}{c_w} .
\end{align}

\begin{remark}
In the procedure described above we have assumed that the scalar curvature $s$ vanishes identically. Now suppose that the scalar curvature does not vanish, and consider the full Higgs potential of \eqref{eq:Higgs_pot_full}. Since the scalar curvature $s$ is a function on $M$, the vacuum expectation value of this full potential will in general not be a spacetime constant, and therefore we can no longer ignore the kinetic term in the Higgs potential. The vacuum expectation value $v$ will now be given by the solution to the equation
\begin{align*}
-\frac12 D_\mu D^\mu v(x) + \left(\frac{-2af_2\Lambda^2 + ef(0)}{af(0)} + \frac1{12} s(x) \right) v(x) + \frac{b\pi^2}{a^2f(0)} v(x)^3 = 0 .
\end{align*}
Unfortunately, this differential equation cannot be solved exactly, and this poses a problem in applying the Higgs mechanism. In \cref{chap:conf} we will propose a solution to this problem, by invoking the conformal invariance of the spectral action. 
\end{remark}

\newpage
\section{The Standard Model} 
\label{chap:ex_SM}

One of the major applications of noncommutative geometry to physics has been the derivation of the Standard Model of high energy physics from a suitably chosen almost-commutative manifold \cite{Connes96,CCM07} (see also \cite{CM07}). In \cref{chap:ex_GWS} we have already discussed the electroweak sector (for one generation) of the Standard Model. In this section we will also incorporate the quark sector with the strong interactions, and show that we obtain the full Standard Model.

\subsection{The finite space}
\label{sec:finite_SM}

The first description of the finite space yielding the Standard Model (without right-handed neutrinos) was given by Alain Connes in \cite{Connes96}. A newer version of this finite space was given in \cite{CCM07}, where now the finite space has KO-dimension $6$. This solved the problem of fermion doubling pointed out in \cite{LMMS97} (see also the discussion in \cite[Ch.~1, \S16.3]{CM07}), and at the same time allowed for the introduction of Majorana masses for right-handed neutrinos, along with the popular seesaw mechanism. 

In \cite{CCM07}, the starting point for the finite space is a left-right symmetric algebra $\A\Sub{LR}$. One then obtains a subalgebra $\A_F\subset\A\Sub{LR}$ by requiring that $\A_F$ should admit the Dirac operator $D_F$ to contain an off-diagonal part. A discussion of how the algebra $\A\Sub{LR}$ occurs naturally is given in \cite{CC08-why}. For the purpose of this section we will not go into these details. Instead, we simply state the finite space that will be used. Keeping in mind the previous sections and the fact that we now wish to obtain the Standard Model, the choices below should not be too mysterious. 

We take the finite space $F\Sub{GWS}$ of \cref{sec:GWS_finite_space} as our starting point. In order to incorporate the strong interactions, we add the $3\times3$ complex matrices $M_3(\C)$ to the algebra, and define
\begin{align*}
\A_F := \C \oplus \qH \oplus M_3(\C) . 
\end{align*}
We keep the Hilbert spaces $\mH_l=\C^4$ and $\mH_{\bar l}=\C^4$ for the description of the leptons and antileptons. For the quarks, we define $\mH_q = \C^4\otimes\C^3$, where the basis of $\C^4$ is given by $\{u_R,d_R,u_L,d_L\}$ and the three colors of the quarks are given by the factor $\C^3$. Similarly, we also have the antiquarks in $\mH_{\bar q}$. Combined, we obtain the $96$-dimensional Hilbert space for three generations of fermions and antifermions:
\begin{align*}
\mH_F := \left( \mH_l \oplus \mH_{\bar l} \oplus \mH_q \oplus \mH_{\bar q} \right)^{\oplus3} .
\end{align*}
An element of the algebra $\A_F$ is given by $a = (\lambda,q,m)$, where the quaternion $q$ can be embedded into $M_2(\C)$ as in \eqref{eq:quaternion_M2}. The action of an element $a$ on the space of leptons $\mH_l$ and the space of quarks $\mH_q$ is given just as in \eqref{eq:alg_rep_GWS} by
\begin{align*}
a = (\lambda,q,m) &\xrightarrow{\mH_l} \matfour{\lambda&0&0&0}{0&\bar\lambda&0&0}{0&0&\alpha&\beta}{0&0&-\bar\beta&\bar\alpha} , & 
a = (\lambda,q,m) &\xrightarrow{\mH_q} \matfour{\lambda&0&0&0}{0&\bar\lambda&0&0}{0&0&\alpha&\beta}{0&0&-\bar\beta&\bar\alpha} \otimes \1_3 .
\end{align*}
For the action of $a$ on an antilepton $\bar l\in\mH_{\bar l}$ we set $a\bar l = \lambda\1_4\bar l$, and on an antiquark $\bar q\in\mH_{\bar q}$ we set $a\bar q = (\1_4 \otimes m) \bar q$. 

The grading and the conjugation operator are also chosen in the same way as in \cref{sec:GWS_finite_space}. The grading $\gamma_F$ is such that all left-handed fermions have eigenvalue $+1$, and all right-handed fermions have eigenvalue $-1$. The conjugation operator $J_F$ interchanges a fermion with its antifermion. The Dirac operator $D_F$ is again of the form 
$$
\mattwo{S}{T^*}{T}{\bar S} .
$$
The operator $S$ is now given by
\begin{align*}
S_l &:= \left.S\right|_{\mH_l} = \matfour{0&0&Y_\nu&0}{0&0&0&Y_e}{Y_\nu^*&0&0&0}{0&Y_e^*&0&0} , & S_q \otimes \1_3 &:= \left.S\right|_{\mH_q} = \matfour{0&0&Y_u&0}{0&0&0&Y_d}{Y_u^*&0&0&0}{0&Y_d^*&0&0} \otimes\1_3 ,
\end{align*}
where $Y_\nu$, $Y_e$, $Y_u$ and $Y_d$ are $3\times3$ mass matrices acting on the three generations. The symmetric operator $T$ only acts on the right-handed (anti)neutrinos, so it is given by $T\nu_R = Y_R\bar{\nu_R}$ for a $3\times3$ symmetric Majorana mass matrix $Y_R$, and $Tf=0$ for all other fermions $f\neq\nu_R$. Note that $\nu_R$ here stands for a vector with $N=3$ components for the number of generations. 

\begin{prop}
\label{prop:spec_trip_SM}
The data 
\begin{align*}
F\Sub{SM} := \left( \A_F, \mH_F, D_F, \gamma_F, J_F \right)
\end{align*}
as given above define a real even finite space of KO-dimension $6$.
\end{prop}
\begin{proof}
The action of $\A_F$ on $\mH_{\bar q}$, given by $\1_4\otimes m$, commutes with all other operators, and hence it has no effect on the commutation relations. The proof is then the same as in \cref{prop:spec_trip_GWS}. 
\end{proof}

\subsection{The gauge theory} 
\label{sec:gauge_SM}

\subsubsection{The gauge group}

We shall now consider the almost-commutative manifold $M\times F\Sub{SM}$, and we wish to describe the gauge theory corresponding to $M\times F\Sub{SM}$. Let us start by examining the subalgebra $(\til\A_F)_{J_F}$ of the algebra $\A_F = \C\oplus\qH\oplus M_3(\C)$, as defined in \cref{sec:subalgs}. For an element $a=(\lambda,q,m)\in\C\oplus\qH\oplus M_3(\C)$, the relation $aJ_F=J_Fa^*$ now yields $\lambda=\bar\lambda=\alpha=\bar\alpha$ and $\beta=0$, as well as $m=\lambda\1_3$. So, $a\in(\til\A_F)_{J_F}$ if and only if $a=(x,x,x)$ for $x\in\R$. Hence we find that 
\begin{align*}
(\til\A_F)_{J_F} \simeq \R .
\end{align*}
Next, let us consider the Lie algebra $\h_F = \lu\big((\til\A_F)_{J_F}\big)$ of \eqref{eq:unitary_subLiealg_F}. Since $\lu(\A_F)$ consists of the anti-hermitian elements of $\A_F$, we again obtain as in \eqref{eq:h_F_triv} that the cross-section $\h_F = \lu\big((\til\A_F)_{J_F}\big)$ is given by the trivial subalgebra $\{0\}$. 

\begin{prop}
The local gauge group $\G(F\Sub{SM})$ of the finite space $F\Sub{SM}$ is given by
\begin{align*}
\G(F\Sub{SM}) \simeq \big(U(1)\times SU(2)\times U(3)\big) / \{1,-1\} .
\end{align*}
\end{prop}
\begin{proof}
As in \cref{prop:gauge_group_GWS}, we find that $U(\qH) = SU(2)$, so the unitary group $U(\A_F)$ is given by $U(1)\times SU(2)\times U(3)$. The subgroup $H_F = U\big((\til\A_F)_{J_F}\big)$ is again given by $H_F = \{1,-1\}$. By \cref{prop:acm_gauge}, the gauge group is given by the quotient of the unitary group with this subgroup.
\end{proof}

The gauge group that we obtain here is not the gauge group of the Standard Model, because (even ignoring the quotient with the finite group $\{1,-1\}$) we have a factor $U(3)$ instead of $SU(3)$. As mentioned in \cref{prop:unimod}, the unimodularity condition is only satisfied for complex algebras. In our case, the action of the algebra $\C\oplus\qH\oplus M_3(\C)$ on the Hilbert space $\mH_F$ is not complex-linear, since it involves complex conjugation. Therefore, the unimodularity condition is not satisfied. As in \cite[\S2.5]{CCM07} (see also \cite[Ch.~1, \S13.3]{CM07}), we shall now \emph{demand} that the unimodularity condition is satisfied, so for $u=(\lambda,q,m)\in U(1)\times SU(2)\times U(3)$ we require
\begin{align*}
\left.\det\right|_{\mH_F}(u) &= 1 \quad\Longrightarrow\quad \big(\lambda\det m\big)^{12} = 1 .
\end{align*}
For $u\in U(1)\times SU(2)\times U(3)$, we denote $U=uJuJ^*$ for the corresponding element in $\G(F\Sub{SM})$. We shall then consider the subgroup
\begin{align*}
S\G(F\Sub{SM}) = \left\{ U=uJuJ^* \in \G(F\Sub{SM}) \mid u = (\lambda,q,m),\; \big(\lambda\det m\big)^{12} = 1 \right\} .
\end{align*}
The effect of the unimodularity condition is that the determinant of $m\in U(3)$ is identified (modulo the finite group $\mu_{12}$ of $12$th-roots of unity) to $\bar\lambda$. In other words, imposing the unimodularity condition provides us, modulo some finite abelian group, with the gauge group $U(1)\times SU(2) \times SU(3)$. Let us go into a little more detail, following (but slightly modifying) \cite[Prop.\ 2.16]{CCM07} (see also \cite[Prop.\ 1.185]{CM07}). The group $U(1)\times SU(2) \times SU(3)$ is actually not the true gauge group of the Standard Model, since it contains a finite abelian subgroup (isomorphic to) $\mu_6$ which acts trivially on all bosonic and fermionic particles in the Standard Model (see for instance \cite{BH10}). 
The group $\mu_6$ is embedded in $U(1)\times SU(2) \times SU(3)$ by $\lambda \mapsto (\lambda, \lambda^3, \lambda^2)$. The true gauge group of the Standard Model is then given by
\begin{align*}
\G\Sub{SM} := U(1)\times SU(2) \times SU(3) / \mu_6 . 
\end{align*}

\begin{prop}
The unimodular gauge group $S\G(F\Sub{SM})$ is isomorphic to 
\begin{align*}
S\G(F\Sub{SM}) \simeq \G\Sub{SM} \rtimes \mu_{12} .
\end{align*}
\end{prop}
\begin{proof}
We define the homomorphism $\rho\colon S\G(F\Sub{SM}) \rightarrow \mu_{12}$ by setting $\rho(U) = \lambda\det m$. The kernel of $\rho$ is given by
\begin{align*}
\Ker(\rho) = \{ U=uJuJ^* \in \G(F\Sub{SM}) \mid u = (\lambda,q,m),\; \lambda\det m = 1 \} .
\end{align*}
The homomorphism $\varphi\colon U(1)\times SU(2)\times SU(3) \rightarrow S\G(F\Sub{SM})$ is given by 
$$
\varphi(\lambda,q,m) = (\lambda^3,q,\lambda^{-1} \bar m)J(\lambda^3,q,\lambda^{-1} \bar m)J^* .
$$
We observe that 
$$
\rho\big(\varphi(\lambda,q,m)\big) = \rho\big((\lambda^3,q,\lambda^{-1} \bar m)J(\lambda^3,q,\lambda^{-1} \bar m)J^*\big) = \lambda^3\det(\lambda^{-1} \bar m) = \det \bar m = 1,
$$
so that $\varphi$ indeed maps into $S\G(F\Sub{SM})$, and we obtain that $\Im(\varphi)=\Ker(\rho)$. 

The kernel of $\varphi$ is given by all $(\lambda,q,m)$ for which $(\lambda^3,q,\lambda^{-1} \bar m) = \pm1$. This implies that $\lambda^3=\pm1$, and thus $q=\lambda^3\1_2$ and $m=\lambda^2\1_3$. The requirement $\lambda^3=\pm1$ implies $\lambda\in\mu_6$, so we obtain that $\Ker(\varphi) = \left\{ (\lambda,\lambda^3,\lambda^2) \mid \lambda\in\mu_6 \right\} \simeq \mu_6$. Hence, the map $\til\varphi\colon \G\Sub{SM} \rightarrow S\G(F\Sub{SM})$ induced by $\varphi$ is an injective group homomorphism. Since $\G\Sub{SM} \simeq \Im(\til\varphi) = \Ker(\rho)$, we see that in fact $\G\Sub{SM}$ is embedded as a normal subgroup of $S\G(F\Sub{SM})$, and the quotient $S\G(F\Sub{SM}) / \G\Sub{SM}$ is then isomorphic to $\Im(\rho) = \mu_{12}$. 
\end{proof}

\subsubsection{The gauge fields and the Higgs field}

Let us now derive the precise form of the gauge field $A_\mu$ of \eqref{eq:fluc_gauge} and the Higgs field $\phi$ of \eqref{eq:fluc_higgs}. The calculations are similar to those in \cref{sec:fields_GWS}, and the formulas for $\Lambda_\mu$ and $Q_\mu$ follow immediately from \eqref{eq:fields_GWS}. The Higgs field $\phi$ is slightly different, and is now given by
\begin{align}
\label{eq:higgs_field_SM}
\left.\phi\right|_{\mH_l} &= \mattwo{0}{Y^*}{Y}{0} , & \left.\phi\right|_{\mH_q} &= \mattwo{0}{X^*}{X}{0} \otimes\1_3 , & \left.\phi\right|_{\mH_{\bar l}} &= 0 , & \left.\phi\right|_{\mH_{\bar q}} &= 0 ,
\end{align}
where, for $\phi_1,\phi_2\in\C$, we now have
\begin{align*}
Y &= \mattwo{Y_\nu\phi_1}{-Y_e\bar\phi_2}{Y_\nu\phi_2}{Y_e\bar\phi_1} , & X &= \mattwo{Y_u\phi_1}{-Y_d\bar\phi_2}{Y_u\phi_2}{Y_d\bar\phi_1} . 
\end{align*}
The Higgs field $\Phi$ is then given as in \eqref{eq:Higgs_field_GWS} by
\begin{align}
\label{eq:Higgs_field_SM}
\Phi = D_F + \mattwo{\phi}{0}{0}{0} + J_F\mattwo{\phi}{0}{0}{0}J_F^* = \mattwo{S+\phi}{T^*}{T}{\bar{(S+\phi)}} ,
\end{align}
The biggest difference is the occurrence of a field $V_\mu' := -im\partial_\mu m'$, acting on $\mH_{\bar q}$, for $m,m'\in M_3(\C)$. Demanding $V_\mu'$ to be hermitian yields $V_\mu' \in i\,\lu(3)$, so $V_\mu'$ is a $U(3)$ gauge field instead of an $SU(3)$ gauge field. As mentioned above, we need to impose the unimodularity condition to obtain an $SU(3)$ gauge field. Hence, we require that the trace of the gauge field $A_\mu$ over $\mH_F$ vanishes, and we obtain
\begin{align*}
\left.\Tr\right|_{\mH_{\bar l}}\big( \Lambda_\mu \1_4 \big) + \left.\Tr\right|_{\mH_{\bar q}}\big( \1_4\otimes V_\mu' \big) = 0  \quad\Longrightarrow\quad  \Tr(V_\mu') = - \Lambda_\mu .
\end{align*}
So, we can define a traceless $SU(3)$ gauge field $V_\mu$ by $\bar V_\mu := - V_\mu' - \frac13 \Lambda_\mu$. The gauge field $A_\mu$ is thus given by 
\begin{align*}
\left.A_\mu\right|_{\mH_l} &= \matthree{\Lambda_\mu&0&}{0&-\Lambda_\mu&}{&&Q_\mu} , & \left.A_\mu\right|_{\mH_q} &= \matthree{\Lambda_\mu&0&}{0&-\Lambda_\mu&}{&&Q_\mu} \otimes\1_3 , \notag\\
\left.A_\mu\right|_{\mH_{\bar l}} &= \Lambda_\mu\1_4 , & \left.A_\mu\right|_{\mH_{\bar q}} &= - \1_4 \otimes (\bar V_\mu+\frac13\Lambda_\mu) , 
\end{align*}
for a $U(1)$ gauge field $\Lambda_\mu$, an $SU(2)$ gauge field $Q_\mu$ and an $SU(3)$ gauge field $V_\mu$. The action of the field $B_\mu = A_\mu - J_FA_\mu J_F^{-1}$ on the fermions is then given by
\begin{align}
\left.B_\mu\right|_{\mH_l} &= \matthree{0&0&}{0&-2\Lambda_\mu&}{&&Q_\mu-\Lambda_\mu\1_2} , \notag\\
\label{eq:Gauge_field_SM}
\left.B_\mu\right|_{\mH_q} &= \matthree{\frac43\Lambda_\mu\1_3+V_\mu&0&}{0&-\frac23\Lambda_\mu\1_3+V_\mu&}{&&(Q_\mu+\frac13\Lambda_\mu\1_2)\otimes\1_3+\1_2\otimes V_\mu} .
\end{align}
Note that the coefficients in front of $\Lambda_\mu$ in the above formulas, are precisely the well-known hypercharges of the corresponding particles, as given by the following table: 
\begin{align*}
\begin{array}{l|cccccccc}
\text{Particle} & \nu_R & e_R & \nu_L & e_L & u_R & d_R & u_L & d_L \\
\hline
\text{Hypercharge} & 0 & -2 & -1 & -1 & \frac43 & -\frac23 & \frac13 & \frac13 \\
\end{array}
\end{align*}
\begin{prop}
\label{prop:gauge_transf_SM}
The action of the gauge group $S\G(M\times F\Sub{SM})$ on the fluctuated Dirac operator
\begin{align*}
D_A = \sD\otimes\1 + \gamma^\mu\otimes B_\mu + \gamma_5\otimes\Phi
\end{align*}
is implemented by
\begin{gather*}
\Lambda_\mu \rightarrow \Lambda_\mu - i \lambda\partial_\mu\bar\lambda , \qquad
Q_\mu \rightarrow qQ_\mu q^* - iq\partial_\mu q^* , \qquad
\bar V_\mu \rightarrow m\bar V_\mu m^* - im\partial_\mu m^* , \\
\vectwo{\phi_1+1}{\phi_2} \rightarrow \bar\lambda\,q \vectwo{\phi_1+1}{\phi_2} ,
\end{gather*}
for $\lambda\in C^\infty\big(M,U(1)\big)$, $q\in C^\infty\big(M,SU(2)\big)$ and $m\in C^\infty\big(M,SU(3)\big)$. 
\end{prop}
\begin{proof}
The proof is similar to \cref{prop:gauge_transf_GWS}. Let us write $u=(\lambda,q,m)\in C^\infty\big(M,U(1)\times SU(2)\times SU(3)\big)$. The term $uAu^*$ now not only replaces $Q_\mu$ by $qQ_\mu q^*$, but also $\bar V_\mu$ by $m\bar V_\mu m^*$. Secondly, we see that the term $-iu\partial_\mu u^*$ is given by $- i \lambda\partial_\mu\bar\lambda$ on $\nu_R$, $u_R$ and $\mH_{\bar l}$, by $- i \bar\lambda\partial_\mu\lambda = i \lambda\partial_\mu\bar\lambda$ on $e_R$ and $d_R$, by $- iq\partial_\mu q^*$ on $(\nu_L,e_L)$ and $(u_L,d_L)$, and finally by $- im\partial_\mu m^*$ on $\mH_{\bar q}$. We thus obtain the desired transformation for $\Lambda_\mu$, $Q_\mu$ and $\bar V_\mu$. The transformation of $\phi$ is exactly as in \cref{prop:gauge_transf_GWS}.
\end{proof}

\subsection{The spectral action}

In this section we will calculate the bosonic part of the Lagrangian of the Standard Model from the spectral action. The general form of this Lagrangian has already been calculated in \cref{prop:acm_spec_act} so we only need to insert the expressions \eqref{eq:Gauge_field_SM,eq:Higgs_field_SM} for the fields $\Phi$ and $B_\mu$. As in \cref{sec:spec_act_GWS}, we first start with a few lemmas, in which we capture the rather tedious calculations that are needed to obtain the traces of $F_{\mu\nu}F^{\mu\nu}$, $\Phi^2$, $\Phi^4$ and $(D_\mu\Phi)(D^\mu\Phi)$. 

\begin{lem}
\label{lem:curv_SM}
The trace of the square of the curvature of $B_\mu$ is given by
\begin{align*}
\Tr(F_{\mu\nu}F^{\mu\nu}) = 24\Big(\frac{10}3 \Lambda_{\mu\nu}\Lambda^{\mu\nu} + \Tr(Q_{\mu\nu}Q^{\mu\nu}) + \Tr(V_{\mu\nu}V^{\mu\nu}) \Big) . 
\end{align*}
\end{lem}
\begin{proof}
The lepton sector yields the same result as in \cref{lem:curv_GWS}, only multiplied by a factor $3$ for the number of generations. For the quark sector, we obtain on $\mH_q$ the curvature
\begin{align*}
\left.F_{\mu\nu}\right|_{\mH_q} &= \matthree{\frac43\Lambda_{\mu\nu}\1_3+V_{\mu\nu}&0&}{0&-\frac23\Lambda_{\mu\nu}\1_3+V_{\mu\nu}&}{&&(Q_{\mu\nu}+\frac13\Lambda_{\mu\nu}\1_2)\otimes\1_3+\1_2\otimes V_{\mu\nu}} ,
\end{align*}
where we have now defined the curvature of the $SU(3)$ gauge field by
\begin{align*}
V_{\mu\nu} &:= \partial_\mu V_\nu - \partial_\nu V_\mu + i[V_\mu,V_\nu] .
\end{align*}
If we calculate the trace of the square of the curvature $F_{\mu\nu}$, the cross-terms again vanish, so we obtain
\begin{align*}
\left.\Tr\right|_{\mH_q}(F_{\mu\nu}F^{\mu\nu}) &= \left(\frac{16}{3} + \frac43 + \frac13 + \frac13\right) \Lambda_{\mu\nu}\Lambda^{\mu\nu} + 3 \Tr(Q_{\mu\nu}Q^{\mu\nu}) + 4 \Tr(V_{\mu\nu}V^{\mu\nu}) .
\end{align*}
We multiply this by a factor $2$ to include the trace over the antiquarks, and by a factor $3$ for the number of generations. Adding the result to the trace over the lepton sector, we finally obtain
\begin{align*}
\Tr(F_{\mu\nu}F^{\mu\nu}) &= 80\Lambda_{\mu\nu}\Lambda^{\mu\nu} + 24 \Tr(Q_{\mu\nu}Q^{\mu\nu}) + 24 \Tr(V_{\mu\nu}V^{\mu\nu}) . \qedhere
\end{align*}
\end{proof}

\begin{lem}
\label{lem:Higgs_terms_SM}
The traces of $\Phi^2$ and $\Phi^4$ are given by
\begin{align*}
\Tr\big(\Phi^2\big) &= 4a |H'|^2 +2c , \\
\Tr\big(\Phi^4\big) &= 4b|H'|^4 + 8e|H'|^2 + 2 d ,
\end{align*}
where $H'$ denotes the complex doublet $(\phi_1+1,\phi_2)$ and, following \cite{CCM07} (see also \cite[Ch.~1, \S15.2]{CM07}), 
\begin{align}
\label{eq:abcde_SM}
a &= \Tr\big(Y_\nu^*Y_\nu + Y_e^*Y_e + 3Y_u^*Y_u + 3Y_d^*Y_d\big) , \notag\\
b &= \Tr\big((Y_\nu^*Y_\nu)^2 + (Y_e^*Y_e)^2 + 3(Y_u^*Y_u)^2 + 3(Y_d^*Y_d)^2\big) , \notag\\
c &= \Tr\big(Y_R^*Y_R\big) , \\
d &= \Tr\big((Y_R^*Y_R)^2\big) , \notag\\
e &= \Tr\big(Y_R^*Y_R Y_\nu^*Y_\nu\big) . \notag
\end{align}
\end{lem}
\begin{proof}
The proof is analogous to \cref{lem:Higgs_terms_GWS}, where the coefficients $a,b,c,d,e$ have now been redefined to incorporate the quark sector, and the trace is taken over the generation space. 
\end{proof}

\begin{lem}
\label{lem:Higgs_kin_min_coupl_SM}
The trace of $(D_\mu\Phi)(D^\mu\Phi)$ is given by
\begin{align*}
\Tr\big((D_\mu\Phi)(D^\mu\Phi)\big) = 4 a |\til D_\mu H'|^2 ,
\end{align*}
where $H'$ denotes the complex doublet $(\phi_1+1,\phi_2)$, and the covariant derivative $\til D_\mu$ on $H'$ is defined as
\begin{align*}
\til D_\mu H' = \partial_\mu H' + i Q_\mu^a \sigma^a H' - i \Lambda_\mu H' .
\end{align*}
\end{lem}
\begin{proof}
The proof is as in \cref{lem:Higgs_kin_min_coupl_GWS}. Since $\Phi$ commutes with the gauge field $V_\mu$, this gauge field does not contribute to the covariant derivative $\til D_\mu$. 
\end{proof}

\begin{prop}
\label{prop:spec_act_SM}
The spectral action of the AC-manifold $M\times F\Sub{SM}$ defined in this section is given by
\begin{align*}
\Tr \left(f\Big(\frac {D_A}\Lambda\Big)\right) &\sim \int_M \L(g_{\mu\nu}, \Lambda_\mu, Q_\mu, V_\mu, H') \sqrt{|g|} d^4x + O(\Lambda^{-1}) ,
\end{align*}
for the Lagrangian
\begin{align*}
\L(g_{\mu\nu}, \Lambda_\mu, Q_\mu, V_\mu, H') := 96\L_M(g_{\mu\nu}) + \L_A(\Lambda_\mu,Q_\mu,V_\mu) + \L_H(g_{\mu\nu}, \Lambda_\mu, Q_\mu, H') .
\end{align*}
Here $\L_M(g_{\mu\nu})$ is defined in \cref{prop:canon_spec_act}. The term $\L_A$ gives the kinetic terms of the gauge fields and equals 
\begin{align*}
\L_A(\Lambda_\mu,Q_\mu,V_\mu) := \frac{f(0)}{\pi^2} \Big( \frac{10}3\Lambda_{\mu\nu}\Lambda^{\mu\nu} + \Tr(Q_{\mu\nu}Q^{\mu\nu}) + \Tr(V_{\mu\nu}V^{\mu\nu}) \Big) .
\end{align*}
The Higgs potential $\L_H$ (ignoring the boundary term) gives
\begin{multline}
\L_H(g_{\mu\nu}, \Lambda_\mu, Q_\mu, H') := \frac{bf(0)}{2\pi^2} |H'|^4 + \frac{-2af_2\Lambda^2 + ef(0)}{\pi^2} |H'|^2 \\
-\frac{cf_2\Lambda^2}{\pi^2} + \frac{df(0)}{4\pi^2} + \frac{af(0)}{12\pi^2} s |H'|^2 + \frac{cf(0)}{24\pi^2} s + \frac{af(0)}{2\pi^2} |\til D_\mu H'|^2 .
\end{multline}
\end{prop}
\begin{proof}
We will use the general form of the spectral action of an almost-commutative manifold as calculated in \cref{prop:acm_spec_act}. The gravitational Lagrangian $\L_M$ now obtains a factor $96$ from the trace over $\mH_F$. From \cref{lem:curv_SM} we immediately find the term $\L_A$. For the newly defined coefficients $a,b,c,d,e$ of \eqref{eq:abcde_SM}, the Higgs potential has exactly the same form as in \cref{prop:spec_act_GWS}.
\end{proof}

\subsubsection{The coupling constants and unification}
\label{sec:couplings_SM}

The $SU(3)$ gauge field $V_\mu$ can be written as $V_\mu = V_\mu^i \lambda^i$, for the Gell-Mann matrices $\lambda^i$ and real coefficients $V_\mu^i$. As in \cref{sec:coupl_const_GWS}, we will introduce coupling constants into the model by rescaling the gauge fields as
\begin{align*}
\Lambda_\mu &= \frac12 g_1 B_\mu , & Q_\mu^a &= \frac12 g_2 W_\mu^a , & V_\mu^i &= \frac12 g_3 G_\mu^i .
\end{align*}
By using the relations $\Tr(\sigma^a \sigma^b) = 2 \delta^{ab}$ and $\Tr(\lambda^i \lambda^j) = 2 \delta^{ij}$, we now find that the Lagrangian $\L_A$ of \cref{prop:spec_act_SM} can be written as
\begin{align*}
\L_A(B_\mu,W_\mu,G_\mu) = \frac{f(0)}{2\pi^2} \Big( \frac53 {g_1}^2 B_{\mu\nu}B^{\mu\nu} + {g_2}^2 W_{\mu\nu}W^{\mu\nu} + {g_3}^2 G_{\mu\nu}G^{\mu\nu} \Big) .
\end{align*}
It is natural to require that these kinetic terms are properly normalized, and this imposes the relations
\begin{align}
\label{eq:couplings_norm}
\frac{f(0)}{2\pi^2} {g_3}^2 = \frac{f(0)}{2\pi^2} {g_2}^2 = \frac{5f(0)}{6\pi^2} {g_1}^2 = \frac14 .
\end{align}
The coupling constants are then related by 
\begin{align*}
{g_3}^2 = {g_2}^2 = \frac53 {g_1}^2 ,
\end{align*}
which is precisely the relation between the coupling constants at unification, common to grand unified theories (GUT). We shall discuss this further in \cref{sec:renorm}. 

By rescaling the Higgs field $H'\rightarrow H$ as in \eqref{eq:Higgs_rescale}, we obtain the following result:
\begin{thm}
\label{thm:spec_act_SM}
The spectral action (ignoring topological and boundary terms) of the AC-manifold $M\times F\Sub{SM}$ is given by 
\begin{align*}
S_B = \int_M &\Bigg( \frac{48f_4\Lambda^4}{\pi^2} - \frac{cf_2\Lambda^2}{\pi^2} + \frac{df(0)}{4\pi^2} + \left(\frac{cf(0)}{24\pi^2} - \frac{4f_2\Lambda^2}{\pi^2} \right) s - \frac{3f(0)}{10\pi^2} C_{\mu\nu\rho\sigma} C^{\mu\nu\rho\sigma} \notag\\
&\quad+ \frac14 B_{\mu\nu} B^{\mu\nu} + \frac14 W_{\mu\nu}^a W^{\mu\nu,a} + \frac14 G_{\mu\nu}^i G^{\mu\nu,i} + \frac{b\pi^2}{2a^2f(0)} |H|^4 \notag\\
&\quad- \frac{2af_2\Lambda^2 - ef(0)}{af(0)} |H|^2 + \frac{1}{12} s |H|^2 + \frac12 |\til D_\mu H|^2 \Bigg) \sqrt{|g|} d^4x .
\end{align*}
\end{thm}

\subsection{The fermionic action}

In order to obtain the full Lagrangian for the Standard Model, we also need to calculate the fermionic action $S_f$ of \cref{defn:act_funct}. First, let us have a closer look at the fermionic particle fields and their interactions. 

By an abuse of notation, let us write $\nu^\lambda, \bar \nu^\lambda, e^\lambda, \bar e^\lambda, u^{\lambda c}, \bar u^{\lambda c}, d^{\lambda c}, \bar d^{\lambda c}$ for a set of independent anticommuting Dirac spinors. We then write a generic Grassmann vector $\til\xi\in\mH^+_\text{cl}$ as follows:
\begin{align*}
\til\xi &= \nu^\lambda_L\otimes \nu^\lambda_L + \nu^\lambda_R\otimes \nu^\lambda_R + \bar \nu^\lambda_R\otimes\bar{\nu^\lambda_L} + \bar \nu^\lambda_L\otimes\bar{\nu^\lambda_R} \notag\\
&\qquad + e^\lambda_L\otimes e^\lambda_L + e^\lambda_R\otimes e^\lambda_R + \bar e^\lambda_R\otimes\bar{e^\lambda_L} + \bar e^\lambda_L\otimes\bar{e^\lambda_R} \notag\\
&\qquad + u^{\lambda c}_L\otimes u^{\lambda c}_L + u^{\lambda c}_R\otimes u^{\lambda c}_R + \bar u^{\lambda c}_R\otimes\bar{u^{\lambda c}_L} + \bar u^{\lambda c}_L\otimes\bar{u^{\lambda c}_R} \notag\\
&\qquad + d^{\lambda c}_L\otimes d^{\lambda c}_L + d^{\lambda c}_R\otimes d^{\lambda c}_R + \bar d^{\lambda c}_R\otimes\bar{d^{\lambda c}_L} + \bar d^{\lambda c}_L\otimes\bar{d^{\lambda c}_R} ,
\end{align*}
where in each tensor product it should be clear that the first component is an anti-commuting Weyl spinor, and the second component is a basis element of $\mH_F$. Here $\lambda=1,2,3$ labels the generation of the fermions, and $c=r,g,b$ labels the color index of the quarks. 

Let us have a closer look at the gauge fields of the electroweak sector. For the physical gauge fields of \eqref{eq:gauge_eigenstates} we can write
\begin{gather}
\label{eq:gauge_expr}
\begin{aligned}
Q_\mu^1 + iQ_\mu^2 &= \frac{1}{\sqrt{2}} g_2 W_\mu , &\qquad
Q_\mu^1 - iQ_\mu^2 &= \frac{1}{\sqrt{2}} g_2 W_\mu^* , \\
Q_\mu^3 - \Lambda_\mu &= \frac{g_2}{2c_w} Z_\mu , &
\Lambda_\mu &= \frac12 s_wg_2 A_\mu - \frac12 \frac{{s_w}^2g_2}{c_w} Z_\mu , 
\end{aligned}\notag\displaybreak[0]\\
\begin{aligned}
- Q_\mu^3 - \Lambda_\mu &= - s_wg_2 A_\mu + \frac{g_2}{2c_w} (1-2{c_w}^2) Z_\mu, \\
Q_\mu^3 + \frac13 \Lambda_\mu &= \frac23 s_wg_2 A_\mu - \frac{g_2}{6c_w} (1-4{c_w}^2) Z_\mu , \\
- Q_\mu^3 + \frac13 \Lambda_\mu &= -\frac13 s_wg_2 A_\mu - \frac{g_2}{6c_w} (1+2{c_w}^2) Z_\mu . 
\end{aligned}
\end{gather}

We have rescaled the Higgs field in \eqref{eq:Higgs_rescale}, so we can write $H = \frac{\sqrt{af(0)}}{\pi} (\phi_1+1,\phi_2)$. We shall parametrize the Higgs field as $H = (v+h+i\phi^0,i\sqrt{2}\phi^-)$, where $\phi^0$ is real and $\phi^-$ is complex. We write $\phi^+$ for the complex conjugate of $\phi^-$. Thus, we can write
\begin{align}
\label{eq:Higgs_rewrite}
(\phi_1+1,\phi_2) = \frac{\pi}{\sqrt{af(0)}} (v+h+i\phi^0,i\sqrt{2}\phi^-) .
\end{align}

As in \cref{remark:real_mass}, we will need to impose a further restriction on the mass matrices in $D_F$, in order to obtain physical mass terms in the fermionic action. From here on, we will require that the matrices $Y_x$ are antihermitian, for $x = \nu,e,u,d$. We shall then define the hermitian mass matrices $m_x$ by writing 
\begin{align}
\label{eq:mass_rewrite}
Y_x =: -i \frac{\sqrt{af(0)}}{\pi v} m_x .
\end{align}
Similarly, we shall also take $Y_R$ to be anti-hermitian, and we introduce a hermitian (and symmetric) Majorana mass matrix $m_R$ by writing
\begin{align}
\label{eq:Maj_mass_rewrite}
Y_R = - i \, m_R .
\end{align}

\begin{thm}
\label{thm:ferm_act_SM}
The fermionic action of the almost-commutative manifold $M\times F\Sub{SM}$ is given by
\begin{align*}
S_F &= \int_M \big( \L_\text{kin} + \L_{gf} + \L_{Hf} + \L_R \big) \sqrt{|g|} d^4x .
\end{align*}
We suppress all generation and color indices. The kinetic terms of the fermions are given by
\begin{align*} 
\L_\text{kin} := - i ( J_M\bar\nu,\gamma^\mu\nabla^S_\mu\nu) - i ( J_M\bar e,\gamma^\mu\nabla^S_\mu e) - i ( J_M\bar u,\gamma^\mu\nabla^S_\mu u) - i ( J_M\bar d,\gamma^\mu\nabla^S_\mu d) .
\end{align*}
The minimal coupling of the gauge fields to the fermions is given by
\begin{align*}
\L_{gf} &:= s_wg_2 A_\mu \Big( - ( J_M\bar e, \gamma^\mu e) + {\textstyle \frac23} ( J_M\bar u, \gamma^\mu u) - {\textstyle \frac13} ( J_M\bar d, \gamma^\mu d)\Big) \notag\\
&\quad + \frac{g_2}{4c_w} Z_\mu \begin{aligned}[t] &\Big(( J_M\bar\nu,\gamma^\mu(1+\gamma_5)\nu) + ( J_M\bar e, \gamma^\mu (4{s_w}^2-1-\gamma_5)e) \\
	& + ( J_M\bar u, \gamma^\mu (-{\textstyle \frac83}{s_w}^2+1+\gamma_5)u) + ( J_M\bar d, \gamma^\mu ({\textstyle \frac43}{s_w}^2-1-\gamma_5)d) \Big) \end{aligned} \notag\\
&\quad + \frac{g_2}{2\sqrt{2}} W_\mu \Big( ( J_M\bar e, \gamma^\mu (1+\gamma_5)\nu) + ( J_M\bar d,\gamma^\mu (1+\gamma_5)u) \Big) \notag\\
&\quad + \frac{g_2}{2\sqrt{2}} W_\mu^* \Big( ( J_M\bar\nu,\gamma^\mu (1+\gamma_5)e) + ( J_M\bar u,\gamma^\mu (1+\gamma_5)d) \Big) \notag\\
&\quad + \frac{g_3}{2} G_\mu^i \Big( ( J_M\bar u, \gamma^\mu \lambda_i u) + ( J_M\bar d, \gamma^\mu \lambda_i d) \Big) .
\end{align*}
The Yukawa couplings of the Higgs field to the fermions are given by
\begin{align*}
\L_{Hf} &:= i(1+\frac{h}{v}) \Big( ( J_M \bar\nu, m_\nu\nu) + ( J_M \bar e, m_e e) + ( J_M \bar u, m_u u) + ( J_M \bar d, m_d d) \Big) \notag\displaybreak[0]\\
&\quad + \frac{\phi^0}{v} \Big( ( J_M \bar\nu,\gamma_5 m_\nu\nu) - ( J_M \bar e,\gamma_5 m_e e) + ( J_M \bar u,\gamma_5 m_u u) - ( J_M \bar d,\gamma_5 m_d d) \Big) \notag\displaybreak[0]\\
&\quad + \frac{\phi^-}{\sqrt{2}v} \Big( ( J_M \bar e, m_e (1+\gamma_5) \nu) - ( J_M \bar e, m_\nu (1-\gamma_5) \nu) \Big) \notag\displaybreak[0]\\
&\quad + \frac{\phi^+}{\sqrt{2}v} \Big( ( J_M \bar\nu, m_\nu (1+\gamma_5) e) - ( J_M \bar \nu, m_e (1-\gamma_5) e) \Big) \notag\displaybreak[0]\\
&\quad + \frac{\phi^-}{\sqrt{2}v} \Big( ( J_M \bar d, m_d (1+\gamma_5) u) - ( J_M \bar d, m_u (1-\gamma_5) u) \Big) \notag\displaybreak[0]\\
&\quad + \frac{\phi^+}{\sqrt{2}v} \Big( ( J_M \bar u, m_u (1+\gamma_5) d) - ( J_M \bar u, m_d (1-\gamma_5) d) \Big) .
\end{align*}
The Majorana masses of the right-handed neutrinos (and left-handed anti-neutrinos) are given by
\begin{align*}
\L_R := i ( J_M \nu_R,m_R\nu_R) + i ( J_M \bar\nu_L,m_R\bar\nu_L) .
\end{align*}
\end{thm}
\begin{proof}
The proof is similar to \cref{prop:fermion_act_ED}, though the calculations are now a little more complicated. From \cref{defn:act_funct} we know that the fermionic action is given by $S_F = \frac12 \langle J\til\xi,D_A\til\xi\rangle$, where the fluctuated Dirac operator is given by 
\begin{align*}
D_A = \sD\otimes\1 + \gamma^\mu\otimes B_\mu + \gamma_5\otimes \Phi .
\end{align*}
We rewrite the inner product on $\mH$ as $\langle\xi,\psi\rangle = \int_M (\xi,\psi) \sqrt{|g|} d^4x$. As in \cref{prop:fermion_act_ED}, the expressions for $J\til\xi = (J_M\otimes J_F)\til\xi$ and $(\sD\otimes\1)\til\xi$ are obtained straightforwardly. We will use the symmetry of the form $(J_M\til\chi,\sD\til\psi)$, and then we obtain the kinetic terms as
\begin{align*}
\frac12 ( J\til\xi,(\sD\otimes\1)\til\xi) &= 
( J_M\bar\nu^\lambda,\sD\nu^\lambda) + ( J_M\bar e^\lambda,\sD e^\lambda) + ( J_M\bar u^{\lambda c},\sD u^{\lambda c}) + ( J_M\bar d^{\lambda c},\sD d^{\lambda c}) .
\end{align*}
The other two terms in the fluctuated Dirac operator yield more complicated expressions. For the calculation of $(\gamma^\mu\otimes B_\mu)\til\xi$, we use \cref{eq:Gauge_field_SM} for the gauge field $B_\mu$, and we can insert the expressions of \eqref{eq:gauge_expr}. 
As in \cref{prop:fermion_act_ED}, we now use the antisymmetry of the form $( J_M\til\chi,\gamma^\mu \til\psi)$. For the coupling of the fermions to the gauge fields, a direct calculation then yields
\begin{align*}
&\frac12 ( J\til\xi,(\gamma^\mu\otimes B_\mu)\til\xi) = \\
&\qquad s_wg_2 A_\mu \Big( - ( J_M\bar e^\lambda, \gamma^\mu e^\lambda) + {\textstyle \frac23} ( J_M\bar u^{\lambda c}, \gamma^\mu u^{\lambda c}) - {\textstyle \frac13} ( J_M\bar d^{\lambda c}, \gamma^\mu d^{\lambda c})\Big) \notag\displaybreak[0]\\
&\quad + \frac{g_2}{4c_w} Z_\mu \begin{aligned}[t] &\Big(( J_M\bar\nu^\lambda,\gamma^\mu(1+\gamma_5)\nu^\lambda) + ( J_M\bar e^\lambda, \gamma^\mu (4{s_w}^2-1-\gamma_5)e^\lambda) \\
	& + ( J_M\bar u^{\lambda c}, \gamma^\mu (-{\textstyle \frac83}{s_w}^2+1+\gamma_5)u^{\lambda c}) + ( J_M\bar d^{\lambda c}, \gamma^\mu ({\textstyle \frac43}{s_w}^2-1-\gamma_5)d^{\lambda c}) \Big) \end{aligned} \notag\displaybreak[0]\\
&\quad + \frac{g_2}{2\sqrt{2}} W_\mu \Big( ( J_M\bar e^\lambda, \gamma^\mu (1+\gamma_5)\nu^\lambda) + ( J_M\bar d^{\lambda c},\gamma^\mu (1+\gamma_5)u^{\lambda c}) \Big) \notag\displaybreak[0]\\
&\quad + \frac{g_2}{2\sqrt{2}} W_\mu^* \Big( ( J_M\bar\nu^\lambda,\gamma^\mu (1+\gamma_5)e^\lambda) + ( J_M\bar u^{\lambda c},\gamma^\mu (1+\gamma_5)d^{\lambda c}) \Big) \notag\displaybreak[0]\\
&\quad + \frac{g_3}{2} G_\mu^i \lambda_i^{dc} \Big( ( J_M\bar u^{\lambda d}, \gamma^\mu u^{\lambda c}) + ( J_M\bar d^{\lambda d}, \gamma^\mu d^{\lambda c}) \Big) ,
\end{align*}
where in the weak interactions the projection operator $\frac12(1+\gamma_5)$ is used to select only the left-handed spinors. 

Next, we need to calculate $\frac12 ( J\til\xi,(\gamma_5\otimes \Phi)\til\xi)$. The Higgs field is given by $\Phi = D_F + \phi + J_F\phi J_F^*$, where $\phi$ is given by \eqref{eq:higgs_field_SM}. Let us first focus on the four terms involving only the Yukawa couplings for the neutrinos. Using the symmetry of the form $( J_M\til\chi,\gamma_5\til\psi)$, we obtain
\begin{multline*}
\frac12 ( J_M \bar\nu^\kappa_R,\gamma_5 Y_\nu^{\kappa\lambda}(\phi_1+1)\nu^\lambda_R) + \frac12 ( J_M \nu^\kappa_R,\gamma_5 Y_\nu^{\lambda\kappa}(\phi_1+1)\bar\nu^\lambda_R) \\
+ \frac12 ( J_M \bar\nu^\kappa_L,\gamma_5 \bar Y_\nu^{\lambda\kappa}(\bar\phi_1+1)\nu^\lambda_L) + \frac12 ( J_M \nu^\kappa_L,\gamma_5 \bar Y_\nu^{\kappa\lambda}(\bar\phi_1+1)\bar\nu^\lambda_L) \\
= ( J_M \bar\nu^\kappa_R,\gamma_5 Y_\nu^{\kappa\lambda}(\phi_1+1)\nu^\lambda_R) + ( J_M \bar\nu^\kappa_L,\gamma_5 \bar Y_\nu^{\lambda\kappa}(\bar\phi_1+1)\nu^\lambda_L) .
\end{multline*}
Using \eqref{eq:Higgs_rewrite,eq:mass_rewrite}, and dropping the generation labels, we can now rewrite
\begin{multline*}
( J_M \bar\nu_R,\gamma_5 Y_\nu(\phi_1+1)\nu_R) + ( J_M \bar\nu_L,\gamma_5 \bar Y_\nu(\bar\phi_1+1)\nu_L) \\
= i(1+\frac{h}{v}) ( J_M \bar\nu, m_\nu\nu) - \frac{\phi^0}{v} ( J_M \bar\nu,\gamma_5 m_\nu\nu) .
\end{multline*}
For $e,u,d$ we obtain similar terms, with the only difference that for $e$ and $d$ the sign for $\phi^0$ is changed. We also find terms that mix the neutrinos and electrons, and by the symmetry of the form $( J_M\til\chi,\gamma_5\til\psi)$, these are given by the four terms
\begin{align*}
\frac{\sqrt{2}}{v} \Big( \phi^- ( J_M \bar e_L, m_e \nu_L) + \phi^+ ( J_M \bar\nu_L, m_\nu e_L) - \phi^- ( J_M \bar e_R, m_\nu\nu_R) - \phi^+ ( J_M \bar \nu_R, m_e e_R) \Big) .
\end{align*}
There are four similar terms with $\nu$ and $e$ replaced by $u$ and $d$, respectively. We can use the projection operators $\frac12(1\pm\gamma_5)$ to select left- or right-handed spinors. Lastly, the off-diagonal part $T$ in the finite Dirac operator $D_F$ yields the Majorana mass terms for the right-handed neutrinos (and left-handed anti-neutrinos). Using \eqref{eq:Maj_mass_rewrite}, these Majorana mass terms are given by
\begin{align*}
( J_M \nu_R,\gamma_5 Y_R\nu_R) + ( J_M \bar\nu_L,\gamma_5 \bar Y_R\bar\nu_L) = i ( J_M \nu_R,m_R\nu_R) + i ( J_M \bar\nu_L,m_R\bar\nu_L) .
\end{align*}
We thus obtain that the mass terms of the fermions and their couplings to the Higgs field are given by
\begin{align*}
&\frac12 ( J\til\xi,(\gamma_5\otimes \Phi)\til\xi) = \\
&\quad i(1+\frac{h}{v}) \Big( ( J_M \bar\nu, m_\nu\nu) + ( J_M \bar e, m_e e) + ( J_M \bar u, m_u u) + ( J_M \bar d, m_d d) \Big) \displaybreak[0]\\
&\quad + \frac{\phi^0}{v} \Big( ( J_M \bar\nu,\gamma_5 m_\nu\nu) - ( J_M \bar e,\gamma_5 m_e e) + ( J_M \bar u,\gamma_5 m_u u) - ( J_M \bar d,\gamma_5 m_d d) \Big) \displaybreak[0]\\
&\quad + \frac{\phi^-}{\sqrt{2}v} \Big( ( J_M \bar e, m_e (1+\gamma_5) \nu) - ( J_M \bar e, m_\nu (1-\gamma_5) \nu) \Big) \displaybreak[0]\\
&\quad + \frac{\phi^+}{\sqrt{2}v} \Big( ( J_M \bar\nu, m_\nu (1+\gamma_5) e) - ( J_M \bar \nu, m_e (1-\gamma_5) e) \Big) \displaybreak[0]\\
&\quad + \frac{\phi^-}{\sqrt{2}v} \Big( ( J_M \bar d, m_d (1+\gamma_5) u) - ( J_M \bar d, m_u (1-\gamma_5) u) \Big) \displaybreak[0]\\
&\quad + \frac{\phi^+}{\sqrt{2}v} \Big( ( J_M \bar u, m_u (1+\gamma_5) d) - ( J_M \bar u, m_d (1-\gamma_5) d) \Big) \displaybreak[0]\\
&\quad + i ( J_M \nu_R,m_R\nu_R) + i ( J_M \bar\nu_L,m_R\bar\nu_L) ,
\end{align*}
where we have suppressed all indices. 
\end{proof}

In \cref{thm:spec_act_SM,thm:ferm_act_SM} we have calculated the action functional of \cref{defn:act_funct} for the almost-commutative manifold $M\times F\Sub{SM}$ defined in this section. However, we should still check whether this action coincides with the action of the Standard Model. This comparison has been worked out in detail in \cite{CCM07} (see also \cite[Ch.~1, \S17]{CM07}), and it confirms that our almost-commutative manifold indeed yields the full Lagrangian of the Standard Model (with neutrino mixing and see-saw mechanism).

\newpage
\section{Conformal invariance}
\label{chap:conf}

In this section, we shall first briefly introduce conformal transformations, and subsequently discuss the conformal invariance of Weyl gravity. We then introduce the idea of conformal symmetry breaking in the context of a single real-valued scalar field conformally coupled to gravity. Next, we shall explicitly calculate the conformal transformation of the asymptotic expansion of the spectral action for a general almost-commutative manifold, and show that it is invariant up to a kinetic term of a dilaton field, similar to what was found in \cite{CC06-scale}. Subsequently, we shall use this conformal invariance to revisit the Higgs mechanism discussed in \cref{sec:Higgs_GWS}.

\subsection{Conformal invariance}

\subsubsection{Conformal transformations}

A \emph{conformal transformation} of the metric is given by $g_{\mu\nu} \rightarrow \til g_{\mu\nu} = \Omega^2 g_{\mu\nu}$, where $\Omega\in C^\infty(M,\R^+)$ is a smooth, strictly positive function. Note that this transformation does not change the coordinates $x^\mu$ of $M$, but only the metric. 

The Riemannian curvature and its derived tensors are completely determined by the metric $g_{\mu\nu}$. Using the conformal transformation of the metric, an explicit calculation will show that the transformed scalar curvature $\til s$, corresponding to the transformed metric $\til g_{\mu\nu}$, is given by \cite[Appendix D]{Wald84}
\begin{align}
\label{eq:conf_transf_scal_curv}
\til s &= \Omega^{-2} \Big( s - 2(m-1) \nabla^\beta \nabla_\beta(\ln\Omega) - (m-1)(m-2) \nabla^\beta(\ln\Omega)\nabla_\beta(\ln\Omega) \Big) ,
\end{align}
where $\nabla$ is the Levi-Civita connection for the metric $g_{\mu\nu}$. Particularly interesting is the Weyl tensor $C^\mu_{\phantom{\mu}\nu\rho\sigma}$, which is seen to be conformally invariant (for more details, see \cite[Appendix D]{Wald84}).

\subsubsection{Conformal gravity}

We define the \emph{Weyl action} by
\begin{equation*}
S_\text{W}[g] := \int_M C_{\mu\nu\rho\sigma} C^{\mu\nu\rho\sigma} \sqrt{|g|} d^4x .
\end{equation*}
We will show in the next proposition that this Weyl action is conformally invariant, and for this reason it is also called the action of \emph{conformal gravity}. 

\begin{prop}
\label{prop:weyl_act_conf_inv}
In the case $\dim(M) = m = 4$, the Weyl action is conformally invariant.
\end{prop}
\begin{proof}
As mentioned before, we know that the Weyl tensor is conformally invariant. However, since the metric does transform under conformal transformations, this invariance depends on the position of the indices. We can calculate that
\begin{align*}
\til C_{\mu\nu\rho\sigma} = \til g_{\mu\alpha} \til C^\alpha_{\phantom{\alpha}\nu\rho\sigma} = \Omega^2 g_{\mu\alpha} C^\alpha_{\phantom{\alpha}\nu\rho\sigma} = \Omega^2 C_{\mu\nu\rho\sigma} .
\end{align*}
Similarly, since $\til g^{\mu\nu} = \Omega^{-2} g^{\mu\nu}$, we have
\begin{align*}
\til C^{\mu\nu\rho\sigma} = \Omega^{-6} C^{\mu\nu\rho\sigma} .
\end{align*}
The determinant $|g|$ of the metric can be written as $|g| = \epsilon^{\mu\nu\rho\sigma} g_{1\mu}g_{2\nu}g_{3\rho}g_{4\sigma}$, so we see that $|\til g| = \Omega^8 |g|$, and hence that $\sqrt{|\til g|} = \Omega^4 \sqrt{|g|}$. Combining this we find that
\begin{align*}
\til C_{\mu\nu\rho\sigma} \til C^{\mu\nu\rho\sigma} \sqrt{|\til g|} &= \Omega^2 C_{\mu\nu\rho\sigma} \Omega^{-6} C^{\mu\nu\rho\sigma} \Omega^4 \sqrt{|g|} = C_{\mu\nu\rho\sigma} C^{\mu\nu\rho\sigma} \sqrt{|g|} . \qedhere
\end{align*}
\end{proof}

\subsection{Conformal symmetry breaking}
\label{sec:CSB}

Consider a real-valued scalar field $\phi$ transforming as $\phi\rightarrow \Omega^{-1} \phi$. The most general invariant action for $\phi$ is:
$$
S(g_{\mu\nu},\phi) = \int_M \left( \frac12 (\partial_\mu \phi)(\partial^\mu \phi) + \frac1{12} s \phi^2 + \lambda \phi^4 \right) \sqrt{|g|} d^4x . 
$$
Note that there is no (tachyonic) Higgs mass term $-\mu^2 \phi^2$. 

Suppose we have a constant $s < 0$ and $\lambda>0$. The potential $\frac1{12} s \phi^2 + \lambda \phi^4$ then obtains a minimum for $\phi=v$ satisfying
$$
\frac1{12} s v + 2 \lambda v^3 = 0 \quad\Rightarrow\quad v^2 = \frac{-s}{24\lambda} .
$$
If on the other hand the scalar curvature is not a constant, then the vacuum expectation value $v$ would also no longer be a constant. Therefore we can no longer ignore the kinetic term in the Higgs potential. We thus have to consider the full Higgs potential
$$
\L = \frac12 (\partial_\mu \phi)(\partial^\mu \phi) + \frac1{12} s \phi^2 + \lambda \phi^4 .
$$
The extremal points of this Lagrangian are obtained for a vacuum expectation value $v$ given by the solution to the equation
$$
-\frac12 \partial_\mu\partial^\mu v + \frac1{12} s v + 2 \lambda v^3 = 0 .
$$
Unfortunately, this differential equation cannot be solved exactly, so we are unable to find an exact solution for the vacuum expectation value. However, this problem can be avoided by invoking the conformal invariance of the Higgs potential. A good treatment of such a spontaneous breaking in the case of conformal gravity with a conformally coupled scalar field can be found in \cite{Mannheim90}. For this purpose, we now perform a conformal transformation for a conveniently chosen $\Omega(x) = \sqrt{s(x) / s_0}$ for some constant $s_0$. The Higgs potential is then transformed into
$$
\L \rightarrow \til\L = \frac12 (\til\partial_\mu \til\phi)(\til\partial^\mu \til\phi) + \frac1{12} s_0 \til\phi^2 + \lambda \til\phi^4 .
$$
Since we now have a constant scalar curvature $s_0$, we can ignore the kinetic term and easily solve for the minimum of this transformed potential. We then find a vacuum expectation value $v_0$ given by
\begin{align*}
{v_0}^2 &= \frac{-s_0}{24\lambda} .
\end{align*}
We shall now introduce a new variable $h$ by writing
$$
\til\phi(x) \rightarrow \Omega(x)^{-1} \big(v_0 + h(x)\big) .
$$
This transformation spontaneously breaks the conformal invariance of the Lagrangian, which results into the broken Lagrangian
\begin{align*}
\L(g_{\mu\nu},h) &= \frac12 (\partial_\mu h)(\partial^\mu h) + \frac1{12} (v_0+h)^2 s_0 + \lambda (v_0+h)^4 \\
&= \frac12 (\partial_\mu h)(\partial^\mu h) + \lambda \Big( h^4 + 4v_0h^3 + 4{v_0}^2h^2 - {v_0}^4 \Big) .
\end{align*}

\begin{remark}
Under the assumption that $\phi(x)\neq0$, there is an alternative approach to the spontaneous symmetry breaking of the conformal invariance (see e.g.\ \cite{Demir04}). Namely, we can then choose the conformal transformation $\Omega(x) = \phi(x) / \phi_0$ for some constant $\phi_0$. The Lagrangian in this case becomes
$$
\L(g_{\mu\nu}) = \frac{{\phi_0}^2}{12} \til s + \lambda {\phi_0}^4 .
$$
Hence with this alternative transformation we recover the usual Einstein-Hilbert Lagrangian for gravity including a cosmological constant, and the scalar field $\phi$ has completely disappeared. 
\end{remark}

\subsection{Conformal transformations of the spectral action}
\label{sec:conf_spec_act}

The scale invariance of the spectral action has been discussed by Chamseddine and Connes in \cite{CC06-scale}. In their approach, the constant cut-off scale $\Lambda$ in the definition of the spectral action (see \eqref{eq:spectral_action}) is replaced by a dynamical scale $\Lambda e^{\eta}$, thus introducing a dilaton field $\eta$. In this section we will not discuss the general approach of \cite{CC06-scale}, but we only focus on the asymptotic expansion of the spectral action for an almost-commutative manifold. We will explicitly show that this asymptotic expansion is invariant under conformal transformations. 

Under a conformal transformation given by $\Omega\in C^\infty(M,\R^+)$, we will let the fields $B_\mu$ and $\Phi$ of \eqref{eq:fluc_Gauge,eq:fluc_Higgs} transform as
\begin{align*}
\til B_\mu &= B_\mu , & \til \Phi &= \Omega^{-1} \Phi .
\end{align*}
The spectral action depends on the choice of the cut-off scale $\Lambda$, and it is no surprise that a conformal transformation should also affect this cut-off scale. We thus transform the constant $\Lambda$ into a dynamical scale given by
\begin{align*}
\til\Lambda = \Omega^{-1} \Lambda .
\end{align*}

\begin{prop}
\label{prop:conf_ACG}
A conformal transformation of the spectral action of an almost-commutative manifold (cf.\ \cref{prop:acm_spec_act}) yields (ignoring boundary terms) the Lagrangian
\begin{align*}
\L(\til g_{\mu\nu},\til B_\mu,\til \Phi,\til \Lambda) = \L(g_{\mu\nu},B_\mu,\Phi,\Lambda) + \frac{N f_2 \Omega^{-2}\Lambda^2}{4\pi^2} \nabla^\beta(\Omega)\nabla_\beta(\Omega) \sqrt{|g|} . 
\end{align*} 
\end{prop}
\begin{proof}
We shall calculate the conformal transformations of the different terms in the Lagrangian separately. We shall ignore topological and boundary terms and use \eqref{eq:lagr_M_simple} for $\L_M$. We then find that
\begin{multline*}
\L_M(\til g_{\mu\nu},\til \Lambda) \sqrt{|\til g|} = \frac{f_4\Omega^{-4}\Lambda^4}{2\pi^2} \Omega^4 \sqrt{|g|} \\
- \frac{f_2\Omega^{-2}\Lambda^2}{24\pi^2} \til s \Omega^4 \sqrt{|g|} - \frac{f(0)}{320\pi^2} C_{\mu\nu\rho\sigma} C^{\mu\nu\rho\sigma} \sqrt{|g|} ,
\end{multline*}
where we have used the conformal invariance of the Weyl action (cf.\ \cref{prop:weyl_act_conf_inv}). Inserting the formula for $\til s$ given by \eqref{eq:conf_transf_scal_curv}, we obtain several extra terms, and we have
\begin{multline*}
\L_M(\til g_{\mu\nu},\til\Lambda) \sqrt{|\til g|} \\
= \L_M(g_{\mu\nu},\Lambda) \sqrt{|g|} + \frac{f_2\Lambda^2}{4\pi^2} \Big( \nabla^\beta \nabla_\beta(\ln\Omega) + \nabla^\beta(\ln\Omega)\nabla_\beta(\ln\Omega) \Big) \sqrt{|g|} .
\end{multline*}
On the second line, the first term is a total divergence and yields a boundary term, which we will ignore, and the second term can be rewritten as
\begin{align*}
\nabla^\beta(\ln\Omega)\nabla_\beta(\ln\Omega) = \Omega^{-2} \nabla^\beta(\Omega)\nabla_\beta(\Omega) .
\end{align*}
Hence we obtain (ignoring the boundary term)
\begin{align*}
\L_M(\til g_{\mu\nu},\til\Lambda) \sqrt{|\til g|} = \L_M(g_{\mu\nu},\Lambda) \sqrt{|g|} + \frac{f_2 \Lambda^2}{4\pi^2} \Omega^{-2} \nabla^\beta(\Omega)\nabla_\beta(\Omega) \sqrt{|g|} . 
\end{align*}

Since the gauge field $B_\mu$ does not transform, neither does $F_{\mu\nu}$. However, we do have 
\begin{align*}
\til F^{\mu\nu} = \til g^{\mu\alpha} \til g^{\nu\beta} \til F_{\alpha\beta} = \Omega^{-4} g^{\mu\alpha} g^{\nu\beta} F_{\alpha\beta} = \Omega^{-4} F^{\mu\nu}.
\end{align*}
From this we find that the kinetic term of the gauge field remains invariant under conformal transformations:
\begin{align*}
\L_B(\til B_\mu) \sqrt{|\til g|} &= \frac{f(0)}{24\pi^2} \Tr(F_{\mu\nu}\Omega^{-4}F^{\mu\nu}) \Omega^4 \sqrt{|g|} = \L_B(B_\mu) \sqrt{|g|} .
\end{align*}

We shall split $\L_H$ into two parts, and we shall write $\L_1$ for the Higgs potential (ignoring the boundary term) and $\L_2$ for the kinetic term and the minimal coupling to the other fields. The Higgs potential $\L_1$ transforms as 
\begin{align*}
\L_1(\til\Phi,\til\Lambda) \sqrt{|\til g|} &= -\frac{f_2\Omega^{-2}\Lambda^2}{2\pi^2} \Tr(\Omega^{-2}\Phi^2) \Omega^4\sqrt{|g|} + \frac{f(0)}{8\pi^2} \Tr(\Omega^{-4}\Phi^4) \Omega^4\sqrt{|g|} \\
&= -\frac{f_2\Lambda^2}{2\pi^2} \Tr(\Phi^2)\sqrt{|g|} + \frac{f(0)}{8\pi^2} \Tr(\Phi^4)\sqrt{|g|} = \L_1(\Phi,\Lambda) \sqrt{|g|} .
\end{align*}
For the last part of the Lagrangian we have
\begin{multline*}
\L_2(\til g_{\mu\nu},\til B_\mu,\til \Phi) \sqrt{|\til g|} = \frac{f(0)}{48\pi^2}  \til s\Tr(\Omega^{-2}\Phi^2) \Omega^4\sqrt{|g|} \\
+ \frac{f(0)}{8\pi^2} \Tr\big((\til D_\mu \Omega^{-1}\Phi)(\til D^\mu \Omega^{-1}\Phi)\big) \Omega^4\sqrt{|g|} .
\end{multline*}
The first term is given by
\begin{multline*}
\frac{f(0)}{48\pi^2} \til s\Tr(\Omega^{-2}\Phi^2) \Omega^4\sqrt{|g|} = \frac{f(0)}{48\pi^2} s\Tr(\Phi^2) \sqrt{|g|} \\
- \frac{f(0)}{8\pi^2} \Big( \nabla^\beta \nabla_\beta(\ln\Omega) + \nabla^\beta(\ln\Omega)\nabla_\beta(\ln\Omega) \Big)\Tr(\Phi^2) \sqrt{|g|} .
\end{multline*}
We shall rewrite
\begin{align*}
\nabla^\beta \nabla_\beta(\ln\Omega) = \nabla^\beta \Big( \Omega^{-1}\nabla_\beta(\Omega) \Big) 
= - \Omega^{-2} \nabla^\beta(\Omega)\nabla_\beta(\Omega) + \Omega^{-1}\nabla^\beta\nabla_\beta(\Omega) 
\end{align*}
and
\begin{align*}
\nabla^\beta(\ln\Omega)\nabla_\beta(\ln\Omega) = \Omega^{-2} \nabla^\beta(\Omega)\nabla_\beta(\Omega) ,
\end{align*}
and obtain
\begin{align*}
\frac{f(0)}{48\pi^2} \til s\Tr(\Omega^{-2}\Phi^2) \Omega^4\sqrt{|g|} = \frac{f(0)}{48\pi^2} s\Tr(\Phi^2) \sqrt{|g|} - \frac{f(0)}{8\pi^2} \Omega^{-1}\nabla^\beta\nabla_\beta(\Omega) \Tr(\Phi^2) \sqrt{|g|} .
\end{align*}

$D_\mu$ has been defined in \cref{prop:acm_Dirac_sq} by $D_\mu\Phi = [\nabla^E_\mu,\Phi]$. The transformation of $\nabla^E$ is determined by the transformation of $\nabla$, and it only yields new terms which commute with $\Phi$. Therefore we can conclude that $\til D_\mu = D_\mu$, and $\til D^\mu = \Omega^{-2}D^\mu$. From this we find that
\begin{align*}
\til D_\mu \Omega^{-1}\Phi = \Omega^{-1}(D_\mu \Phi) + (\partial_\mu \Omega^{-1})\Phi .
\end{align*}
We then find that the second term of $\L_2$ decomposes as
\begin{multline*}
\frac{f(0)}{8\pi^2} \Tr\big((D_\mu \Omega^{-1}\Phi)(D^\mu \Omega^{-1}\Phi)\big) \Omega^2\sqrt{|g|} = \frac{f(0)}{8\pi^2} \Tr\big((D_\mu\Phi)(D^\mu\Phi)\big) \sqrt{|g|} \\
+ \frac{f(0)}{8\pi^2} \Omega(\partial_\mu \Omega^{-1}) \Tr\big(D^\mu\Phi^2\big) \sqrt{|g|} + \frac{f(0)}{8\pi^2} \Omega^2(\partial_\mu \Omega^{-1})(\partial^\mu \Omega^{-1}) \Tr\big(\Phi^2\big) \sqrt{|g|} .
\end{multline*}
Note that $\Tr\big([B^\mu,\Phi^2]\big) = 0$ by the cyclic property of the trace, so that $\Tr\big(D^\mu\Phi^2\big) = \Tr\big(\partial^\mu\Phi^2\big)$. Combining both terms of $\L_2$, we see that
\begin{multline*}
\L_2(\til g_{\mu\nu},\til B_\mu,\til \Phi) \sqrt{|\til g|} = \L_2(g_{\mu\nu},B_\mu,\Phi) \sqrt{|g|} - \frac{f(0)}{8\pi^2} \Omega^{-1}\nabla^\beta\nabla_\beta(\Omega) \Tr(\Phi^2) \sqrt{|g|} \\
\quad- \frac{f(0)}{8\pi^2} \Omega^{-1}(\partial_\mu \Omega) \Tr\big(\partial^\mu\Phi^2\big) \sqrt{|g|} + \frac{f(0)}{8\pi^2} \Omega^{-2}(\partial_\mu \Omega)(\partial^\mu \Omega) \Tr\big(\Phi^2\big) \sqrt{|g|} .
\end{multline*}
By using the fact that $\nabla_\beta f = \partial_\beta f$ for functions $f\in C^\infty(M)$, we note that
\begin{multline*}
\nabla^\beta \Big( \Omega^{-1}\nabla_\beta(\Omega) \Tr(\Phi^2) \Big) = - \Omega^{-2} \partial^\beta(\Omega) \partial_\beta(\Omega) \Tr(\Phi^2) \\
+ \Omega^{-1}\nabla^\beta\nabla_\beta(\Omega) \Tr(\Phi^2) +  \Omega^{-1}\partial_\beta(\Omega) \Tr(\partial^\beta\Phi^2) ,
\end{multline*}
and thus we can conclude that
\begin{align*}
\L_2(\til g_{\mu\nu},\til B_\mu,\til \Phi) \sqrt{|\til g|} = \L_2(g_{\mu\nu},B_\mu,\Phi) \sqrt{|g|} - \frac{f(0)}{8\pi^2} \nabla^\beta \Big( \Omega^{-1}\nabla_\beta(\Omega) \Tr(\Phi^2) \Big) \sqrt{|g|} .
\end{align*}
Ignoring this boundary term, we see that $\L_H$ is invariant under conformal transformations. 
\end{proof}

\begin{remark}
\label{remark:conf_ACG}
Let us write the conformal scaling factor as $\Omega = e^\eta$. We can then write the transformation of the Lagrangian as
\begin{align*}
\L(\til g_{\mu\nu},\til B_\mu,\til \Phi,\til\Lambda) \sqrt{|\til g|} = \L(g_{\mu\nu},B_\mu,\Phi,\Lambda) \sqrt{|g|} + \frac{N f_2}{4\pi^2} \Lambda^2 (\partial^\beta\eta)(\partial_\beta\eta) \sqrt{|g|} .
\end{align*}
So the only effect of the conformal transformation is that  we obtain in the Lagrangian the kinetic term $(\partial^\beta\eta)(\partial_\beta\eta)$ of a \emph{dilaton field} $\eta$. 
\end{remark}

\subsection{The Higgs mechanism revisited}

In \cref{sec:Higgs_GWS} we have discussed the Higgs mechanism for the Glashow-Weinberg-Salam model under the assumption that the scalar curvature $s$ vanishes identically. If the scalar curvature does not vanish, we gain an additional term in the Higgs potential given by the conformal coupling $s|H|^2$. Since the scalar curvature $s$ is a function on $M$, the vacuum expectation value of this full potential will in general not be a spacetime constant, and therefore we can no longer ignore the kinetic term in the Higgs potential. The vacuum expectation value $v$ will now be given by the solution to the equation
\begin{align*}
-\frac12 D_\mu D^\mu v(x) + \left(\frac{-2af_2\Lambda^2 + ef(0)}{af(0)} + \frac1{12} s(x) \right) v(x) + \frac{b\pi^2}{a^2f(0)} v(x)^3 = 0 .
\end{align*}
Unfortunately, this differential equation cannot be solved exactly, and this poses a problem in applying the Higgs mechanism. However, as in \cref{sec:CSB}, we can avoid this problem by invoking the conformal invariance of the spectral action. Thus, let us perform a conformal transformation for a conveniently chosen function $\Omega(x)$. This transformation needs to be such that the coefficients in the above equation become constants. Before we continue, consider the effect of a conformal transformation given by $\Omega(x)$ on the field $\Phi$, which transforms as 
\begin{align*}
\Phi = \mattwo{S+\phi}{T^*}{T}{\bar{(S+\phi)}} \longrightarrow \Omega^{-1}(x) \Phi = \mattwo{\Omega^{-1}(x)\big(S+\phi\big)}{\Omega^{-1}(x)T^*}{\Omega^{-1}(x)T}{\Omega^{-1}(x)\bar{(S+\phi)}} .
\end{align*}
The rescaling of $S+\phi$ is given by the rescaling of the doublet $H\rightarrow\Omega(x)^{-1}H$. However, the conformal transformation also affects the off-diagonal part $T$ (which gives the Majorana mass for the right-handed neutrinos) and hence it affects the constants $c$, $d$ and $e$ of \eqref{eq:abcde_GWS}. So, when performing a conformal transformation, these constants must be transformed accordingly. 

We now choose the conformal transformation given by
$$
\Omega(x) := \sqrt{\frac{2af_2\Lambda^2 - ef(0)}{af(0)} - \frac1{12} s(x)} \Omega_0 ,
$$
where $\Omega_0$ is some arbitrary constant. For the transformed $\til v(x) = \Omega^{-1}(x) v(x)$ we then obtain the equation
\begin{align*}
-\Omega_0^{-2} \til v(x) + \frac{b\pi^2}{a^2f(0)} \til v(x)^3 = 0 ,
\end{align*}
and the solution to this equation yields the constant vacuum expectation value
$$
v_0 = \frac{a\sqrt{f(0)}}{\sqrt{b}\pi\Omega_0} .
$$
Note that, because of the freedom we have in choosing $\Omega_0$, we are free to take $v_0=1$ through a \emph{global} conformal transformation. However, for clarity we will simply leave $v_0$ as it is, without specifying its value. 

\begin{remark}
The minimum of the Higgs potential is obtained for a \emph{non-vanishing} Higgs field if and only if $2af_2\Lambda^2 - ef(0) - \frac1{12} af(0) s(x) > 0$. When $\frac1{12} af(0) s(x) > 2af_2\Lambda^2 - ef(0)$, the total coefficient in front of $|H|^2$ becomes positive, and hence the minimum of the Higgs potential is obtained for $H=0$. So, when the scalar curvature $s$ becomes large enough, there will be no spontaneous symmetry breaking. A varying scalar curvature can thus cause a transition from a broken to an unbroken theory. This is interesting in the context of cosmological applications of the spectral action, as studied for instance in \cite{MPT10,BFS10,NOS10,Sak10}.
\end{remark}

\begin{thm}
\label{thm:spec_act_broken_GWS}
The (gauge and conformal) transformation of the Higgs field, given by
\begin{align}
\label{eq:Higgs_transf_vac2}
H &= \Omega^{-1}(x) u(x) \vectwo{v_0+h(x)}{0} , & h(x) &:= \Omega(x) |H(x)| - v_0 ,
\end{align}
breaks both the gauge symmetry and the conformal symmetry. The resulting spontaneously broken bosonic action (ignoring topological and boundary terms) of the Glashow-Weinberg-Salam model is given by
\begin{align*}
S_B &= \int_M \Bigg( \frac{4f_4\Lambda^4}{\pi^2} - \frac{cf_2\Lambda^2}{\pi^2} + \frac{df(0)}{4\pi^2} - \frac{b\pi^2}{2a^2f(0)} {v_0}^4 \notag\\
&\quad+ \left( \frac{cf(0)}{24\pi^2} - \frac{f_2\Lambda^2}{3\pi^2} \right) s - \frac{f(0)}{40\pi^2} C_{\mu\nu\rho\sigma} C^{\mu\nu\rho\sigma} \notag\\
&\quad+ \frac14 B_{\mu\nu} B^{\mu\nu} + \frac14 W_{\mu\nu}^a W^{\mu\nu,a} + \frac{f_2 \Lambda^2}{3\pi^2} (\partial^\beta\eta)(\partial_\beta\eta) \notag\\ 
&\quad+ \frac12 (\partial^\mu h)(\partial_\mu h) + \frac{b\pi^2}{2a^2f(0)} \Big(h^4 + 4v_0h^3 + 4{v_0}^2h^2\Big) \notag\\
&\quad+ \frac14 {g_2}^2 (v_0+h)^2 W^{\mu*} W_\mu + \frac18 \frac{{g_2}^2}{{c_w}^2} (v_0+h)^2 Z^\mu Z_\mu \Bigg) \sqrt{|g|} d^4x .
\end{align*}
Here we have introduced the dilaton field $\eta$ by setting $\Omega(x) = e^{\eta(x)}$, as in \cref{remark:conf_ACG}.
\end{thm}
\begin{proof}
As has been shown in \cref{prop:conf_ACG}, the full Lagrangian $\L_H$ (when integrated over $M$) is conformally invariant. Thus, we can simply replace $H$ by the doublet $(v_0+h(x),0)$ in the formula for $\L_H$ given in \cref{prop:spec_act_GWS}. Using the rescaling of the Higgs field of \eqref{eq:Higgs_rescale}, the Higgs potential and the Higgs kinetic term are then rewritten as in \eqref{eq:Higgs_pot,eq:Higgs_kin}. For the gauge kinetic terms, we use \eqref{eq:gauge_kin} and impose the relations \eqref{eq:ew_uni_rel}. The gravitational Lagrangian $\L_M$ of \cref{prop:spec_act_GWS} obtains a kinetic term for a dilaton field $\eta$ by \cref{remark:conf_ACG}. Combining all terms then proves the proposition. 
\end{proof}

\begin{remark}
In the Higgs Lagrangian $\L_H$, as given in \cref{prop:spec_act_GWS}, there is one term proportional to $s|H|^2$. One might expect that after the spontaneous symmetry breaking this would yield a contribution to the Einstein-Hilbert Lagrangian in the form of $s{v_0}^2$. However, this is not the case, since this term combines with the other terms proportional to ${v_0}^2$, and then their coefficient is also seen to be proportional to ${v_0}^2$. In this way, we are only left with the constant term $- \frac{b\pi^2}{2a^2f(0)} {v_0}^4$, which contributes to the cosmological constant. 

Furthermore, after the conformal transformation there also remains no coupling between the Higgs field $h$ and the scalar curvature $s$, since this coupling has been absorbed into the mass term ${v_0}^2h^2$. Hence, the term $s|H|^2$ has completely disappeared from the action. 
\end{remark}

\subsubsection{The Higgs mechanism for the Standard Model}

The above approach for the conformal symmetry breaking of the Glashow-Weinberg-Salam model generalizes straightforwardly to the description of the full Standard Model of \cref{chap:ex_SM}. In exactly the same way as in \cref{thm:spec_act_broken_GWS}, we can now rewrite the spectral action of \cref{prop:spec_act_SM}, and obtain the following result.

\begin{thm}
\label{thm:spec_act_broken_SM}
The spontaneously broken bosonic action (ignoring topological and boundary terms) of the Standard Model is given by 
\begin{align*}
S_B &= \int_M \Bigg( \frac{48f_4\Lambda^4}{\pi^2} - \frac{cf_2\Lambda^2}{\pi^2} + \frac{df(0)}{4\pi^2} - \frac{b\pi^2}{2a^2f(0)} {v_0}^4 \notag\\
&\quad+ \frac{cf(0) s}{24\pi^2} - \frac{4f_2\Lambda^2 s}{\pi^2} - \frac{3f(0)}{10\pi^2} C_{\mu\nu\rho\sigma} C^{\mu\nu\rho\sigma} \notag\\
&\quad+ \frac14 B_{\mu\nu} B^{\mu\nu} + \frac14 W_{\mu\nu}^a W^{\mu\nu,a} + \frac14 G_{\mu\nu}^i G^{\mu\nu,i} + \frac{4 f_2 \Lambda^2}{\pi^2} (\partial^\beta\eta)(\partial_\beta\eta) \notag\\
&\quad+ \frac12 (\partial^\mu h)(\partial_\mu h) + \frac{b\pi^2}{2a^2f(0)} \Big(h^4 + 4v_0h^3 + 4{v_0}^2h^2\Big) \notag\\
&\quad+ \frac14 {g_2}^2 (v_0+h)^2 W^{\mu*} W_\mu + \frac18 \frac{{g_2}^2}{{c_w}^2} (v_0+h)^2 Z^\mu Z_\mu \Bigg) \sqrt{|g|} d^4x .
\end{align*}
\end{thm}

\newpage
\section{Phenomenology} 
\label{chap:RGE}

In \cref{thm:spec_act_broken_SM,thm:ferm_act_SM}, we have derived the full Lagrangian for the Standard Model from the almost-commutative manifold $M\times F\Sub{SM}$. The coefficients in this Lagrangian are given in terms of:
\begin{itemize}
\item the moments $f(0)$, $f_2$ and $f_4$ of the function $f$ in the spectral action;
\item the cut-off scale $\Lambda$ in the spectral action;
\item the vacuum expectation value $v_0$ of the Higgs field; 
\item the coefficients $a,b,c,d,e$ of \eqref{eq:abcde_SM} determined by the mass matrices in the finite Dirac operator $D_F$. 
\end{itemize}
We can find several relations among these coefficients in the Lagrangian, which we shall derive in the following section. Inspired by the relation ${g_3}^2 = {g_2}^2 = \frac53 {g_1}^2$ obtained from \eqref{eq:couplings_norm}, we will assume that these relations hold at the unification scale. Subsequently, we will use the renormalization group equations to obtain predictions for the Standard Model at ordinary energies. For the first part, we mainly follow the same outline as in \cite[\S5]{CCM07} (see also \cite[Ch.~1, \S17]{CM07}). We then incorporate also the running of the neutrino masses as in \cite{JKSS07-seesaw}.

\subsection{Mass relations}
\label{sec:mass_relations}

\subsubsection{Fermion masses}

Recall from \eqref{eq:mass_rewrite} that we have defined the mass matrices $m_x$ of the fermions by rewriting the matrices $Y_x$ in the finite Dirac operator $D_F$. Inserting the formula \eqref{eq:mass_rewrite} for $Y_x$ into the expression for $a$ given by \eqref{eq:abcde_SM}, we obtain
\begin{align*}
a = \frac{af(0)}{\pi^2{v_0}^2} \Tr\big(m_\nu^*m_\nu + m_e^*m_e + 3m_u^*m_u + 3m_d^*m_d\big) ,
\end{align*}
which yields
\begin{align*}
\Tr\big(m_\nu^*m_\nu + m_e^*m_e + 3m_u^*m_u + 3m_d^*m_d\big) = \frac{\pi^2{v_0}^2}{f(0)} .
\end{align*}
From \eqref{eq:gauge_masses} we know that the mass of the W-boson is given by $M_W = \frac12 v_0 g_2$. Using the normalization equation \eqref{eq:couplings_norm}, which expresses $g_2$ in terms of $f(0)$, we can then write
\begin{align}
\label{eq:f(0)}
f(0) = \frac{\pi^2{v_0}^2}{8{M_W}^2} .
\end{align}
Inserting this into the expression above, we obtain a relation between the fermion mass matrices $m_x$ and the W-boson mass $M_W$:
\begin{align}
\label{eq:masses_ferm_W}
\Tr\big(m_\nu^*m_\nu + m_e^*m_e + 3m_u^*m_u + 3m_d^*m_d\big) = 2{g_2}^2{v_0}^2 = 8 {M_W}^2 .
\end{align}
If we would assume that the mass of the top quark is much larger than all other fermion masses, we may neglect the other fermion masses. In that case the above relation would yield the constraint
\begin{align}
\label{eq:top_mass}
m\Sub{\text{top}} \lesssim \sqrt\frac83 M_W .
\end{align}

\subsubsection{The Higgs mass}

For the Higgs boson $h$ we obtain a mass $m_h$ by writing the term proportional to $h^2$ in \cref{thm:spec_act_broken_SM} in the form
\begin{align*}
\frac{b\pi^2}{2a^2f(0)} 4 {v_0}^2 h^2 = \frac12 {m_h}^2 h^2 .
\end{align*}
We then see that the Higgs mass is given by 
\begin{align}
\label{eq:Higss_mass}
m_h = \frac{2\pi\sqrt{b}v_0}{a\sqrt{f(0)}} .
\end{align}
By inserting \eqref{eq:f(0)} into this expression for the Higgs mass, we see that $M_W$ and $m_h$ are related by
\begin{align*}
{m_h}^2 = 32 \frac{b}{a^2} {M_W}^2 .
\end{align*}
Next, we introduce the quartic Higgs coupling constant $\lambda$ by writing 
\begin{align*}
\frac{b\pi^2}{2a^2f(0)} h^4 =: \frac{1}{24} \lambda h^4 .
\end{align*}
From \eqref{eq:couplings_norm}, we then find
\begin{align}
\label{eq:Higgs_quartic_coupling}
\lambda = 24 \frac{b}{a^2} {g_2}^2 .
\end{align}
We then find that the Higgs mass can be expressed in terms of the mass $M_W$ of the W-boson, the coupling constant $g_2$ and the quartic Higgs coupling $\lambda$ as 
\begin{align}
\label{eq:Higgs_mass}
{m_h}^2 =  \frac{4\lambda{M_W}^2}{3{g_2}^2} .
\end{align}

\subsubsection{The seesaw mechanism}

Let us consider the mass terms for the neutrinos. The matrix $D_F$ described in \cref{sec:finite_SM} provides the Dirac masses as well as the Majorana masses of the fermions. After a rescaling as in \eqref{eq:mass_rewrite}, the mass matrix restricted to the subspace of $\mH_F$ with basis $\{\nu_L, \nu_R, \bar{\nu_L}, \bar{\nu_R}\}$ is given by
\begin{align*}
\matfour{0&m_\nu^*&m_R^*&0}{m_\nu&0&0&0}{m_R&0&0&\bar m_\nu^*}{0&0&\bar m_\nu&0} .
\end{align*}
Suppose we consider only one generation, so that $m_\nu$ and $m_R$ are just scalars. The eigenvalues of the above mass matrix are then given by
\begin{align*}
\pm \frac12 m_R \pm \frac12 \sqrt{{m_R}^2 + 4{m_\nu}^2} .
\end{align*} 
If we assume that $m_\nu \ll m_R$, then these eigenvalues are approximated by $\pm m_R$ and $\pm\frac{{m_\nu}^2}{m_R}$. This means that there is a heavy neutrino, for which the Dirac mass $m_\nu$ may be neglected, and its mass is given by the Majorana mass $m_R$. However, there is also a light neutrino, for which the Dirac and Majorana terms conspire to yield a mass $\frac{{m_\nu}^2}{m_R}$, which is in fact much smaller than the Dirac mass $m_\nu$. This is called the \emph{seesaw mechanism}. Thus, even though the observed masses for these neutrinos may be very small, they might still have a large Dirac mass (or Yukawa coupling). 

From \eqref{eq:masses_ferm_W} we have obtained a relation between the masses of the top quark and the W-boson, by neglecting all other fermion masses. However, because of the seesaw mechanism, it might be that one of the neutrinos has a Dirac mass of the same order of magnitude as the top quark. It would then not be justified to neglect all other fermion masses, but instead we need to correct for such massive neutrinos. 

Let us introduce a new parameter $\rho$ (typically taken to be of order $1$) for the ratio between the Dirac mass $m_\nu$ for the tau-neutrino and the mass $m\Sub{\text{top}}$ of the top quark (at unification scale), so we shall write $m_\nu = \rho m\Sub{\text{top}}$. Instead of \eqref{eq:top_mass} we then obtain the restriction
\begin{align}
\label{eq:top_mass_corr}
m\Sub{\text{top}} \lesssim \sqrt{\frac{8}{3+\rho^2}} M_W .
\end{align}

\subsection{Renormalization group flow}
\label{sec:renorm}

In this section we shall evaluate the renormalization group equations (RGEs) for the Standard Model from ordinary energies up to the unification scale. For the validity of these RGEs we need to assume the existence of a big desert up to the unification scale. One assumes 
\begin{itemize}
\item that there exist no new particles (besides the known Standard Model particles and the Higgs boson) with a mass below the unification scale;
\item that perturbative quantum field theory remains valid throughout the big desert. 
\end{itemize}
Furthermore, we shall also ignore any gravitational contributions to the renormalization group flow. 

\subsubsection{Coupling constants}

In \cref{sec:couplings_SM} we have introduced the coupling constants for the gauge fields, and we have obtained the relation ${g_3}^2 = {g_2}^2 = \frac53 {g_1}^2$. This is precisely the relation between the coupling constants at unification, common to grand unified theories (GUT). Thus, it would be natural to assume that our model is defined at the scale $\Lambda\Sub{GUT}$. However, it turns out that there is no scale at which the relation ${g_3}^2 = {g_2}^2 = \frac53 {g_1}^2$ holds exactly, as we will see below. 

The renormalization group $\beta$ functions of the (minimal) standard model are taken from \cite{MV83,MV84,MV85} or \cite{FJSE93}. We shall simplify the expressions by ignoring the $2$-loop contributions, and instead consider only the $1$-loop approximation. 
By inserting the number of $n_g=3$ generations, the renormalization group equations (RGEs) then read (see \cite[Eq.~(B.2)]{MV83} or \cite[Eq.~(A.1)]{FJSE93})
\begin{align*}
\frac{dg_i}{dt} &= -\frac{1}{16\pi^2} b_i g_i^3 , & b &= \left(-\frac{41}{6}, \frac{19}{6}, 7 \right) ,
\end{align*}
where $t=\log\mu$. At first order, these equations are uncoupled from all other parameters of the Standard Model, and the solutions for the running coupling constants $g_i(\mu)$ at the energy scale $\mu$ are easily seen to satisfy
\begin{align}
\label{eq:running_couplings}
g_i(\mu)^{-2} = g_i(M_Z)^{-2} + \frac{b_i}{8\pi^2} \log \frac\mu{M_Z} ,
\end{align}
where $M_Z$ is the mass of the Z-boson \cite{PDG10}:
\begin{align*}
M_Z = 91.1876 \pm 0.0021 \;\GeV . 
\end{align*}
The experimental values of the coupling constants at the energy scale $M_Z$ are known, and are given by \cite{PDG10} 
\begin{align*}
g_1(M_Z) &= 0.3575 \pm 0.0001 , & g_2(M_Z) &= 0.6519 \pm 0.0002, & g_3(M_Z) &= 1.220 \pm 0.004 . 
\end{align*}
Using these experimental values, we obtain the running of the coupling constants in \cref{fig:running_couplings}. As can be seen in this figure, the running coupling constants do not meet in one point, and hence they do not determine a unique unification scale $\Lambda\Sub{GUT}$. In other words, the relation ${g_3}^2 = {g_2}^2 = \frac53 {g_1}^2$ cannot hold exactly at any energy scale, unless we drop the big desert hypothesis. Nevertheless, in the remainder of this section we shall assume that this relation holds approximately. We shall consider the range for $\Lambda\Sub{GUT}$ determined by the triangle of the running coupling constants in \cref{fig:running_couplings}. The scale $\Lambda_{12}$ at the intersection of $\sqrt{\frac53}g_1$ and $g_2$ determines the lowest value for $\Lambda\Sub{GUT}$, and is given by
\begin{align}
\label{eq:GUT-scale_12}
\Lambda_{12} = M_Z \exp\left(\frac{8\pi^2(\frac35g_1(M_Z)^{-2}-g_2(M_Z)^{-2})}{b_2-\frac35b_1}\right) = 1.03 \times 10^{13} \;\GeV . 
\end{align}
The highest value $\Lambda_{23}$ is given by the solution of $g_2 = g_3$, which yields the value 
\begin{align}
\label{eq:GUT-scale_23}
\Lambda_{23} = M_Z \exp\left(\frac{8\pi^2(g_3(M_Z)^{-2}-g_2(M_Z)^{-2})}{b_2-b_3}\right) = 9.92 \times 10^{16} \;\GeV . 
\end{align}
We will assume that the Lagrangian we have derived from the AC-manifold $M\times F\Sub{SM}$ is valid at some scale $\Lambda\Sub{GUT}$, which we take between $\Lambda_{12}$ and $\Lambda_{23}$. All relations obtained in \cref{sec:mass_relations} are assumed to hold approximately at this scale, and all predictions that will follow from these relations are therefore also only approximate. 

\begin{remark}
In our approach, we compare the Lagrangian derived from the AC-manifold at the GUT-scale with experimental values obtained in the low-energy regime. Therefore our approach has nothing to say about the physics beyond the GUT-scale (e.g.\ the occurrence of Landau poles). We believe that extending the model beyond the GUT-scale would require a deeper understanding of a theory of quantum gravity, which is beyond the scope of this review. 
\end{remark}

\begin{figure}
\includegraphics[height=0.42\textheight]{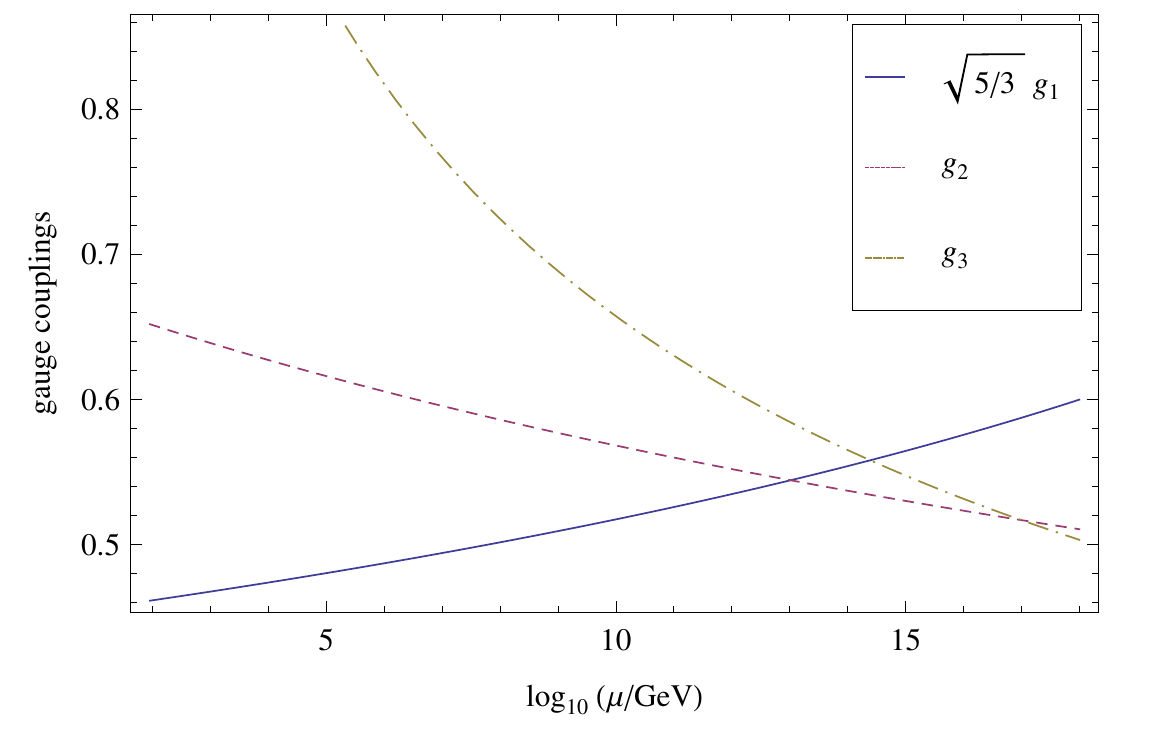}
\caption{The running of the gauge coupling constants.\label{fig:running_couplings}}
\label{figure:running-couplings}
\end{figure}

\subsubsection{Renormalization group equations}

The running of the neutrino masses has been studied in a general setting for non-degenerate seesaw scales in \cite{AKLR02}. In the following we shall consider the case where only the tau-neutrino has a large Dirac mass $m_\nu$, which cannot be neglected with respect to the mass of the top-quark. In the remainder of this section, we shall calculate the running of the Yukawa couplings for the top-quark and the tau-neutrino, as well as the running of the quartic Higgs coupling. Let us write $y\Sub{\text{top}}$ and $y_\nu$ for the Yukawa couplings of the top quark and the tau-neutrino, defined by
\begin{align}
\label{eq:Yukawa_mass}
m\Sub{\text{top}} &= \frac12 \sqrt2 y\Sub{\text{top}} v_0 , & m_\nu &= \frac12 \sqrt2 y_\nu v_0 ,
\end{align}
where $v_0$ is the vacuum expectation value of the Higgs field. 

We shall follow a similar approach as in \cite{JKSS07-seesaw}. Let $m_R$ be the Majorana mass for the right-handed tau-neutrino. By the Appelquist-Carazzone decoupling theorem \cite{AC75}, we then distinguish two energy domains: $E>m_R$ and $E<m_R$. We shall again neglect all fermion masses except for the top quark and the tau neutrino. For high energies $E>m_R$, the renormalization group equations are given by (see \cite[Eq.~(B.4)]{MV84}, \cite[Eq.~(15)]{AKLR02} and \cite[Eq.~(B.3)]{MV85})
\begin{align}
\label{eq:RGE_high}
\frac{dy\Sub{\text{top}}}{dt} &= \frac1{16\pi^2} \left(\frac92 y\Sub{\text{top}}^2 + y_\nu^2 - \frac{17}{12}g_1^2 - \frac94 g_2^2 - 8g_3^2\right) y\Sub{\text{top}} ,\notag\\
\frac{dy_\nu}{dt} &= \frac1{16\pi^2} \left(3 y\Sub{\text{top}}^2 + \frac52 y_\nu^2 - \frac34 g_1^2 - \frac94 g_2^2 \right) y_\nu , \\
\frac{d\lambda}{dt} &= \frac{1}{16\pi^2} \begin{aligned}[t] \bigg( & 4\lambda^2 - (3 {g_1}^2 + 9 {g_2}^2) \lambda + \frac94 ({g_1}^4 + 2 {g_1}^2{g_2}^2 + 3 {g_2}^4) \\&+ 4 (3y\Sub{\text{top}}^2 + {y_\nu}^2) \lambda - 12 (3y\Sub{\text{top}}^4 + {y_\nu}^4) \bigg) \end{aligned} .\notag
\end{align}
Below the threshold $E=m_R$, the Yukawa coupling of the tau-neutrino drops out of the RG equations and is replaced by an effective coupling 
\begin{align*}
\kappa = 2 \frac{{y_\nu}^2}{m_R} ,
\end{align*}
which provides an effective mass $m_l = \frac14 \kappa {v_0}^2$ for the light tau-neutrino. The renormalization group equations of $y\Sub{\text{top}}$ and $\lambda$ for $E<m_R$ are then given by
\begin{align}
\label{eq:RGE_low}
\frac{dy\Sub{\text{top}}}{dt} &= \frac1{16\pi^2} \left(\frac92 y\Sub{\text{top}}^2 - \frac{17}{12}g_1^2 - \frac94 g_2^2 - 8g_3^2\right) y\Sub{\text{top}} ,\notag\\
\frac{d\lambda}{dt} &= \frac{1}{16\pi^2} \bigg( \begin{aligned}[t] &4\lambda^2 - (3 {g_1}^2 + 9 {g_2}^2) \lambda + \frac94 ({g_1}^4 + 2 {g_1}^2{g_2}^2 + 3 {g_2}^4) \\ &+ 12 y\Sub{\text{top}}^2 \lambda - 36 y\Sub{\text{top}}^4 \bigg) . \end{aligned} 
\end{align}
The equation for $y_\nu$ is replaced by an equation for the effective coupling $\kappa$ given by \cite[Eq.~(14)]{AKLR02}
\begin{align}
\label{eq:RGE_kappa}
\frac{d\kappa}{dt} = \frac{1}{16\pi^2} \left( 6 y\Sub{\text{top}}^2 - 3 {g_2}^2 + \frac\lambda6 \right) \kappa .
\end{align}

\subsubsection{Running masses}

The numerical solutions to the coupled differential equations of \eqref{eq:RGE_high,eq:RGE_low,eq:RGE_kappa} for $y\Sub{\text{top}}$, $y_\nu$ and $\lambda$ depend on the choice of three input parameters:
\begin{itemize}
\item the scale $\Lambda\Sub{GUT}$ at which our model is defined;
\item the ratio $\rho$ between the masses $m_\nu$ and $m\Sub{\text{top}}$;
\item the Majorana mass $m_R$ which produces the threshold in the renormalization group flow.
\end{itemize}
The scale $\Lambda\Sub{GUT}$ is taken to be either $\Lambda_{12} = 1.03 \times 10^{13} \;\GeV$ or $\Lambda_{23} = 9.92 \times 10^{16} \;\GeV$, given by \eqref{eq:GUT-scale_12,eq:GUT-scale_23}. We will determine the numerical solution to \eqref{eq:RGE_high,eq:RGE_low,eq:RGE_kappa} for a range of values for $\rho$ and $m_R$. First, we need to start with the initial conditions of the running parameters at the scale $\Lambda\Sub{GUT}$. By inserting the top-quark mass $m\Sub{\text{top}} = \frac12 \sqrt2 y\Sub{\text{top}} v_0$, the tau-neutrino mass $m_\nu = \rho m\Sub{\text{top}}$, and the W-boson mass $M_W = \frac12 g_2 v_0$ into \eqref{eq:top_mass_corr}, we obtain the constraints
\begin{align*}
y\Sub{\text{top}}(\Lambda\Sub{GUT}) &\lesssim \frac{2}{\sqrt{3+\rho^2}} g_2(\Lambda\Sub{GUT}) , & y_\nu(\Lambda\Sub{GUT}) &\lesssim \frac{2\rho}{\sqrt{3+\rho^2}} g_2(\Lambda\Sub{GUT}) ,
\end{align*}
where from \eqref{eq:running_couplings} we have the values $g_2(\Lambda_{12}) = 0.5444$ and $g_2(\Lambda_{23}) = 0.5170$. 

From \eqref{eq:Higgs_quartic_coupling} we have an expression for the quartic coupling $\lambda$ at $\Lambda\Sub{GUT}$. Approximating the coefficients $a$ and $b$ from \eqref{eq:abcde_SM} by $a \approx (3+\rho^2) m\Sub{\text{top}}^2$ and $b \approx (3+\rho^4) m\Sub{\text{top}}^4$, we obtain the boundary condition
\begin{align*}
\lambda(\Lambda\Sub{GUT}) \approx 24 \frac{3+\rho^4}{(3+\rho^2)^2} g_2(\Lambda\Sub{GUT})^2 .
\end{align*}

Using these boundary conditions, we can now numerically solve the RG equations of \eqref{eq:RGE_high} from $\Lambda\Sub{GUT}$ down to $m_R$, which provides us with the values for $y\Sub{\text{top}}(m_R)$, $y_\nu(m_R)$ and $\lambda(m_R)$. At this point, the Yukawa coupling $y_\nu$ is replaced by the effective coupling $\kappa$ with the boundary condition
\begin{align*}
\kappa(m_R) = 2 \frac{y_\nu(m_R)^2}{m_R} .
\end{align*}
Next, we numerically solve the RG equations of \eqref{eq:RGE_low,eq:RGE_kappa} down to $M_Z$ to obtain the values for $y\Sub{\text{top}}$, $\kappa$ and $\lambda$ at ordinary energy scales. 

The running mass of the top quark at these ordinary energies is then given by \eqref{eq:Yukawa_mass}. We find the running Higgs mass by inserting $\lambda$ into \eqref{eq:Higgs_mass}. We shall evaluate these running masses at their own energy scale. For instance, our predicted mass for the Higgs boson is the solution for $\mu$ of the equation $\mu = \sqrt{\lambda(\mu)/3}v_0$. In this equation we shall ignore the running of the vacuum expectation value $v_0$. 

\begin{table}[t]
\begin{align*}
&\begin{array}{l|ccccccc}
\Lambda\Sub{GUT} \; (10^{13}\;\GeV) & 1.03 & 1.03 & 1.03 & 1.03 & 1.03 & 1.03 & 1.03 \\
\rho & 0 & 0.90 & 0.90 & 1.00 & 1.00 & 1.10 & 1.10 \\
m_R \; (10^{13}\;\GeV) & - & 0.25 & 1.03 & 0.30 & 1.03 & 0.35 & 1.03 \\
\hline
m\Sub{\text{top}} \; (\GeV) & 183.2 & 173.9 & 174.1 & 171.9 & 172.1 & 169.9 & 170.1 \\
m_l \; (\eV) & 0 & 2.084 & 0.5037 & 2.076 & 0.6030 & 2.080 & 0.7058 \\
m_h \; (\GeV) & 188.3 & 175.5 & 175.7 & 173.4 & 173.7 & 171.5 & 171.8 \\
\end{array}\\
&\mbox{}\\
&\begin{array}{l|ccccccccc}
\Lambda\Sub{GUT} \; (10^{16}\;\GeV) & 9.92 & 9.92 & 9.92 & 9.92 & 9.92 \\
\rho & 0 & 1.10 & 1.10 & 1.20 & 1.20 \\
m_R \; (10^{13}\;\GeV) & - & 0.30 & 2.0 & 0.35 & 9900 \\
\hline
m\Sub{\text{top}} \; (\GeV) & 186.0 & 173.9 & 174.2 & 171.9 & 173.5 \\
m_l \; (\eV) & 0 & 1.939 & 0.2917 & 1.897 & 6.889\times10^{-5} \\
m_h \; (\GeV) & 188.1 & 171.3 & 171.6 & 169.1 & 171.2 \\
\end{array}\\
&\mbox{}\\
&\begin{array}{l|cccc}
\Lambda\Sub{GUT} \; (10^{16}\;\GeV) & 9.92 & 9.92 & 9.92 & 9.92 \\
\rho & 1.30 & 1.30 & 1.35 & 1.35 \\
m_R \; (10^{13}\;\GeV) & 0.40 & 9900 & 100 & 9900 \\
\hline
m\Sub{\text{top}} \; (\GeV) & 169.9 & 171.6 & 169.8 & 170.6 \\
m_l \; (\eV) & 1.866 & 7.818\times10^{-5} & 8.056\times10^{-3} & 8.286\times10^{-5} \\
m_h \; (\GeV) & 167.1 & 169.3 & 167.4 & 168.4 \\
\end{array}
\end{align*}
\caption{Numerical results for the masses $m\Sub{\text{top}}$ of the top-quark, $m_l$ of the light tau-neutrino, and $m_h$ of the Higgs boson, as a function of $\Lambda\Sub{GUT}$, $\rho$, and $m_R$.
\label{tab:num_results_23}}
\end{table}

The effective mass of the light neutrino is determined by the effective coupling $\kappa$, and we choose to evaluate this mass at the scale $M_Z$. Thus, we calculate the masses by
\begin{align*}
m\Sub{\text{top}}(m\Sub{\text{top}}) &= \frac12 \sqrt2 y\Sub{\text{top}}(m\Sub{\text{top}}) v_0 , 
& m_l(M_Z) &= \frac14 \kappa(M_Z) {v_0}^2 ,
& m_h(m_h) &= \sqrt{\frac{\lambda(m_h)}{3}}v_0 ,
\end{align*}
where, from the W-boson mass \cite{PDG10} $M_W = 80.399 \pm 0.023 \;\GeV$, we insert the value $v_0 = 246.66 \pm 0.15$. The results of this procedure for $m\Sub{\text{top}}$, $m_l$ and $m_h$ are given in \cref{tab:num_results_23}. In this table, we have chosen the range of values for $\rho$ and $m_R$ such that the mass of the top-quark and the light tau-neutrino are in agreement with the experimental values \cite{PDG10} 
\begin{align*}
m\Sub{\text{top}} &= 172.0 \pm 0.9 \pm 1.3 \;\GeV , & m_l \leq 2 \;\eV . 
\end{align*}
For comparison, we have also included the simple case where we ignore the Yukawa coupling of the tau-neutrino (by setting $\rho=0$), in which case there is also no threshold at the Majorana mass scale. As an example, we have plotted the running of $\lambda$, $y\Sub{\text{top}}$, $y_\nu$ and $\kappa$ for the values of $\Lambda\Sub{GUT} = \Lambda_{23} = 9.92\times10^{16}\;\GeV$, $\rho = 1.2$, and $m_R = 3\times10^{12}\;\GeV$ in \cref{fig:lambda_with_neutrinos,fig:top-yukawa_with_neutrinos,fig:tau-neutrino-yukawa_with_neutrinos,fig:kappa}. 

For the allowed range of values for $\rho$ and $m_R$ which yield plausible results for $m\Sub{\text{top}}$ and $m_l$, the mass $m_h$ of the Higgs boson takes its value within the range
\begin{align*}
167 \;\GeV \leq m_h \leq 176 \;\GeV. 
\end{align*}
The errors in this prediction produced by the initial conditions (other than $m\Sub{\text{top}}$ and $m_l$) taken from experiment, and by ignoring higher-loop corrections to the RGEs, are smaller than this range of possible values for the Higgs mass, and therefore we shall ignore these errors. 
However, recent results \cite{ATLAS12,CMS12} from the ATLAS and CMS experiments at the Large Hadron Collider at CERN have already excluded our predicted range at the 99\% confidence level. In fact, the latest preliminary results, announced by CERN on 4 July 2012, show the discovery of a new boson in the mass region around $125-126 \;\GeV$. 

It might be that our big desert hypothesis all the way up to GUT-scale is wrong. In fact, the mismatch of the three lines in Figure \ref{figure:running-couplings} indicates that this hypothesis can not be correct. Improvements for this problem in the context of noncommutative geometry have been proposed in \cite{CIS99}. Also, it is interesting to see what supersymmetry has to say in the context of noncommutative geometry, since in, for instance, the minimally supersymmetric Standard Model the three lines supposedly do meet. The recent \cite{BroS10,BroS11} appear to be promising in this respect.

\begin{figure}[p]
\includegraphics[height=0.42\textheight]{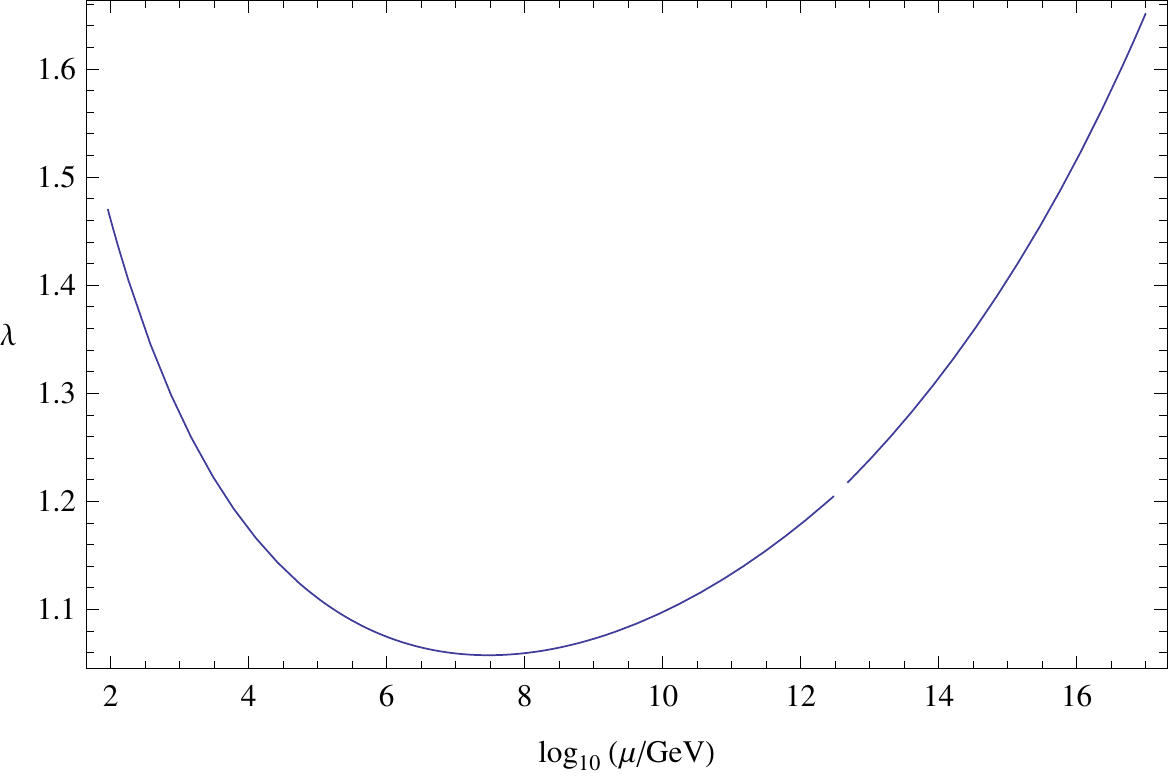}
\caption{The running of the quartic Higgs coupling $\lambda$ for $\Lambda\Sub{GUT} = 9.92\times10^{16}\;\GeV$, $\rho = 1.2$, and $m_R = 3\times10^{12}\;\GeV$.\label{fig:lambda_with_neutrinos}}
\end{figure}

\begin{figure}[p]
\includegraphics[height=0.42\textheight]{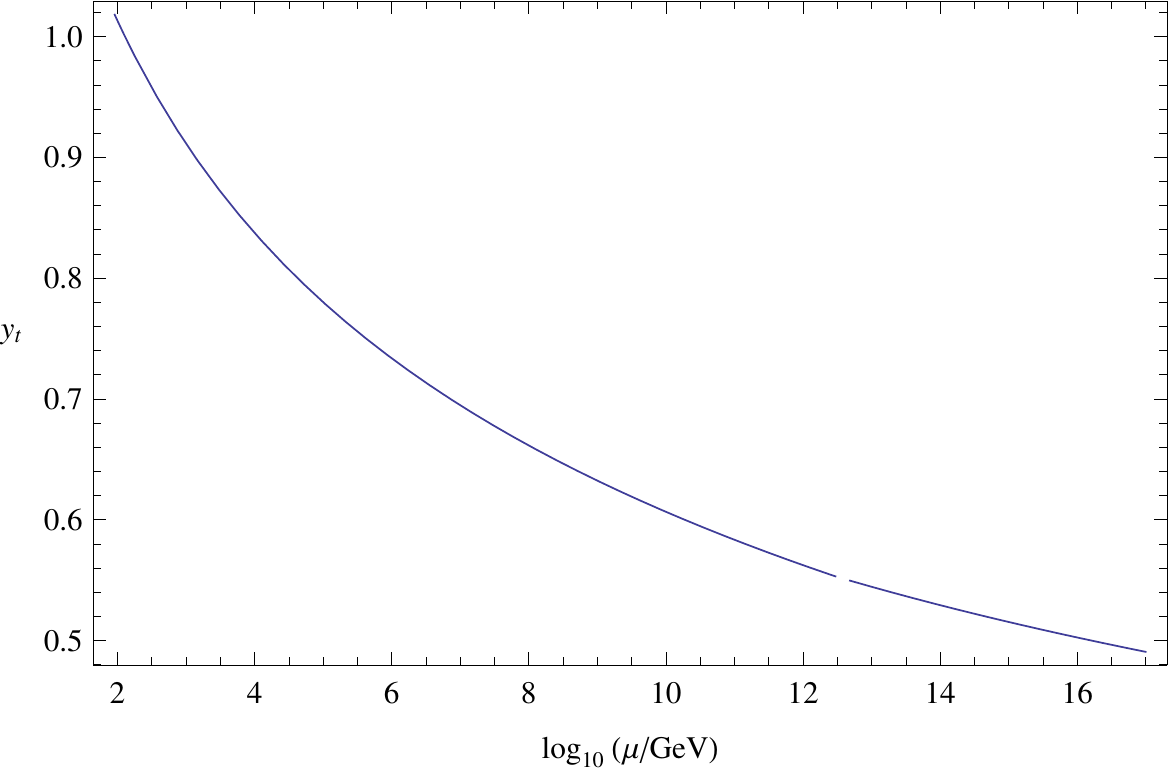}
\caption{The running of the top-quark Yukawa coupling $y\Sub{\text{top}}$ for $\Lambda\Sub{GUT} = 9.92\times10^{16}\;\GeV$, $\rho = 1.2$, and $m_R = 3\times10^{12}\;\GeV$.\label{fig:top-yukawa_with_neutrinos}}
\end{figure}

\begin{figure}[p]
\includegraphics[height=0.42\textheight]{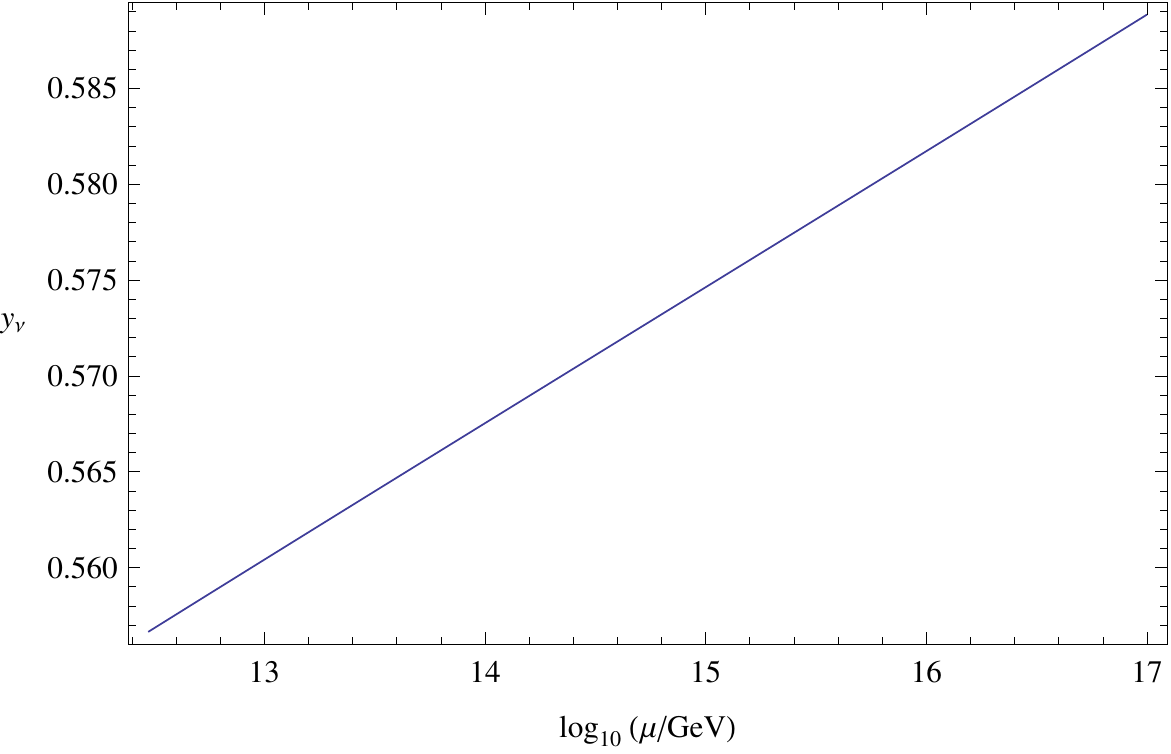}
\caption{The running of the tau-neutrino Yukawa coupling $y_\nu$ for $\Lambda\Sub{GUT} = 9.92\times10^{16}\;\GeV$, $\rho = 1.2$, and $m_R = 3\times10^{12}\;\GeV$.\label{fig:tau-neutrino-yukawa_with_neutrinos}}
\end{figure}

\begin{figure}[p]
\includegraphics[height=0.42\textheight]{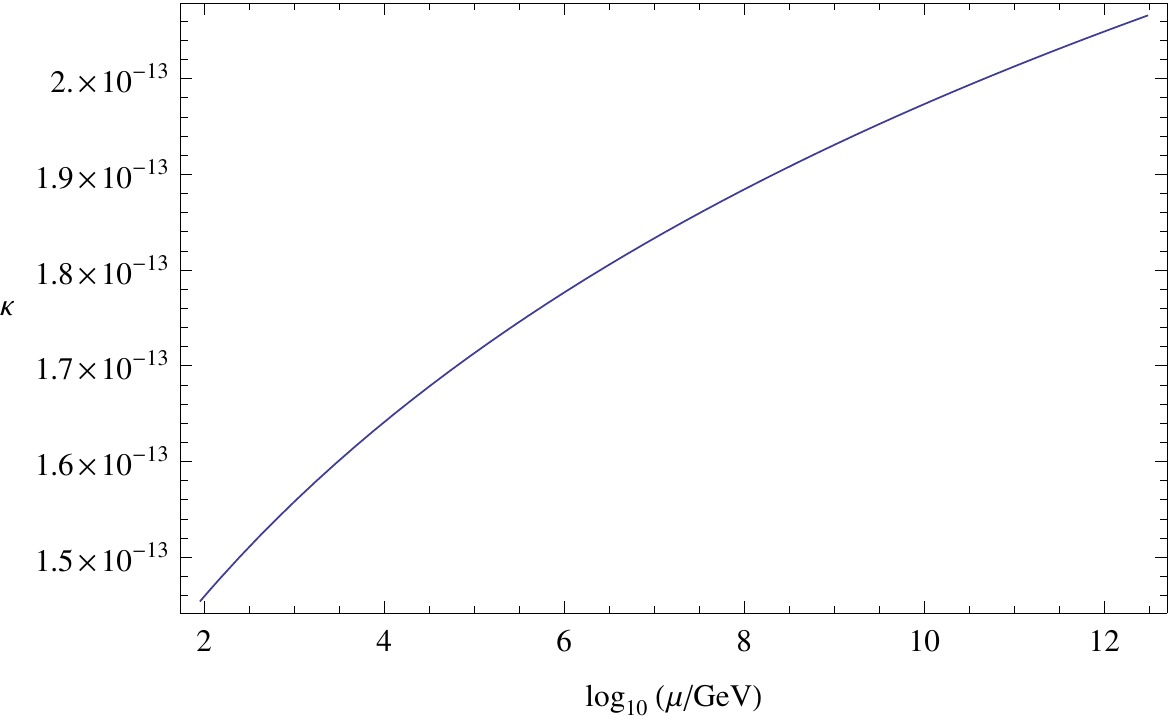}
\caption{The running of the effective coupling $\kappa$ for $\Lambda\Sub{GUT} = 9.92\times10^{16}\;\GeV$, $\rho = 1.2$, and $m_R = 3\times10^{12}\;\GeV$.\label{fig:kappa}}
\end{figure}

\newpage
\section{Outlook}
\label{chap:outlook}

In this review article, we have presented the ideas from noncommutative geometry that allows for a geometrical description of the Standard Model of high-energy physics, in addition allowing for predictions. As a geometrical theory, it unifies gravity and the Standard Model, albeit on a classical level. At present, this geometrical description of gravity and the Standard Model faces two main challenges. First, as mentioned in \cref{remark:Lorentz}, the almost-commutative spacetimes discussed in this article are all based on Riemannian spacetimes. In order to properly describe physical theories, a generalization of the noncommutative geometry framework to pseudo-Riemannian (e.g.\ Lorentzian or Minkowskian) spacetimes is required. Some progress in this direction has been obtained in \cite{Haw96,Str01,Mor02b,Sui04,PV04,PS06,Franco10,vdDPR12} though this program is far from being completed.

Second, noncommutative geometry only describes physics at the classical level. Let us now comment on some recent developments on the problem of quantization in the context of noncommutative geometry.

Ever since Heisenberg, it was believed that uncertainty relations between spacetime coordinates might improve the short-distance singularities appearing in a quantum theory of fields. Eventually, this led Snyder \cite{Sny47} to the study of the Moyal-type relations already presented in the Introduction. The corresponding uncertainty relations were also found in \cite{DFR95} when combining the principles of quantum theory with those of general relativity. 

Although this noncommutativity indeed improves its UV-behaviour, it turned out \cite{MRS00} that scalar quantum field theory on a Moyal plane has bad behaviour at the IR-side: the notorious UV/IR-mixing. Because of this, it was very surprising that \cite{GW05} came up with a scalar field theory that was renormalizable. This theory is presently {\it under construction} (in the sense of constructive quantum field theory) \cite{GW09, Riv11, GW12}.

As far as the spectral action approach is concerned, there are some recent results by one of the authors on renormalization of the asymptotically expanded spectral action as a higher-derivative theory \cite{Sui11b,Sui11c,Sui11d,Sui12}. As usual for higher-derivative gauge theories (cf. \cite[Section 4.4]{FS80}), this renders the asymptotically expanded Yang--Mills spectral action (on a flat background spacetime) superrenormalizable with counterterms proportional to the Yang--Mills action, indeed appearing at lowest order in the spectral action. 

An interesting open problem is whether the above results could lead to a more intrinsic understanding of quantization, that is, in the context of noncommutative geometry. Inevitably, such a quantum noncommutative geometry has to combine ideas from (loop) quantum gravity with those from quantum gauge theories, and such an analysis has started in \cite{AG06,AG07,AGN08,MZ08}. 
Within the noncommutative geometry setup, one would then arrive at a unified and geometrical formulation of quantum gravity with matter.

\newcommand{\noopsort}[1]{}\def\cprime{$'$}
\providecommand{\bysame}{\leavevmode\hbox to3em{\hrulefill}\thinspace}
\providecommand{\MR}[1]{}

\end{document}